\documentclass[10.5pt,onecolumn,draftclsnofoot]{IEEEtran}
\IEEEoverridecommandlockouts
\usepackage{bm}
\usepackage{amsmath, amssymb, cite}
\usepackage{amsthm,bm,bbm}
\usepackage{mathtools}
\usepackage[size=small]{caption}

\usepackage{amsthm,  multirow, color, amsfonts}
\usepackage{tabulary}
\usepackage{subcaption}
\usepackage{graphicx}
\usepackage{setspace}
\usepackage{enumerate}

\usepackage[ruled]{algorithm2e}
\usepackage{comment}
\usepackage{pdfpages}
\usepackage{hyperref}
\usepackage{stfloats}
\usepackage{lscape, hhline}

\usepackage{etoolbox}
\makeatletter
\patchcmd{\@makecaption}
  {\scshape}
  {}
  {}
  {}
\makeatother

\makeatletter
\def\BState{\State\hskip-\ALG@thistlm}
\makeatother

\usepackage[utf8]{inputenc} 
\usepackage[T1]{fontenc}
\usepackage{url}
\usepackage{ifthen}

\usepackage[dvipsnames]{xcolor}

\definecolor{pptgreen}{RGB}{0,176,80}

\DeclareMathOperator{\minrank}{minrk}
\DeclareMathOperator{\rk}{{\rm rk}}
              
\newtheorem{theo}{Theorem}
\newtheorem{defi}{Definition}

\newtheorem{prop}{Proposition}
\newtheorem{remark}{Remark}

\newtheorem{cor}{Corollary}

\newtheorem{lem}{Lemma}
\newtheorem{ex}{Example}

\usepackage{mathrsfs}

\newcommand{\deryar}[1]{\textcolor{red}{ #1}}

\newcommand{\GU}{G_{\mathrm B}}

\newcommand{\GUn}{G^{\dot{\vee} n}_{\mathrm B}} 
\newcommand{\GUt}{G^{\dot{\vee} t}_{\mathrm B}}

\def\compactify{\itemsep=0pt \topsep=0pt \partopsep=0pt \parsep=0pt}
\let\latexusecounter=\usecounter

\ifx\du\undefined
  \newlength{\du}
\fi
\usepackage{siunitx}
\usepackage{tikz}
\usetikzlibrary{shapes.geometric}
\usepackage{verbatim}
\usepackage{color,soul}
\usetikzlibrary{spy,calc}

\usepackage{breqn}

\hypersetup{breaklinks=true}
\usepackage{diagbox}

\tikzstyle{vertex}=[circle,minimum size=1pt,inner sep=0pt]
\tikzstyle{selected vertex} = [vertex, fill=red!24]
\tikzstyle{edge} = [draw,thick,-]
\tikzstyle{weight} = [font=\small]
\tikzstyle{selected edge} = [draw,line width=5pt,-,red!50]
\tikzstyle{ignored edge} = [draw,line width=5pt,-,black!20]

\usepackage{ifthen}
\usetikzlibrary{shapes.geometric}

\tikzset{
    edge/.style={draw},
    weight/.style={font=\small}
}

\tikzset{
    rednode/.style={draw, rectangle, inner sep=2pt, fill=red},
    greennode/.style={draw, regular polygon, regular polygon sides=3, inner sep=2pt, fill=green},
    bluenode/.style={draw, circle, inner sep=2pt, fill=blue},
    edge/.style={draw},
    weight/.style={font=\small}
}

\begin{document}
\title{Graph-Theoretic Bounds for Non-Linear Function Computation Broadcast 
\thanks{This research was partially supported by European Research Council ERC-StG Project SENSIBILITÉ (101077361), the Huawei France-Funded Chair Toward Future Wireless Networks, and the French National Research Agency (ANR-22-PEFT-0010) as part of France 2030.}}

\author{Derya Malak, Vijith Kumar K. P., and M. Reza Deylam Salehi\thanks{D. Malak is with the Communication Systems Department, EURECOM, Sophia Antipolis, FRANCE. This work was conducted in part while V. Kumar K. P. and M. R. Deylam Salehi were with EURECOM. 
The order of authors does not represent their relative contributions to this work.\hfill Corresponding author: D. Malak (malak@eurecom.fr)} \thanks{An early version of the paper appeared at the 2025 Int. Symp. Inf. Theory~\cite{DeyKizMal2025nonlinear}. \hfill Manuscript last revised: {\today}.}}

\maketitle
\vspace{-2.5em}
\begin{abstract}
This work studies non-linear function computation broadcast (NFCB), in which a sender with access to $N$ datasets $(X_1,\dots,X_N)$ broadcasts a common message to $K$ users, each possessing side information and requesting a function of the datasets. The goal is to minimize the rate required for asymptotically lossless recovery of all demands. We introduce a broadcast graph that jointly captures the source distribution, side information, and demanded functions. Using K\"orner’s characteristic-graph framework, we develop an achievable scheme for arbitrary $K$, general source distributions, and general finite-field demands, including linear, non-separable, and non-linear functions, without restricting the encoding or decoding operations to be linear. We also present a side-information-assisted independent-set scheme and characterize the optimal graph-based achievable rate for compatible functions. For the converse, we derive a multi-letter response-profile bound that strengthens a basic side-information converse, zero-error and asymptotically lossless clique-entropy bounds based on the operational block broadcast graph, and a genie-aided lower bound. For binary NFCB with $N=K$ and side information $X_i$ at user $i$, we characterize the optimal rate in several special cases and bound the worst-case and average additive gaps between the proposed achievable rate and the genie-aided converse for $K=3$ and $K=4$. Finally, three-user examples with Boolean and linear demands illustrate the proposed bounds and show that non-linear encoding can strictly outperform the best scalar and vector linear schemes.
\end{abstract}

\begin{IEEEkeywords}
Non-linear function computation broadcast (NFCB), noiseless broadcast channel, distributed functional compression, characteristic graphs, union graphs, distributed non-linear function computation, and minrank.
\end{IEEEkeywords}

\section{Introduction}
\label{sec:intro}

Large-scale data processing increasingly relies on communication networks in distributed computation settings, including distributed learning~\cite{yang2023federated}, over-the-air computing~\cite{csahin2023survey}, private-key encryption~\cite{kate2010distributed}, and coded distributed computing~\cite{li2016coded,li2017fundamental, MaddahAli2013,tandon2017gradient,dutta2019optimal,khalesi2022multi,yu2020straggler, namboodiri2026fundamental}. Frameworks such as MapReduce~\cite{dean2008mapreduce} and Spark~\cite{zaharia2010spark} support large-scale parallel processing across distributed servers, but limited computation and communication resources, privacy constraints, and bounded user memory create fundamental bottlenecks in the design of distributed computing schemes.
Beyond classical parallelization, coded distributed computing (CDC) showed that coding across intermediate computation results can substantially reduce communication load and mitigate straggler and bandwidth bottlenecks~\cite{li2017fundamental, MaddahAli2013}. This idea has been applied to distributed machine learning~\cite{lee2017speeding} and coded matrix multiplication~\cite{jia2021capacity}. Several works (see, e.g.,~\cite{tandon2017gradient,dutta2019optimal,khalesi2022multi,yu2020straggler, namboodiri2026fundamental}) optimized the computation-communication tradeoff for linearly separable computations, but extending these approaches to non-linear tasks remains challenging.

Broadcast transmission also underlies many applications, including coded caching~\cite{MaddahAli2013, NadMadAve2016,bidokhti2018noisy}, coded distributed computing~\cite{elkordy2021compressed}, index coding~\cite{bar2011index, LeeDimaindex,tahmasbi2014critical}, federated learning~\cite{tian2022fedbroadcast}, and hypothesis testing~\cite{salehkalaibar2021distributed}. In such problems, a sender with access to all datasets transmits a common broadcast message to multiple users, each with side information and a demand given by a multivariate, possibly nonlinear, function of the datasets. The side information available at each user and the dependencies among data and functions determine the communication cost. Related applications arise in multiple description coding~\cite{gamal1982achievable}, video broadcasting~\cite{reimers1998digital}, cryptography~\cite{cryptoeprint:2024/154}, coordination~\cite{montalban2021broadcast}, and synchronization~\cite{halpern1991clock}. These connections motivate the study of NFCB and its communication complexity.

\subsection{Function Computation Broadcast and Known Instances}
\label{sec:known_instances}

We review several known instances of function computation broadcast: index coding, linear computation broadcast, and complementary delivery.

\subsubsection{Index Coding}
\label{sec:index_coding}

Index coding problem involves a sender with $N$ independent datasets ${\bf X}_{[N]}$ and $K=N$ users, where user $i$ demands $f_i=X_i$ and has side information $\mathcal{S}_i\subseteq \{X_j\}_{j\in[N]}$. The sender broadcasts a common message $M$ so that each user can recover its demand. A standard representation is the {\emph{side information graph}}.

\begin{defi}[{\bf Side information graph}~\cite{bar2011index}]
\label{def:side_info_graph}
Given an index coding problem with $N$ users and datasets ${\bf X}_{[N]}$, the side information graph is a directed graph $G=(\mathcal{V}(G),\mathcal{E}(G))$, where $\mathcal{V}=[N]$ and $\mathcal{E}=\{ i\to j\}$, with $i\to j$ indicating that user $i$ knows dataset $X_j$.
\end{defi}

The optimal length of a binary linear index code equals the minimum rank of a matrix fitting the side information graph $G$, denoted $\minrank_2(G)$, where $\rk_2(\cdot)$ denotes the rank over $\mathbb{F}_2$~\cite{bar2011index}.

\begin{defi}
\label{def:minrank}
Given a side information graph $G$, the minrank over $\mathbb{F}_q$ is~\cite{alon2020minrank}
\begin{align}
\label{eq:minrank}
\minrank_q(G)= \min\{\rk({\bf A}) \ : \ {\bf A}=(a_{ij})\in\mathbb{F}_q^{N\times N}\,\,\mbox{fits}\,\, G \} 
\end{align}
where ${\bf A}$ fits $G$ if $a_{ii}\neq 0$ for all $i$, and $a_{ij}=0$ whenever edge $i \to j$ does not exist in $G$.
\end{defi}

The graph-theoretic quantity $\minrank_q(G)$ characterizes the optimal scalar linear index coding rate over $\mathbb{F}_q$~\cite{bar2011index}, while also providing a lower bound on the minimum length of non-linear index codes for arbitrary graphs~\cite{bar2011index}. 
In particular, 
there exists an optimal scalar linear code of length $\minrank_2(G)$~\cite[Theorem~1]{bar2011index}. The following converse also holds for general index codes~\cite[Theorem~3]{bar2011index}.

\begin{theo}
\label{theo:MAIS}
Given a side information graph $G=(\mathcal{V}(G),\mathcal{E}(G))$, let $G|_{\mathcal J}$ be the subgraph induced by $\mathcal J\subseteq \mathcal{V}(G)$. The size of the maximum acyclic induced subgraph is
\begin{align}
\label{eq:MAIS_bound}
{\mathrm{MAIS}}(G)=\max\limits_{\mathcal J\subseteq \mathcal{V}\ : \ G|_{\mathcal J}\ \mbox{is acyclic}}\ |\mathcal J| .
\end{align}
The length of any $\delta$-error randomized index code for $G$ is lower bounded by~\cite[Theorem~3]{bar2011index}
\begin{align}
R_{\mathrm{MAIS}}={\mathrm{MAIS}}(G)\cdot (1-h(\delta)) \ .
\end{align}
\end{theo}

The index coding problem is closely related to network coding and matroid theory~\cite{arbabjolfaei2018fundamentals,el2010index}. Extensions include 
multiple senders jointly serving user demands~\cite{khalesi2025hyper,li2018multi}, coded cooperative data exchange~\cite{el2010coding}, where all users recover the full dataset, and functional index coding, where user demands are functions of the datasets~\cite{gupta2016relation}.

\subsubsection{Linear Computation Broadcast}
\label{sec:MinRank_ach}

{\emph{Linear computation broadcast}} generalizes index coding by allowing users to possess and request linear functions of the dataset~\cite {sun2019capacity,yao2022capacity,yao2024capacity,yao2024generic,MaTunin2025}. In a system with $N$ datasets and $K$ users, user $i$ requests $f_i={\bf X}_{[N]}^{\intercal}{\bf V}_i$ 
and the side information 
may be an arbitrary linear function of the datasets, namely $\mathcal{S}_i={\bf X}_{[N]}^{\intercal}{\bf V}'_i$, given a linear encoding matrix ${\bf V}'_i$, where ${\bf V}_i,{\bf V}'_i\in\mathbb{F}_q^{N\times 1}$. 
The optimal broadcast message can be obtained by solving a generalized index coding problem. In particular, \cite{LeeDimaindex} reformulated this task as a minrank problem over $\mathbb{F}_q$, where the fitting matrix captures the subspace relations induced by users’ demands and side information. This optimization is NP-hard.

The optimal broadcast rate has been characterized for two-user systems when the server stores either the full dataset~\cite{sun2019capacity} or an arbitrary linear function of it~\cite{yao2022capacity}. For the three-user case, the optimal rate under vector linear coding was derived using linear subspace decompositions~\cite{yao2024capacity}. This approach was extended to arbitrary $K$, yielding an achievable scheme~\cite{MaTunin2025}. For the symmetric $K$-user setting, the optimal rate is known for $K>d$, and the generic capacity for arbitrary $K$ is characterized within a factor of $2$~\cite{yao2024generic}.

For coded side information and all $K\geq 2$, adapting~\cite[Theorem~1]{LeeDimaindex}, the optimal linear vector index code length over $\mathbb{F}_q$ is given by~\cite{LeeDimaindex}
\begin{align}
\label{eq-min-rank-linear-demand} 
\minrank_q(\{{\bf V}_i, {\bf V}'_i\}_{i\in[K]})
\triangleq\min_{\lambda_1, \dots, \lambda_K}\left\{ \rk({\bf A}) \ : \ {\bf A}(i,:)\triangleq 
 {\bf V}_i^\intercal \oplus_q \lambda_i {\bf V}^{\prime\intercal}_i \ , \ \lambda_i\in\mathbb{F}_q\ ,\,\forall\, i\in[K] \right\} ,
 \end{align}
which generalizes matrix minrank to coded side information~\cite{LeeDimaindex}, and remains open in general. In the special binary case where ${\bf V}_i$ and ${\bf V}'_i$ encode standard demands and uncoded side information,~\eqref{eq-min-rank-linear-demand} reduces to~\eqref{eq:minrank}.

\subsubsection{Complementary Delivery}
\label{sec:comp_del}

The complementary delivery problem is a special case of NFCB with $f_1=X_2$ and $f_2=X_1$. 
In the more general two-user non-linear setting in which each decoder seeks a function of $(X_1,X_2)$ and has complementary side information,~\cite{ravi2016broadcast} characterizes the broadcast rate using the $f_1f_2$-modified rook’s graph.


\begin{defi}[Modified rook's graph~\cite{ravi2016broadcast}]
\label{def:modified_rook_graph}
$f_1f_2$-modified rook's graph $G^{f_1f_2}_{X_1X_2}=(\mathcal{V},\mathcal{E})$ is a subgraph of the rook's graph on $\mathcal{X}_1\times\mathcal{X}_2$, defined as follows: $\mathcal{V}$ is the vertex set
 \begin{align}
 \label{eq:vertex_set_rook}
\mathcal{V}=\{((x_1^{(k)},x_2^{(k)})\in\mathcal{X}_1\times\mathcal{X}_2\} \ ,
 \end{align}
and $\mathcal{E}$ is the edge set
 \begin{align}
 \label{eq:edge_set_rook}
\mathcal{E}=\{\{(x_1^{(1)},x_2^{(1)}),(x_1^{(2)},x_2^{(2)})\}\ : \ &(x_1^{(1)},x_2^{(1)}),(x_1^{(2)},x_2^{(2)})\in\mathcal{V},\ (x_1^{(1)},x_2^{(1)})\neq (x_1^{(2)},x_2^{(2)})\ , \nonumber\\
&x_1^{(1)}=x_1^{(2)}\quad \mbox{and}\quad f_1(x_1^{(1)},x_2^{(1)})\neq f_1(x_1^{(2)},x_2^{(2)}) \ ,\nonumber\\
&x_2^{(1)}=x_2^{(2)}\quad \mbox{and}\quad f_2(x_1^{(1)},x_2^{(1)})\neq f_2(x_1^{(2)},x_2^{(2)})\} \ .
 \end{align}
\end{defi}

Under an $\epsilon$-error criterion, for a $(2,2,\{\mathcal{S}_i\}_{i\in[2]},\{f_i\}_{i\in[2]})$ NFCB problem with complementary side information, i.e., $\mathcal{S}_1=X_1$ and $\mathcal{S}_2=X_2$,~\cite[Theorem~3]{ravi2016broadcast} gave a single-letter characterization of the optimal broadcast rate when either (i) $f_1(X_1,X_2)$ and $f_2(X_1,X_2)$ belong to a compatible class, including the case $f_1=f_2$, or (ii) $X_1,X_2\in\mathbb F_2$. 

\begin{defi}[Compatible functions~\cite{ravi2016broadcast}]
\label{def:compatible_functions}
Functions $f_1(X_1, X_2)$ and $f_2(X_1, X_2)$ are compatible if there exists a function $f(X_1, X_2)$ such that $G^{ff}_{X_1X_2} = G^{f_1f_2}_{X_1X_2}$.
\end{defi}

The NFCB framework generalizes index coding~\cite{BirKol2006,el2010index,BarBirJayKol2011,arbabjolfaei2018fundamentals} and remains open in general. Prior work has mainly focused on linear computation broadcast~\cite{wan2021optimal,sun2019capacity, yao2022capacity,yao2024capacity, yao2024generic} or non-linear settings with complementary demands~\cite{ravi2016broadcast}. These results typically rely on strong structural assumptions and often hold only for specific numbers of users. The framework also contains multiple description coding~\cite{gamal1982achievable} as a special case. To the best of our knowledge, characterizing the optimal broadcast rate for general $(K,N,\{\mathcal{S}_i\}_{i\in[K]},\{f_i\}_{i\in[K]})$ NFCB is at least as hard as solving linear computation broadcast, and hence at least as hard as the general index coding problem. Our goal is thus to derive general upper and lower bounds on the broadcast rate.

\subsection{Contributions}
\label{sec:contributions} 

Motivated by the above, this work studies non-linear function computation broadcast (NFCB) with general function demands, including linear, non-separable, and non-linear functions over finite fields. A sender with access to $N$ datasets transmits a common broadcast message $M$ to $K$ users. Each user has side information and seeks to compute a possibly non-linear function of the datasets using the broadcast message and its side information. The communication rate is referred to as the \emph{computation broadcast rate}.

The technical contributions of this work are as follows.
\begin{itemize}

\item On the achievability side, extending~\cite{ravi2016broadcast,sun2019capacity, yao2022capacity,yao2024capacity, yao2024generic} to arbitrary $K$, we develop a \emph{characteristic graph-based encoding scheme}—originally due to K\"orner~\cite{korner1971coding}—for NFCB with general source distributions and non-linear demands over finite fields. 
In Theorem~\ref{theo-general-achievable-rate-via-broadcast}, the sender exploits the structure of the $N$ datasets and $K$ functions to construct a broadcast message enabling asymptotically lossless decoding at all users, without requiring linearity or separability, or linear encoding/decoding. 
Corollary~\ref{cor:side-information-assisted-graph-bound} provides a side-information-assisted independent-set coding scheme.  
Proposition~\ref{prop_achievable-rate-via-broadcast_compatible} characterizes the optimal achievable rate for compatible functions.

\item We derive novel graph-theoretic converse bounds on the broadcast rate for $(K\geq 2, N)$ NFCB. In particular, we present a multi-letter response-profile converse bound (Theorem~\ref{theo-lower-bound-mutual-info}) on the broadcast rate for the $(K, N,\{\mathcal{S}_i\}_{i\in[K]},\{f_i\}_{i\in[K]})$ NFCB scheme. By carefully capturing each user’s demand and side information, Theorem~\ref{theo-lower-bound-mutual-info} gives tighter converses than a simple side-information-based converse. Compared with the achievable rate in Theorem~\ref{theo-general-achievable-rate-via-broadcast}, which optimizes the coloring distribution over MISs of $\GU$, the converse bound in Theorem~\ref{theo-lower-bound-mutual-info} relies on a carefully designed variable to approximate the fundamental rate limits induced by $({\bf X}_{[N]}, \{\mathcal{S}_i\}_{i\in[K]},\{f_i\}_{i\in[K]})$. 
We will subsequently detail two characteristic graph-based converses, namely the zero-error clique-entropy converse  (Theorem~\ref{thm:zero_error_clique_converse}), and the vanishing-error clique-entropy converse (Theorem~\ref{thm:vanishing_error_clique_converse}) using the length-$n$ operational block broadcast graph $G_{\mathrm B}^{(n)}$, and relying on cliques that capture the demand and side information of each user. 
We will then present an additional genie-aided lower bound (Proposition~\ref{prop-genie}). 

\item We further establish a genie-aided lower bound  (Proposition~\ref{prop-genie}
) that extend techniques from the $(2,N)$ setting in~\cite{jafar2006capacity} and~\cite{sun2018capacity} to the $(K>2,N)$ case. We then upper bound the average additive gap between 
Corollary~\ref{cor:side-information-assisted-graph-bound} 
and Proposition~\ref{prop-genie}. For general NFCB with $q=2$, covering both linear and non-linear demands, our converses are shown to be tight for several special topologies (see Propositions~\ref{prop-2-user-rate_optimality},~\ref{prop-3-user-gap-genie}, and~\ref{prop-4-user-gap-genie}).

\item Examples~\ref{ex-nonlinear-comp},~\ref{ex:nonlinear-outperforms-linear}, and~\ref{ex-nonlinear-comp_self-complementary_graph} examine Boolean (non-linear) function demands, while Example~\ref{ex-linear-comp-jafar} focuses on linear function demands. 
Example~\ref{ex:nonlinear-outperforms-linear} demonstrates the clear savings of non-linear encoding over best (scalar and vector) linear encodings (cf.~\cite{yao2024capacity, yao2022capacity, MaTunin2025}) through Proposition~\ref{prop:advantage_nonlinear}, and upper bounds the zero-error and asymptotically lossless operational broadcast entropies by the performance of linear encoding via Proposition~\ref{prop:nonlinear-graph-separation}, underscoring the advantage of the proposed graph-based approach.
\end{itemize}

\subsection{Organization}
\label{sec:organization}

Section~\ref{sec:intro} reviewed known instances of computation broadcast and summarized our contributions. Section~\ref{sec:sys-mod} presents the NFCB system model. Section~\ref{sec:achievability} details achievability schemes that exploit the notion of operational block broadcast graph. Section~\ref{sec:converse} develops converse bounds. Section~\ref{sec:optimality_gap} characterizes the average additive gap between our achievability scheme (
Corollary~\ref{cor:side-information-assisted-graph-bound}) and the genie-aided converse  (Proposition~\ref{prop-genie}). Section~\ref{sec:3-user-comp-broad} compares our scheme with the state of the art via examples. Section~\ref{sec:conc} concludes the paper and discusses future directions.

{\bf Notation.} 
We denote random variables and their realizations in regular typeface, and matrices in boldface, with elements chosen from a finite field $\mathbb{F}_q$  of order $q\geq 2$, and the nonzero elements of $\mathbb{F}_q$ are denoted by $\mathbb{F}_q^*$. Let $\log$ and $\log_q$ be the binary logarithm and the logarithm in base $q > 2$, respectively. $\oplus_q$ denotes modulo-$q$ sum operation.
Given a pair of discrete random variables $(X,Y)$ with values over the space $\mathcal{X}\times\mathcal {Y}$, we denote by $P_{X,Y}$ their joint distribution, by $P_X$ and $P_Y$ their marginal distributions, respectively, and by $P_{X\mid Y}$ the conditional distribution of $Y$ given $X$. 
The Shannon entropies of $X$ in binary and q-ary units are given by $H_2(X)$ and $H_q(X) =\mathbb{E}[-\log_q P_X]= (1/ \log q)H_2(X)$, respectively, where $\mathbb{E}[X]$ is the expected value of a random variable $X$. Similarly, $H_q(X, Y)$ and $H_q(X \mid Y)$ denote the joint and conditional entropies, respectively. 
For a probability vector $\mathbf p=(p_1,\ldots,p_m)$, define $\mathsf H_q(\mathbf p)\triangleq-\sum_{j=1}^{m}p_j\log_q p_j$, where $0\log_q 0\triangleq 0$. 
The mutual information is defined as $I_q(X;Y)=D_{\mathrm {KL} }(P_{X,Y}\parallel P_X\otimes P_Y)$, where $D_{\mathrm {KL} }$ is the Kullback–Leibler divergence, and $P_X\otimes P_Y$ is the outer product distribution which assigns probability $P_X(x)\cdot P_Y(y)$ to each $(x,y)$. 
Given $X^n=X_1, X_2 \cdots, X_n$, an independent and identically distributed (i.i.d.) sequence, each element $X_i$, $i\in [n]$ is distributed according to $P_X$.

We use $X \sim \mathrm{Unif}(\mathbb{F}_q)$ to denote a uniform random variable over $\mathbb{F}_q$, and $X \sim {\mathrm{Bern}}(\epsilon)$ to denote a Bernoulli random variable with parameter $\epsilon \in [0, 1]$, with $h(\epsilon)$ denoting the binary entropy function. 
A trit is a ternary digit with three states, i.e., $1 \ \mbox{trit}=\log_2 3 \  \mbox{bits}$. 
We denote the probability of event $A$ by $\mathbb{P}(A)$. If $A$ and $B$ are sets, their union is $A\cup B=\{x:x\in A{\text{  or  }}x\in B\}$, and $A \setminus B = \{x \in A \mid x \notin B \}$.  
We let $[n]=\{1,\dots,n\}$ for $n\in\mathbb{Z}^+$, and $[a,b]=\{a, a+1,\cdots, b\}$ for $a,b\in\mathbb{Z}^+$. $\mathrm{Perm}([n])$ denotes the set of all permutations of $[n]$.  
Given a matrix ${\bf A}=(a_{ij})_{i\in [m],\ j\in [n]}\in\mathbb{F}_q^{m\times n}$, its $i$-th row and $j$-th column are given by ${\bf A}(i,:)$ and ${\bf A}(:,j)$, respectively, and its transpose and rank are denoted by ${\bf A}^{\intercal}$ and $\rk({\bf A})$, respectively. 
We denote a column vector by ${\bf a} = (a_i)_{i\in[m]} \in \mathbb{F}^{m\times 1}_q$ (or $\mathbb{F}^{m}_q$) and a row vector by ${\bf a} = (a_j)_{j\in[n]} \in \mathbb{F}^{1\times n}_q$, respectively. 
%
$\wedge$ and $\vee$ represent binary AND and OR operations, respectively. 
$a_n=o(n)$ means that $a_n$ grows sublinearly in $n$, i.e., $\lim\limits_{n\to\infty}\frac{a_n}{n}=0$.


\section{System Model}
\label{sec:sys-mod}

This section formulates our $(K, N,\{\mathcal{S}_i\}_{i\in[K]},\{f_{i}(X_1, \dots, X_N)\}_{i\in[K]})$ NFCB problem. The model, shown in Figure~\ref{fig:CBNs}, consists of a sender node, that stores $N\in\mathbb{Z}^+$ datasets, denoted by ${\bf X}_{[N]}\triangleq (X_j)_{j\in[N]}\in\mathbb{F}_q^N$, where $q \geq 2$, and $K$ distributed users. 
The side information available to user~$i$ is $\mathcal S_i\triangleq{\bf X}_{\mathcal A_i}=\left(X_j:j\in\mathcal A_i\right)$. Its alphabet is $\mathsf S_i\triangleq\prod_{j\in\mathcal A_i}\mathcal X_j$ where $\mathcal A_i\subseteq[N]$ is the side-information index set. 
User $i\in[K]$ seeks to compute a function $f_{i}: \mathbb{F}_q^N\to\mathbb{F}_q$ that takes as its arguments the $N$ datasets. Given $t\in[n]$, $i\in[K]$, the pair $f_{i}(\mathbf X_{[N]},t)\triangleq f_{i}({X}_{1,t}, \cdots, X_{N,t})$ and $\mathcal{S}_{i,t}$ are potentially correlated. The requested functions $\{f_{i}\}_{i\in [K]}$ may be distinct.

Consider a length-$n$ source sequence $\mathbf{x}_{[N]}^n=(\mathbf{x}_{[N],1},\ldots,\mathbf{x}_{[N],n})\in\mathcal X_{[N]}^n$, where its components may be arbitrarily dependent. This sequence consists of $n$ i.i.d. realizations drawn according to $P_{X_{[N]}}$. Thus, $X_{[N],t}$ is i.i.d. across $t$. 
%
Define the side-information sequence of user $i$ as $s_i^n(\mathbf{x}_{[N]}^n)\triangleq\bigl(s_i(\mathbf{x}_{[N],1}),\ldots,s_i(\mathbf{x}_{[N],n})\bigr)$, 
and the demand sequence of user $i$ as $f_i^n(\mathbf{x}_{[N]}^n)\triangleq\bigl(f_i(\mathbf{x}_{[N],1}),\ldots,f_i(\mathbf{x}_{[N],n})\bigr)$, which is a length-$n$ i.i.d. sequence of component-wise 
realizations of the desired function $f_i$.


In an $(n,R_n)$ computation-broadcast code, the source is block-encoded with blocklength $n$. The coding scheme consists of an encoder of the source $\mathbf X_{[N]}$ defined by a function $\psi_n:\mathcal X_{[N]}^n\longrightarrow\mathcal M_n$, where $\mathcal M_n$ is a finite broadcast message alphabet, and $K$ decoders $g_{i,n}:\mathcal M_n\times\mathsf S_i^n\longrightarrow\mathcal F_i^n$, $i\in[K]$. 
Specifically, broadcast message $M_n=\psi_n({\bf X}_{[N]}^n)$ is the random message taking values in $\mathcal M_n$, which is a function of $N$ datasets and the set of demanded functions $\{f_i\}_{i\in[K]}$ and side information $\{\mathcal{S}_i\}_{i\in[K]}$ by all users.

The tuple of functions $(\psi_n,\{g_{i,n}\}_{i\in[K]})$ is called an $(n,\epsilon)$-coding scheme if the block error probability satisfies 
\begin{align} 
\label{eq:block-error-probability}
P_{\mathrm{e}}^{(n)}\triangleq\mathbb{P}\left(\bigcup_{i=1}^K\left\{\widehat f_i^n\neq f_i^n\right\}\right)<\epsilon  
\end{align}
where $\widehat f_i^n\triangleq g_{i,n}(M_n,\mathcal S_i^n)$ denotes the length-$n$ sequence estimate. 

Our goal is to minimize the communication cost from the server so that each user can recover its desired function asymptotically losslessly from length-$n$ i.i.d. dataset realizations. 
The system operates in four phases: demand, computing, communication, and reconstruction.

\emph{Demand phase.} Each user $i\in[K]$ communicates over noiseless and interference free channels its requested function $f_{i}(X_1, \dots, X_N)$ to the sender node, which we denote by $f_i$.

\emph{Computing phase.} The sender node, using length-$n$ i.i.d. sequence ${\bf X}_{[N]}^n$, the sender constructs the common broadcast message $M_n=\psi_n({\bf X}_{[N]}^n)$. We assume no constraints on the sender node's computational ability, i.e., the sender can perform non-linear encoding, and the broadcast message can be a nonlinear function of the datasets.

\emph{Communication phase.} The sender node broadcasts $M_n$ to all users. We assume the communication links between the sender node and the users to be noiseless and non-interfering. Given ${\bf X}_{[N]}^n$, the sender transmits $M_n$ at a rate
\begin{align}
\frac1nH_q(M_n)=\frac1n I_q({\bf X}_{[N]}^n;M_n).
\end{align}

At finite blocklenth $n$, the minimum blocklength-$n$ rate with error probability at most $\epsilon$ is defined as
\begin{align}
\label{eq:min_finite_n_rate_lossless}
R_{n,\epsilon}^*(K,N)\triangleq\inf\limits_{\substack{\psi_n,\ \{g_{i,n}\}_{i\in[K]}:\\P_e^{(n)}\leq \epsilon}} \frac1n H_q(M_n) .
\end{align}
Setting $\epsilon=0$ in~\eqref{eq:min_finite_n_rate_lossless}, the minimum blocklength-$n$ zero-error rate is denoted by $R_{n,0}^*(K,N)$. 

\emph{Reconstruction phase}. Each user $i\in[K]$ is allowed to perform any computation (possibly non-linear) using the received broadcast message $M_n$ and its length-$n$ i.i.d. sequence of realizations of the side information $\mathcal{S}_i$, defined as $\mathcal{S}^n_i=\{\mathcal{S}_{i,t}\}_{t\in[n]}$, employing the decoder $g_{i,n}$, user $i$ forms $\widehat f_i^n$.

A rate $R$ is said to be achievable if there exists a sequence of blocklength-$n$ encoders $\psi_n$ and decoders $g_{i,n}$, $i\in[K]$, such that the resulting sequence of codes achieves a block error probability $\lim_{n\to\infty} P_{\mathrm{e}}^{(n)}=0$. 
Equivalently, for the entropy-based formulation, we require the users' probability of error must vanish as the blocklength $n$ increases, while the communication rate approaches a value no larger than $R$: 
\begin{align}
\lim_{n\to\infty}\frac{1}{n}H_q(f_i^n\mid M_n,\mathcal{S}_i^n)=0,\quad i\in[K] , \qquad \lim_{\epsilon\downarrow0}\limsup_{n\to\infty}R_{n,\epsilon}(K,N)\leq R  .
\end{align}
The optimal computation-broadcast rate is the lowest achievable rate
\begin{align}
\label{eq:comm_cost_optimal}
R^*(K,N)\triangleq\inf\left\{R:\ R\text{ is achievable}\right\} 
\end{align}
where the infimum is taken over all sequences of encoding and decoding policies that jointly determine the computing, communication, and reconstruction phases and may depend on 
$(\{f_i\}_{i\in[K]},\{\mathcal{S}_i\}_{i\in[K]})$.

In the case of asymptotic zero-error computation, we require each user $i\in[K]$ to reconstruct $f_i^n$ with zero error, using $M_n$ and $\mathcal S_i^n$, while the communication rate approaches a value no larger than $R_0$: 
\begin{align}
f_i^n = g_{i,n}(M_n,\mathcal S_i^n) \, ,\, \quad i \in [K],\qquad \limsup_{n\to\infty}R_{n,0}(K,N)\leq R_0 .
\end{align}
The optimal computation-broadcast rate is the lowest achievable rate
\begin{align}
\label{eq:def_communication_cost_exact_computing}
R_0^*(K,N)\triangleq\inf\left\{R_0:\ R_0\text{ is achievable}\right\} .
\end{align}

\begin{figure}[t!]
\centering
\includegraphics[width=0.45\textwidth]{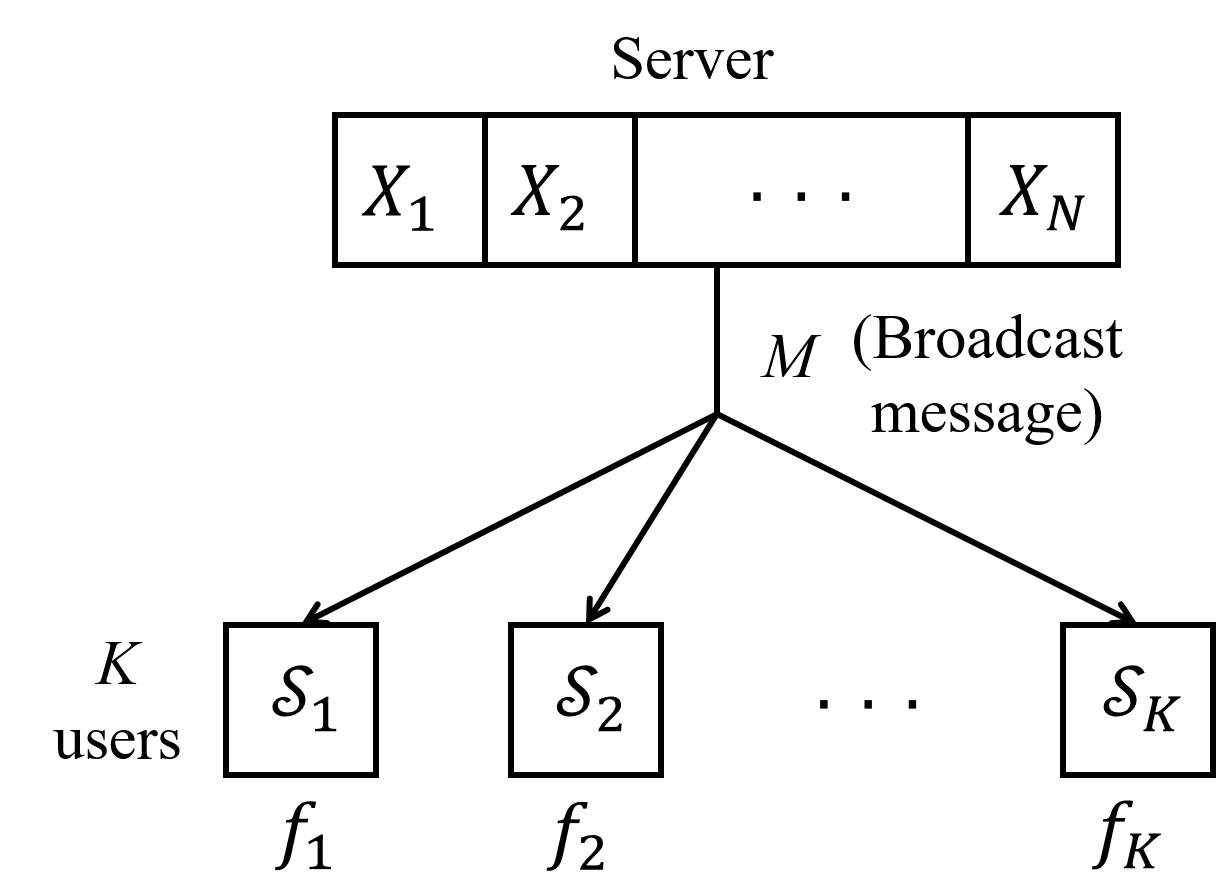}
\caption{\small A $(K,N,\{\mathcal{S}_i\}_{i\in[K]},\{f_i\}_{i\in[K]})$ NFCB model with $N$ datasets ${\bf X}_{[N]}$, $K$ users, with side information $\mathcal{S}_i$ and demand $f_i$, for each $i\in [K]$.}
\label{fig:CBNs}
\end{figure}

\section{Achievability}
\label{sec:achievability}

This section presents novel achievability schemes for the $(K, N,\{\mathcal{S}_i\}_{i\in[K]},\{f_i\}_{i\in[K]})$ NFCB problem, using characteristic graph-based compression. Specifically, we introduce the broadcast graph notion $G_{\mathrm B}$, drawing on the characteristic graph framework (see Definition~\ref{def:broadcast_graph}). The characteristic graph framework, introduced in~\cite{korner1971coding}, can capture structural dependencies among datasets, user-requested functions, and side information of each user, and does not impose any structural assumptions on them. 
The same framework has been used to characterize the communication rate for a point-to-point model, where a source holds a random variable $X_1$ and a receiver with side information $X_2$ aims to reconstruct an arbitrary function $f(X_1, X_2)$ in an asymptotically lossless manner~\cite{OrlRoc2001}. 
Theorem~\ref{theo-general-achievable-rate-via-broadcast} details an achievability scheme under the asymptotically lossless reconstruction criterion. 
Corollary~\ref{cor:side-information-assisted-graph-bound} provides a side-information-assisted independent-set coding scheme.  
Proposition~\ref{prop_achievable-rate-via-broadcast_compatible} characterizes the optimal achievable rate provided that the demanded functions are compatible.

We next define characteristic graphs, generalizing the definitions of source and side information sets (cf.~\cite{korner1971coding},~\cite{OrlRoc2001}).

\begin{defi}[{\bf Characteristic graph}~\cite{OrlRoc2001}]
\label{def-char-graph}
Let ${\bf X}_{[N]}=(X_1,\ldots,X_N)$ be a finite-alphabet source with distribution $P_{{\bf X}_{[N]}}$. Consider a receiver having side information $\mathcal S={\bf X}_{\mathcal A}$, $\mathcal A\subseteq[N]$, and requesting the function $f({\bf X}_{[N]})$.

The source characteristic graph associated with $({\bf X}_{[N]},\mathcal S,f)$ is $G_f=(\mathcal V, \mathcal E)$, where
\begin{align}
\label{eq:vertex_set}
\mathcal V \triangleq {\mathrm{supp}}\left(P_{{\bf X}_{[N]}}\right)\subseteq \prod_{j=1}^N\mathcal X_j \ , 
\end{align} 
and two distinct vertices ${\bf x}_{[N]},\widetilde{\bf x}_{[N]}\in\mathcal V$ are adjacent if and only if
\begin{align}
\label{eq:characteristic-graph-adjacency}
\left\{{\bf x}_{[N]},\widetilde{\bf x}_{[N]}\right\}\in\mathcal E\,\,\Longleftrightarrow\,\,{\bf x}_{\mathcal A}=\widetilde{\bf x}_{\mathcal A}\ \text{and}\ f({\bf x}_{[N]})\neq f(\widetilde{\bf x}_{[N]}) .
\end{align}

\end{defi}

A proper coloring of $G$ is such that adjacent vertices in $G$ (cf.~\eqref{eq:characteristic-graph-adjacency}) have distinct colors. Let $\mathscr C(G)$ denote the collection of proper colorings of a graph $G$ and $c(\mathbf X_{[N]})$ be a proper coloring. The one-shot chromatic entropy is 
\begin{align}
\label{eq:chromatic_entropy_def}
H_\chi(G,{\bf X}_{[N]})\triangleq\min_{c\in\mathscr C(G)}H_q\left(c({\bf X}_{[N]})\right) .
\end{align}

\begin{defi}[{\bf Independent set~\cite{karp2009reducibility}}]
\label{def-In-MIS independent set}
An independent set $\text{IS}$ in a graph $G(\mathcal{V}, \mathcal{E})$ is a subset of vertices with no edges between them. Such non-adjacent vertices can be encoded without ambiguity, i.e., assigned the same codeword. 
A maximal independent set (MIS) is an IS that is not a subset of any other independent set
~\cite{karp2009reducibility}. The set of all MISs of $G$ is denoted by $\Gamma(G)$, and an MIS of $G$ by $W\in \Gamma(G)$, respectively. 
\end{defi}

K\"orner introduced the notion of graph entropy in~\cite{korner1971coding}, which is defined next.
\begin{defi}[K\"orner's graph entropy~\cite{korner1971coding}]
\label{def:Korner_graph_entropy}
Given a graph $G = (\mathcal{V} ,\mathcal{E})$ with a distribution on its vertices $\mathcal{V}$ and the set of all MISs $\Gamma(G)$, graph entropy is
\begin{align}
\label{eq:graph_entropy}
H_\kappa(G,{\bf X}_{[N]})
\triangleq\min_{\substack{P_{W\mid{\bf X}_{[N]}}:\\{\bf X}_{[N]}\in W,\;W\in\Gamma(G)}}I_q({\bf X}_{[N]};W)
=\lim_{n\to\infty}\frac1n H_\chi\left(G^{\dot \vee n},{\bf X}_{[N]}^n\right) 
\end{align}
where ${\bf X}_{[N]}\in W \in \Gamma(G)$ means that the minimization is over all distributions $P_{{\bf X}_{[N]},W}({\bf x}_{[N]},w)$ such that $P_{{\bf X}_{[N]},W}({\bf x}_{[N]},w) > 0$ implies ${\bf x}_{[N]} \in w$, where $w$ is a MIS of $G$.
\end{defi}

\begin{defi}[$n$-fold OR power]
\label{def:or-power}
The $n$-fold OR power of a characteristic graph $G=(\mathcal V(G),\mathcal E(G))$ is denoted by
\begin{align}
G^{\dot \vee n}=\left(\mathcal V^n,\mathcal E^{\dot \vee n}\right) ,
\end{align}
which has vertex set $\mathcal V^n={\mathrm{supp}}\left(P_{\mathbf X_{[N]}^n}\right)=\left({\mathrm{supp}}\left(P_{\mathbf X_{[N]}}\right)\right)^n$, and two distinct sequences ${\bf x}_{[N]}^n,\widetilde{\bf x}_{[N]}^n\in\mathcal V^n$ are adjacent if they are adjacent in $G_f$ in at least one coordinate:
\begin{align}
\label{eq:or-power-adjacency}
\left\{{\bf x}_{[N]}^n,\widetilde{\bf x}_{[N]}^n\right\}\in\mathcal E^{\dot \vee n} \,\,\Longleftrightarrow\,\, \exists t\in[n]:\left\{{\bf x}_{[N],t},\widetilde{\bf x}_{[N],t}\right\}\in\mathcal E \,\,\Longleftrightarrow\,\,
\exists t\in[n]:
\begin{cases}
{\bf x}_{\mathcal A,t}=\widetilde{\bf x}_{\mathcal A,t},\\
f({\bf x}_{[N],t})\neq f(\widetilde{\bf x}_{[N],t}).
\end{cases} 
\end{align}
\end{defi}


If $c$ is a coloring of G then $W \triangleq c^{-1}(c(X))$, the color class of ${\bf X}_{[N]}$, ranges disjointly over $\Gamma(G)$ and always contains ${\bf X}_{[N]}$. 
Conversely, When $W$ ranges disjointly over $\Gamma(G)$ and always contains ${\bf X}_{[N]}$, it represents the color class of ${\bf X}_{[N]}$ in a coloring of $G$. 
$H_q(W\mid {\bf X}_{[N]})=0$. In all these cases, $H_q(W\mid {\bf X}_{[N]})=0$. We write $W\overset{\circ}{\in}\Gamma(G)$, if each $W$ is in $\Gamma(G)$ and distinct values are disjoint. 
Thus, the one-shot chromatic entropy satisfies 
\begin{align}
\label{eq:chromatic_entropy_vs_graph_entropy}
H_\chi(G,{\bf X}_{[N]})=\min\limits_{\substack{P_{W\mid{\bf X}_{[N]}}:\\{\bf X}_{[N]}\in W,\;W\overset{\circ}{\in}\Gamma(G)}}\ H_q(W) \geq \min\limits_{\substack{P_{W\mid{\bf X}_{[N]}}:\\{\bf X}_{[N]}\in W,\;W\in\Gamma(G)}}\ H_q(W)\geq H_\kappa(G,{\bf X}_{[N]}) .
\end{align}
Thus, graph entropy $H_\kappa$ generalizes chromatic entropy $H_\chi$ by assigning a random color to every vertex (cf.~\cite{CarFloJor2008}).

Characteristic graphs identify which source values can share the same codeword without introducing ambiguity to the user, enabling distributed computation in various settings, as explored in prior literature (see e.g.,~\cite{AlonOrlit1996, OrlRoc2001, feizi2014network, salehi2023achievable, malak2024multis, deylam2025graph, korner1971coding,malak2023distributed}). For a point-to-point scheme with side information $\mathcal{S}$, exploiting~\cite{OrlRoc2001}, the minimum rate required to compute $f(\mathbf X_{[N]})$ with vanishing error 
is characterized by the conditional graph entropy of ${\bf X}_{[N]}$ given $\mathcal{S}$:
\begin{align}
\label{eq:Rate_ach}
H_\kappa(G,{\bf X}_{[N]}\mid \mathcal{S})\triangleq \min_{\substack{P_{W\mid{\bf X}_{[N]}}:\\W-{\bf X}_{[N]}-\mathcal S,\\{\bf X}_{[N]}\in W\in\Gamma(G)\ \text{a.s.}}}\ I_q({\bf X}_{[N]};W\mid \mathcal{S}) 
\end{align}
where $W-{\bf X}_{[N]}-\mathcal{S}$ indicates a Markov chain,  $I_q({\bf X}_{[N]}; W \mid \mathcal{S})=H_{q}(W\mid \mathcal{S})-H_{q}(W\mid {\bf X}_{[N]})$, and  ${\bf X}_{[N]}\in W\in \Gamma(G)$ means that the minimization is over all $P_{\mathbf X_{[N]}}(\mathbf x_{[N]},w)>0$ such that $W$ is an IS of $G$. 

To address the NFCB problem, below we generalize the characteristic graph framework to a broadcast graph framework, which captures the statistics of ${\bf X}_{[N]}$, the $K$ (possibly non-linear) demanded functions $\{f_i\}_{i\in[K]}$, and the distinguishability requirements induced by $\{f_i,\mathcal{S}_i\}_{i\in[K]}$. For each user $i\in[K]$, we denote by $G_{f_i}=(\mathcal{V}, \mathcal{E}_i)$ the characteristic graph constructed from $\mathbf X_{[N]}$ under $P_{\mathbf X}$, with side information $\mathcal{S}_i$ and demand $f_i$.


\begin{defi}[{\bf Broadcast graph}~\cite{korner1986fredman}]
\label{def:broadcast_graph}
For each $i\in[K]$, let $G_{f_i} =(\mathcal{V}, \mathcal{E}_i)$ be the characteristic graph associated with $(f_i,\mathcal S_i)$ according to~\eqref{eq:characteristic-graph-adjacency}, where $\mathcal{V}={\mathrm{supp}}\left(P_{{\bf X}_{[N]}}\right)$. The single-letter broadcast graph is
\begin{align}
\label{eq:single-letter_broadcast_graph}
G_{\mathrm B}\triangleq \bigcup_{i=1}^{K}G_{f_i}=\left(\mathcal V,\bigcup_{i=1}^{K}\mathcal E_i\right) . 
\end{align}
\end{defi}

We next define the blocklength-$n$ characteristic graph of user $i$.
\begin{defi}[Length-$n$ graph]
\label{def:blocklength_n_BC}
For each $i\in[K]$, let $G_{f_i}^{(n)}=\bigl(\mathcal X_{[N]}^n,\mathcal E_i^{(n)}\bigr)$ be the length-$n$ graph, where
\begin{align}
\label{eq:block_same_side_information_block_distinct_demands}
\mathcal E_i^{(n)}\triangleq\left\{\left\{{\bf x}_{[N]}^n,\widetilde{\bf x}_{[N]}^n\right\}:\;
{\bf x}_{\mathcal A_i}^n=\widetilde{\bf x}_{\mathcal A_i}^n,\ 
f_i^n({\bf x}_{[N]}^n)\neq f_i^n(\widetilde{\bf x}_{[N]}^n)
\right\} .
\end{align}
\end{defi}

Using Definition~\ref{def:blocklength_n_BC}, we next define the operational block broadcast graph.

\begin{defi}[Length-$n$ operational block broadcast graph]
\label{def:operational-block-graph}
For each blocklength $n$, let 
\begin{align}
\label{eq:block_broadcast_graph}
G_{\mathrm B}^{(n)}\triangleq\bigcup_{i=1}^K G_{f_i}^{(n)}=\left(\mathcal X_{[N]}^n,\mathcal E_{\mathrm B}^{(n)}\right) ,\qquad \text{where}\quad \mathcal E_{\mathrm B}^{(n)}\triangleq\bigcup_{i=1}^K\mathcal E_i^{(n)} ,
\end{align}
denote the operational block broadcast graph, where $\mathcal E_i^{(n)}$ is given in~\eqref{eq:block_same_side_information_block_distinct_demands}. 
Thus,
\begin{align}
\label{eq:block-broadcast-adjacency}
\bigl\{{\bf x}_{[N]}^n,\widetilde{\bf x}_{[N]}^n\bigr\}\in\mathcal E_{\mathrm B}^{(n)}
\,\,\Longleftrightarrow\,\,
\exists i\in[K]:\bigl\{{\bf x}_{[N]}^n,\widetilde{\bf x}_{[N]}^n\bigr\}\in
\mathcal E_i^{(n)}
\,\,\Longleftrightarrow\,\, \exists i\in[K]:
\begin{cases}
{\bf x}_{\mathcal A_i}^n=\widetilde{\bf x}_{\mathcal A_i}^n,\\[1mm]
f_i^n({\bf x}_{[N]}^n)\neq f_i^n(\widetilde{\bf x}_{[N]}^n) .
\end{cases}
\end{align}
Thus, two source sequences are adjacent in the length-$n$ broadcast graph if there exists one user whose side-information sequences are identical but whose demanded function sequences are different.
\end{defi}

In general, for a given $q\geq 2$, $|\mathcal{V}|=q^N$ for the broadcast graph $\GU=(\mathcal{V},\mathcal E)$, the $n$-fold OR product $\GUn=G(\mathcal{V}^n,\mathcal E^{\dot \vee n})$ where  $|\mathcal{V}^n|=|\mathcal{V}|^n = q^{Nn}$. The exponential growth in $n$ renders the characteristic graph-based approach impractical, as finding a closed-form expression for~\eqref{eq:graph_entropy} (see also Theorem~\ref{theo-general-achievable-rate-via-broadcast}) is infeasible. 
Independent sets induced by the single-letter block graph are a subset of the independent sets of the exact block broadcast graph $G_{\mathrm B}^{(n)}$, which may contain fewer edges\footnote{Recall that the $n$-fold OR power $\GUn$ (Definition~\ref{def:or-power}) permits the user witnessing confusability to vary across coordinates, whereas $G_{\mathrm B}^{(n)}$ (Definition~\ref{def:operational-block-graph}) requires one fixed user to have identical side information over the entire block.} versus the OR-power construction $\GUn$ of the single-letter graph $G_{\mathrm B}=G_{\mathrm B}^{(1)}$, i.e.,
\begin{align}
\label{eq:operational-subgraph-or-power}
G_{\mathrm B}^{(n)}\subseteq \GUn ,
\end{align}
and the inclusion can be strict, as Example~\ref{ex-linear-comp-jafar} demonstrates. 
Employing~\eqref{eq:operational-subgraph-or-power}, we will later show that the resulting expression that uses $\GUn$ in Theorem~\ref{theo-general-achievable-rate-via-broadcast} is an achievable bound, and it need not equal the optimal broadcast rate.

The operational block broadcast graph definition, i.e., Definition~\ref{def:operational-block-graph} (cf.~\eqref{eq:block_broadcast_graph}), requires a single user to witness confusability over the entire block. 
The block graph is generally neither an OR nor an AND power. 
A single-letter characteristic graph records only whether two source realizations must be distinguished. It does not separately record all the relations needed to determine whether two length-$n$ source sequences are confusable for the same user.

For each user $i\in[K]$, define the side-information compatibility and the demand-disagreement relations
\begin{align}
\label{eq:side_information_compatibility_demand_disagreement}
A_i(\mathbf{x}_{[N]},\widetilde{\mathbf{x}}_{[N]})\triangleq\mathbf{1}\left\{\mathbf{x}_{\mathcal A_i}=\widetilde{\mathbf{x}}_{\mathcal A_i}\right\} ,\qquad
D_i(\mathbf{x}_{[N]},\widetilde{\mathbf{x}}_{[N]})\triangleq\mathbf{1}\left\{f_i(\mathbf{x}_{[N]})\neq f_i(\widetilde{\mathbf{x}}_{[N]})\right\} ,
\end{align}
respectively. Two distinct realizations $\mathbf{x}_{[N]},\widetilde{\mathbf{x}}_{[N]}\in\mathcal X_{[N]}$ are adjacent in the single-letter graph $G_{f_i}=(\mathcal X_{[N]},\mathcal E_i)$ if and only if $\{\mathbf{x}_{[N]},\widetilde{\mathbf{x}}_{[N]}\}\in \mathcal E_i \,\Longleftrightarrow\, A_i(\mathbf{x}_{[N]},\widetilde{\mathbf{x}}_{[N]}) D_i(\mathbf{x}_{[N]},\widetilde{\mathbf{x}}_{[N]})=1$. Thus, $G_{f_i}$ retains the conjunction $A_i\wedge D_i$, but does not separately retain the 
compatibility relation $A_i$, which becomes relevant at blocklength $n$.

\medskip
\noindent\textit{\textbf{Operational block-confusability condition.}}
Consider two source sequences
\begin{align}
\mathbf{x}_{[N]}^n=(\mathbf{x}_{[N],1},\ldots,\mathbf{x}_{[N],n})\in\mathcal X_{[N]}^n,\qquad 
\widetilde{\mathbf{x}}_{[N]}^n=(\widetilde{\mathbf{x}}_{[N],1},\ldots,\widetilde{\mathbf{x}}_{[N],n})\in\mathcal X_{[N]}^n .
\end{align}
They are confusable for user $i$ if that user observes the same side-information sequence under both source realizations, while the corresponding demand sequences differ in at least one coordinate. Alternatively, 
\begin{align}
\label{eq:block_side_information_compatibility_block_demand_disagreement}
\mathbf{x}_{{\mathcal A_i},t}=\widetilde{\mathbf{x}}_{{\mathcal A_i},t},\quad \forall t\in[n] ,\qquad\text{and}\qquad f_i(\mathbf{x}_{[N],t})\neq f_i(\widetilde{\mathbf{x}}_{[N],t})\quad\text{for some }t\in[n] .
\end{align}
In terms of~\eqref{eq:side_information_compatibility_demand_disagreement},
\begin{align}
\label{eq:block_user_confusability_relations}
\{\mathbf{x}_{[N]}^n,\widetilde{\mathbf{x}}_{[N]}^n\}\in \mathcal E_i^{(n)}\,\,\Longleftrightarrow\,\,\left[\bigwedge_{t=1}^n A_i(\mathbf{x}_{[N],t},\widetilde{\mathbf{x}}_{[N],t})\right]\wedge \left[\bigvee_{t=1}^n D_i(\mathbf{x}_{[N],t},\widetilde{\mathbf{x}}_{[N],t})\right]=1 .
\end{align}
The single-letter graph $G_{f_i}$ records only whether $\mathbf{x}_{\mathcal A_i}=\widetilde{\mathbf{x}}_{\mathcal A_i}$ and $f_i(\mathbf{x}_{[N]})\neq f_i(\widetilde{\mathbf{x}}_{[N]})$. However, it does not separately record the non-edge relation $\mathbf{x}_{\mathcal A_i}=\widetilde{\mathbf{x}}_{\mathcal A_i}$ and $f_i(\mathbf{x}_{[N]})=f_i(\widetilde{\mathbf{x}}_{[N]})$, which is also needed to construct the correct block (the blocklength-$n$) confusability graph $G_{\mathrm B}^{(n)}$ obtained by taking the union over the users:
\begin{align}
\label{eq:block_broadcast_confusability_relations}
\{\mathbf{x}_{[N]}^n,\widetilde{\mathbf{x}}_{[N]}^n\}\in\mathcal E_{\mathrm B}^{(n)}
\,\,\Longleftrightarrow\,\,\bigvee_{i=1}^K\Bigg\{\left[\bigwedge_{t=1}^n A_i(\mathbf{x}_{[N],t},\widetilde{\mathbf{x}}_{[N],t})\right]\wedge\left[\bigvee_{t=1}^n D_i(\mathbf{x}_{[N],t},\widetilde{\mathbf{x}}_{[N],t})\right]\Bigg\}=1.
\end{align}
Importantly, the same user $i$ must satisfy both block conditions in~\eqref{eq:block_broadcast_confusability_relations}. That user must observe identical side information over the entire block and require different demand sequences.

\medskip
\noindent\textit{\textbf{Comparison with the OR power.}}
The $n$-fold OR power of the single-letter broadcast graph $G_{\mathrm B}$ declares two sequences adjacent if
\begin{align}
\label{eq:or_power_relation}
\{\mathbf{x}_{[N]}^n,\widetilde{\mathbf{x}}_{[N]}^n\}\in \mathcal E^{\dot \vee n}\,\,\Longleftrightarrow\,\,\bigvee_{t=1}^n \bigvee_{i=1}^K \left[A_i(\mathbf{x}_{[N],t},\widetilde{\mathbf{x}}_{[N],t})\wedge D_i(\mathbf{x}_{[N],t},\widetilde{\mathbf{x}}_{[N],t})\right]=1.
\end{align}
Condition~\eqref{eq:or_power_relation} requires only one coordinate at which some user is confused. It does not require the user's side-information realizations to agree at all the remaining coordinates.

For example, let $n=2$ and suppose that
\begin{align}
A_1(\mathbf{x}_{[N],1},\widetilde{\mathbf{x}}_{[N],1})=1,
\qquad
D_1(\mathbf{x}_{[N],1},\widetilde{\mathbf{x}}_{[N],1})=1,
\qquad
A_1(\mathbf{x}_{[N],2},\widetilde{\mathbf{x}}_{[N],2})=0.
\end{align}
The OR power declares $\mathbf{x}_{[N]}^2$ and $\widetilde{\mathbf{x}}_{[N]}^2$ adjacent because their first-coordinate symbols are adjacent in $G_{f_1}$. However, user~$1$ can distinguish the two source sequences from its side information at the second coordinate. Hence, the two sequences are not necessarily operationally confusable for user~$1$.

Moreover, the user witnessing a coordinatewise edge may vary with the coordinate. For instance, user~$1$ may induce the first-coordinate edge, whereas user~$2$ may induce the second-coordinate edge. In that case, the OR power detects coordinatewise edges, even though no single user necessarily observes identical side-information sequences over the complete block. Therefore, $G_{\mathrm B}^{(n)}\neq \GUn$ in general.

\medskip
\noindent\textit{\textbf{Comparison with the AND power.}}
The $n$-fold AND power requires that every coordinate where the two source symbols differ corresponds to an edge in the single-letter graph, which can be too restrictive for function computation.

For example, consider one user and $n=2$. Suppose that at the first coordinate, $\mathbf{x}_{[N],1}\neq\widetilde{\mathbf{x}}_{[N],1}$, $A_i(\mathbf{x}_{[N],1},\widetilde{\mathbf{x}}_{[N],1})=1$, and $D_i(\mathbf{x}_{[N],1},\widetilde{\mathbf{x}}_{[N],1})=1$, so that $\{\mathbf{x}_{[N],1},\widetilde{\mathbf{x}}_{[N],1}\}\in \mathcal E_i$. 
At the second coordinate, suppose that $\mathbf{x}_{[N],2}\neq\widetilde{\mathbf{x}}_{[N],2}$, $A_i(\mathbf{x}_{[N],2},\widetilde{\mathbf{x}}_{[N],2})=1$, and $D_i(\mathbf{x}_{[N],2},\widetilde{\mathbf{x}}_{[N],2})=0$. The second-coordinate symbols are distinct but nonadjacent because they produce the same demanded value for user~$i$. Nevertheless, the complete source sequences are operationally confusable as user~$i$ observes the same side-information sequence under both realizations, while its demand sequences differ at the first coordinate. Hence, $\{\mathbf{x}_{[N]}^2,\widetilde{\mathbf{x}}_{[N]}^2\}
\in \mathcal E_i^{(2)}$, although the AND power does not connect them because the second-coordinate pair is nonadjacent. Therefore, $G_{\mathrm B}^{(n)}\neq G_{\mathrm B}^{\dot \wedge n}$ in general.

The single-letter graph $G_{f_i}$ records only whether $\mathbf{x}_{\mathcal A_i}=\widetilde{\mathbf{x}}_{\mathcal A_i}$ and $f_i(\mathbf{x}_{[N]})\neq f_i(\widetilde{\mathbf{x}}_{[N]})$. However, it does not separately record the non-edge relation $\mathbf{x}_{\mathcal A_i}=\widetilde{\mathbf{x}}_{\mathcal A_i}$ and $f_i(\mathbf{x}_{[N]})=f_i(\widetilde{\mathbf{x}}_{[N]})$, which is also needed to construct the correct block confusability graph $G_{\mathrm B}^{(n)}$~\eqref{eq:block_broadcast_confusability_relations}. 
Thus, $G_{f_i}$ alone generally does not contain enough information to recover the block graph through an ordinary OR or AND product. 

Taking the union of the graphs $G_{f_i}$, $i\in[K]$ before forming a graph power also removes the identity of the user responsible for each edge. Thus, neither the OR nor the AND power generally commutes with the multi-user union:
\begin{align}
\label{nonequality_BroadcastBlockGraph_OR_AND}
\left(\bigcup_{i=1}^K G_{f_i}\right)^{\dot{\vee}}\neq \bigcup_{i=1}^K G_{f_i}^{(n)} ,\quad \left(\bigcup_{i=1}^K G_{f_i}\right)^{\dot{\wedge}}\neq \bigcup_{i=1}^K G_{f_i}^{(n)} .
\end{align}

Every proper coloring $c_n$ of $G_{\mathrm B}^{(n)}$ gives a zero-error message $M_n=c_n(\mathbf X_{[N]}^n)$. Conversely, the preimage of each message produced by a deterministic zero-error encoder is an independent set of $G_{\mathrm B}^{(n)}$. Hence, every deterministic zero-error encoder induces a proper coloring. 
The zero-error operational graph entropy is then defined next.
\begin{defi}[Zero-error operational broadcast graph entropy]
\label{def:zero-error-operational-entropy}
The chromatic entropy of $G_{\mathrm B}^{(n)}$ is
\begin{align}
\label{eq:operational-chromatic-entropy}
H_\chi\left(G_{\mathrm B}^{(n)},{\bf X}_{[N]}^n\right)\triangleq\min\limits_{c_n\in\mathscr C(G_{\mathrm B}^{(n)})}H_q\left(c_n({\bf X}_{[N]}^n)\right) .
\end{align}
Defining $R_{n,0}^*(K,N)=\frac{1}{n}H_\chi\left(G_{\mathrm B}^{(n)},{\bf X}_{[N]}^n\right)$, the asymptotic zero-error operational broadcast graph entropy is
\begin{align}
\label{eq:asymptotic-zero-error-broadcast-entropy}
R_0^*(K,N)
\triangleq
\limsup_{n\to\infty}R_{n,0}^*(K,N)
\end{align}
where the existence of the limit is ensured by subadditivity\footnote{Concatenating proper colorings of $G_{\mathrm B}^{(m)}$ and $G_{\mathrm B}^{(n)}$ produces a proper coloring of $G_{\mathrm B}^{(m+n)}$. Therefore, $H_\chi\left(G_{\mathrm B}^{(m+n)},{\bf X}_{[N]}^{m+n}\right)\leq H_\chi\left(G_{\mathrm B}^{(m)},{\bf X}_{[N]}^m\right)+H_\chi\left(G_{\mathrm B}^{(n)},{\bf X}_{[N]}^n\right)$. Thus, Fekete's lemma gives $R_0^*(K,N)
=\lim\limits_{n\to\infty}\frac{1}{n}H_\chi\left(G_{\mathrm B}^{(n)},{\bf X}_{[N]}^n\right)$.} of the sequence $\{H_\chi\left(G_{\mathrm B}^{(n)},{\bf X}_{[N]}^n\right)\}_{n\geq1}$. 
\end{defi}

From $G_{\mathrm B}^{(n)} \subseteq \GUn$ and K\"orner's graph-entropy theorem, we obtain 
$H_\chi\left( G_{\mathrm B}^{(n)},{\bf X}_{[N]}^n \right) \leq H_\chi\left( \GUn,{\bf X}_{[N]}^n \right)$. 
After normalization and taking the limit, 
\begin{align}
\label{eq:operational-entropy-upper-bound}
R_0^*(K,N)
\leq H_\kappa\left(G_{\mathrm B},\mathbf X_{[N]}\right) .
\end{align}

The asymptotically lossless operational broadcast graph entropy is defined as follows.
\begin{defi}[Asymptotically lossless operational graph entropy]
\label{def:asymptotically-lossless-operational-entropy}
For $\epsilon\in(0,1)$, let
\begin{align}
H_{\chi,\epsilon}\Bigl(G_{\mathrm B}^{(n)},\mathbf X_{[N]}^n\Bigr)\triangleq\min_{\substack{\mathcal T_n\subseteq\mathcal X_{[N]}^n:\\
\mathbb{P}({\bf X}_{[N]}^n\in\mathcal T_n)\geq1-\epsilon}}\;
\min_{c_n\in\mathscr C(G_{\mathrm B}^{(n)}[\mathcal T_n])} H_q\left(\widetilde c_n({\bf X}_{[N]}^n)\right) 
\label{eq:epsilon-chromatic-entropy}
\end{align}
denote the minimum chromatic entropy after excluding a set of source sequences having total probability at most $\epsilon$, 
where $G_{\mathrm B}^{(n)}[\mathcal T_n]$ is the subgraph induced by $\mathcal T_n\subseteq\mathcal X_{[N]}^n$, and $\widetilde c_n$ is the extension $\widetilde c_n({\bf x}_{[N]}^n)=c_n({\bf x}_{[N]}^n) \mathbf{1}\left\{{\bf x}_{[N]}^n\in\mathcal T_n\right\}+\mathsf e \mathbf{1}\left\{{\bf x}_{[N]}^n\notin\mathcal T_n\right\}$, with $\mathsf e$ denoting an error symbol, and $H_{\chi,\epsilon}$.

Defining $R_{n,\epsilon}^*(K,N)=\frac{1}{n}H_{\chi,\epsilon}\left(G_{\mathrm B}^{(n)},\mathbf X_{[N]}^n\right)$, the asymptotically lossless operational broadcast graph entropy is
\begin{align}
\label{eq:asymptotically-lossless-operational-entropy}
R^*(K,N)
\triangleq\lim_{\epsilon\downarrow0}\limsup_{n\to\infty} R_{n,\epsilon}^*(K,N) .
\end{align}
\end{defi}


For every $\epsilon>0$, the choice $\mathcal T_n=\mathcal X_{[N]}^n$ is feasible in the definition of $H_{\chi,\epsilon}$. Therefore,
\begin{align}
H_{\chi,\epsilon}\left(G_{\mathrm B}^{(n)},{\bf X}_{[N]}^n\right)\leq H_\chi\left(G_{\mathrm B}^{(n)},\mathbf X_{[N]}^n\right) = \lim_{\epsilon\downarrow0} H_{\chi,\epsilon}\left(G_{\mathrm B}^{(n)},{\bf X}_{[N]}^n\right)
\end{align}
where the equality is because for every fixed $n$, the source alphabet is finite. 
Therefore,
\begin{align}
R_0^*(K,N)
=\lim_{n\to\infty}\lim_{\epsilon\downarrow0}\frac{1}{n}H_{\chi,\epsilon}\left(G_{\mathrm B}^{(n)},\mathbf X_{[N]}^n\right) .
\end{align}

Noting that every zero-error code is also asymptotically lossless, and employing~\eqref{eq:operational-entropy-upper-bound}
\begin{align}
\label{eq:broadcast-entropy-hierarchy}
R^*(K,N)
\leq R_0^*(K,N)
\leq H_\kappa\left(G_{\mathrm B},{\bf X}_{[N]}\right) 
\end{align}
where $R_0^*(K,N)
$ and $R^*(K,N)
$ are given  in~\eqref{eq:asymptotic-zero-error-broadcast-entropy} and in~\eqref{eq:asymptotically-lossless-operational-entropy}, respectively.


Next, exploiting Definition~\ref{def:broadcast_graph}, we present our achievability scheme.


\begin{theo}
\label{theo-general-achievable-rate-via-broadcast}
Consider a $(K, N,\{\mathcal{S}_i\}_{i\in[K]},\{f_i\}_{i\in[K]})$ NFCB model, where $\mathcal S_i={\bf X}_{\mathcal A_i}$, $\mathcal A_i\subseteq[N]$. 
Let $G_{\mathrm B}\triangleq G_{\cup_{i=1}^{K}f_i}$ denote the broadcast characteristic graph, and let $\Gamma(G_{\mathrm B})$ denote the collection of independent sets of $G_{\mathrm B}$. 
Define 
\begin{align}
\label{eq-theo-general-achievable-rate-via-broadcast}
 H_\kappa(G_{\mathrm B},{\bf X}_{[N]}) \triangleq\min\limits_{\substack{P_{W\mid{\bf X}_{[N]}}:\\{\bf X}_{[N]}\in W,\;W\in\Gamma(G_{\mathrm B})}} I_q({\bf X}_{[N]};W) .
\end{align}
Then $H_\kappa(G_{\mathrm B},{\bf X}_{[N]})$ is achievable under the asymptotically lossless reconstruction criterion. Thus,
\begin{align}
R^*(K,N)\leq H_\kappa(G_{\mathrm B},{\bf X}_{[N]}) 
\end{align}
where $H_\kappa(G_{\mathrm B},{\bf X}_{[N]})$ is K\"orner's graph entropy of $\mathbf X_{[N]}$ with respect to the single-letter broadcast graph $G_{\mathrm B}$.
\end{theo}

Although the reconstruction criterion considered here is asymptotically lossless, the proposed graph-entropy construction has zero error. A proper coloring of the OR power graph $\GUn$ separates every pair of source sequences that are adjacent in at least one coordinate. Thus, every user can recover its demanded function sequence exactly from the transmitted color and its side-information sequence. 
By applying K\"orner's graph-entropy theorem to the $n$-fold OR product of the single-letter graph $G_{\mathrm B}$~\cite{korner1971coding} (see also~\cite{alon1996source}),
\begin{align}
\label{eq:broadcast_graph_entropy}
R^*(K,N)\leq H_\kappa\bigl(G_{\mathrm B},\mathbf X_{[N]}\bigr)=\lim_{n\to\infty} \frac1n \min_{c_n\in\mathscr C(\GUn)} H_q\bigl(c_n(\mathbf X_{[N]}^n)\bigr) ,
\end{align}
where the minimum is over all proper colorings of $\GUn$.

\begin{proof}[Proof of Theorem~\ref{theo-general-achievable-rate-via-broadcast}]
\label{proof_theo-general-achievable-rate-via-broadcast} 
We denote a text channel attaining~\eqref{eq-theo-general-achievable-rate-via-broadcast} by
\begin{align}
\label{eq:Wopt}
P^*_{W\mid{\bf X}_{[N]}}\in\arg\min_{\substack{P_{W\mid{\bf X}_{[N]}}:\\{\bf X}_{[N]}\in W,\;W\in\Gamma(G_{\mathrm B})}}I_q({\bf X}_{[N]};W)
\end{align}
where $W$ denotes the independent-set-valued random variable\footnote{We may restrict $W$ to maximal independent sets, since every independent set can be enlarged to a maximal one.} 
induced by $P^*_{\mathbf X_{[N]},W}=P_{\mathbf X_{[N]}}P^*_{W\mid{\mathbf X}_{[N]}}$.

Assigning a unique codeword to each MIS in $\Gamma(\GU)$, which is determined by ${\bf X}_{[N]}$, $\{f_i\}_{i\in[K]}$, and $\{\mathcal{S}_i\}_{i\in[K]}$, ensures that each user $i\in[K]$ can correctly reconstruct $f_{i}$ from the received $M$ and $\mathcal{S}_i$. 
For asymptotically lossless recovery of all demands, the sender node builds the $n$-fold OR power graph\footnote{For more details on the $n$-fold OR products of characteristic graphs, we refer the reader to~\cite{AlonOrlit1996}.} $\GUn\triangleq G(\mathcal{V}^n ,\mathcal{E}^n)$ corresponding to the length-$n$ i.i.d realizations of $N$ datasets, namely ${\bf X}_{[N]}^{n}=\{(X_{1,t},\dots,X_{N,t})\}_{t\in [n]}$. 
The vertex set satisfies $\mathcal{V}^n=\{{\bf x}_{[N]}^n \triangleq \{{\bf x}_{[N],t}\}_{t\in [n]} 
\,:\, {\bf x}_{[N],t}\in \mathcal{V},\,\forall\, t\in [n]\}$ where each ${\bf x}_{[N],t}\triangleq (x_{1,t},x_{2,t},\dots,x_{N,t})$ is drawn i.i.d. from the same joint distribution $P_{X_1,X_2,\dots,X_N}$. 
The edge set satisfies $\mathcal{E}^n=\{({\bf x}_{[N]}^n,{{\bf x}'}_{[N]}^n)\, \mbox{if}\, \ \exists\, t\in[n] \, :\, ({\bf x}_{[N],t},{\bf x}'_{[N],t})\in \mathcal{E}\}$. 

\paragraph{Codebook generation}  
After building $\GUn$ using ${\bf X}_{[N]}^{n}$, $\{f_{i}^{n}\}_{i\in[K]}$, and $\{\mathcal{S}_{i}^{n}\}_{i\in[K]}$, the sender node employs random binning of the MISs $\{W^n\in\Gamma(\GUn)\}$ via the mapping $en:\Gamma(\GUn)\to [1, q^{nR(K,N)}]$. We refer to each subset of $W^n$ with the same index $m$ as a bin $\mathcal{B}(m)$, $m\in[1, q^{nR(K,N)}]$. The codebook, denoted as $\mathcal C=\{m\in [1,q^{nR}]\}$, is revealed to all users.

\paragraph{Encoding}
Upon observing $W^{n}$, the sender finds an index $m\in [1:q^{nR}]$ such that $W^n\in\mathcal{B}(m)$. Thus, the broadcast message is $en(W^n)=m$. The sender broadcasts the bin index $m$.

\paragraph{Decoding}
Upon receiving $m$, the user declares $\hat{f}_{i}^{n}\triangleq F_i(W^n, \mathcal{S}_{i}^{n})$, a 
function that denotes an estimate of 
$f_{i}^{n}$ if $W^n$ is the unique typical sequence in $\mathcal{B}(m)$. Otherwise, it declares an error.

\paragraph{Analysis of the probability of error}
We show that the probability of error averaged over bin assignments tends to zero as $n\to\infty$ if $R>H_q(W)$.

Let $M$ denote the random bin index of $W^n$, i.e., $W^n\in\mathcal{B}(M)$. Note that $M\sim{\mathrm Unif}[1,2^{nR}]$ is independent of $W^n$. User $i\in[K]$ makes an error if and only if at least one of the following occurs:
\begin{align}
    E_{r_{i1}}=\{W^{n}\notin \mathcal{T}_{\epsilon}^{(n)}\} \ ,\quad
    E_{r_{i2}}=\{\tilde{w}^{n}\in \mathcal{B}(M)\,\mbox{for some}\ \tilde{w}^{n}\neq W^n,\ \tilde{w}^n\in \mathcal{T}_{\epsilon}^{(n)}\} \nonumber
\end{align}
where $\mathcal{T}_{\epsilon}^{(n)}$ represents the set of $\epsilon$-typical $n$-sequences $W^n$ for some $\epsilon>0$. 
Then, by the symmetry of codebook construction, the union of events bound, random binning ensures that the probability of error averaged over bin assignments tends to zero as $n\to\infty$ if $R>H_q(W)+\delta(\epsilon)$ where $|\mathcal{T}_{\epsilon}^{(n)}|\leq q^{n(H_q(W)+\delta(\epsilon))}$. The details of the bounding of each probability of error term is standard and detailed in~\cite[Chapter~10.3.]{el2011network}. 

Thus, using random binning and generalizing~\eqref{eq:Rate_ach} to the broadcast setting, $\lim_{n\to\infty} P_{\mathrm{e}}^{(n)}=0$ holds when
\begin{align}
\label{eq-com-rate-main-mutual-info}
 R^*(K,N)\overset{(a)}{\leq}\min_{\substack{P_{W\mid{\bf X}_{[N]}}:\\{\bf X}_{[N]}\in W,\;W\in\Gamma(G_{\mathrm B})}} I_q({\bf X}_{[N]} ;W) \overset{(b)}{=} H_\kappa\big(G_{\mathrm B},{\bf X}_{[N]}\big) 
\end{align}
where in $(a)$, the minimization is over all $P_{\mathbf X_{[N]},W}>0$ such that $W\in\Gamma(\GU)$, and the inequality follows from noting that $W\in \Gamma(G_{\mathrm B})$ is a function of ${\bf X}_{[N]}$ such that $H_q(f_i\mid W, \mathcal{S}_i)=0$ for all $i\in[K]$, ensuring $\lim_{n\to\infty} P_{\mathrm{e}}^{(n)}=0$, and $(b)$ follows from employing~\eqref{eq:graph_entropy}. The RHS in~\eqref{eq-com-rate-main-mutual-info} is achievable by assigning colors to 
$\GUn$, which achieves the entropy bound of the coloring variable~\cite{korner1971coding,feizi2014network}.
\end{proof}

\begin{cor}
\label{cor:side-information-assisted-graph-bound}
Let the single-letter side-information-assisted quantity be 
\begin{align}
\label{eq:single_letter_side_info_assisted_def_G_B}
H_{\kappa,\mathrm B}^{\mathrm{SI}}\left(G_{\mathrm B},\mathbf X_{[N]}\mid\mathcal S_{[K]}\right)\triangleq\min_{\substack{P_{W\mid{\bf X}_{[N]}}:\\\substack{W-{\bf X}_{[N]}-(\mathcal S_1,\ldots,\mathcal S_K),\\{\bf X}_{[N]}\in W\in\Gamma(G_{\mathrm B}) \ \text{a.s.}}}}\max_{i\in[K]} I_q({\bf X}_{[N]};W\mid\mathcal S_i) .
\end{align}
Under the same assumptions as in Theorem~\ref{theo-general-achievable-rate-via-broadcast}, a side-information-assisted independent-set coding scheme gives
\begin{align}
R^*(K,N)\leq H_{\kappa,\mathrm B}^{\mathrm{SI}}\left(G_{\mathrm B},\mathbf X_{[N]}\mid\mathcal S_{[K]}\right)\leq H_\kappa(G_{\mathrm B},{\bf X}_{[N]})  .
\end{align}
\end{cor}

Single-letter setting uses~\eqref{eq:single_letter_side_info_assisted_def_G_B}. Every realization of $W$ satisfies $W\in G_{\mathrm B}$. Applying this construction repeatedly treats two sequences as requiring separation whenever $\exists t\in[n],\exists i\in[K]:\ X_{i,t}=\widetilde X_{i,t},\ f_{i,t}\neq\widetilde f_{i,t}$, where the witnessing user $i\in[K]$ may change with $t$. However, the single-letter conditional independent-set optimization does not, in general, characterize the operational multiletter broadcast graph entropy, $R_0^*(K,N)
$. 
Operationally, two length-$n$ sequences must be separated only if one fixed user $i$ is confused over the entire block: $\exists i\in[K]:\ X_i^n=\widetilde X_i^n,\ f_i^n\neq\widetilde f_i^n$. We expect to see a multiletter gain because $G_{\mathrm B}^{(n)}\subsetneq G_{\mathrm B}^{\dot\vee n}$.

\begin{proof}[Proof of Corollary~\ref{cor:side-information-assisted-graph-bound}]
Fix a test channel $P_{W\mid{\bf X}_{[N]}}$ satisfying ${\bf X}_{[N]}\in W\in\Gamma(G_{\mathrm B})$ almost surely and $W-{\bf X}_{[N]}-(\mathcal S_1,\ldots,\mathcal S_K)$. 
Because every realization of $W$ is an independent set of $G_{\mathrm B}$, 
\begin{align}
\label{eq:function-from-W-Si}
H_q(f_i\mid W,\mathcal S_i)=0,
\qquad i\in[K] .
\end{align}
Indeed, if two source realizations in the same independent set had identical side information $\mathcal S_i$ but different values of $f_i$, they would be adjacent in $G_{f_i}$ and hence in $G_{\mathrm B}$.

Generate a random codebook containing $q^{n(I_q({\bf X}_{[N]};W)+\delta)}$ independently generated $W^n$ sequences. By the covering lemma, the encoder can find, with probability approaching one, a codeword jointly typical with ${\bf X}_{[N]}^n$. Randomly partition the codewords into $q^{nR}$ bins and transmit the bin index.

User~$i$ can recover the selected $W^n$ sequence from the received bin and its side-information sequence $\mathcal S_i^n$ whenever
\begin{align}
R>I_q({\bf X}_{[N]};W)-I_q(W;\mathcal S_i)\overset{(a)}{=}I_q({\bf X}_{[N]};W\mid\mathcal S_i) ,
\end{align}
where $(a)$ follows from $W-{\bf X}_{[N]}-\mathcal S_i$. Thus, the same bin index permits all users to recover $W^n$ provided that 
\begin{align}
R>\max_{i\in[K]}\ I_q({\bf X}_{[N]};W\mid\mathcal S_i) .
\end{align}
Each user then reconstructs its demanded sequence symbolwise using~\eqref{eq:function-from-W-Si}. Since $K$ is finite, the union of all decoding-error events has probability tending to zero. 
Minimizing over the admissible test channels gives
\begin{equation}
R^*(K,N) \leq H_{\kappa,\mathrm B}^{\mathrm{SI}}\left(G_{\mathrm B},\mathbf X_{[N]}\mid\mathcal S_{[K]}\right) .\qedhere
\end{equation}
\end{proof}

Since the $W$-based construction (Corollary~\ref{cor:side-information-assisted-graph-bound}) is one possible coding scheme,
\begin{align}
R^*(K,N)
\leq H_{\kappa,\mathrm B}^{\mathrm{SI}}\left(G_{\mathrm B},\mathbf X_{[N]}\mid\mathcal S_{[K]}\right) .
\end{align}

The side-information-assisted independent-set bound does not require the demands $\{f_i\}_{i\in[K]}$ to be compatible. Compatibility is required only for the stronger conclusion that this rate reduces to $\max\limits_{i\in[K]} H_q(f_i\mid\mathcal S_i)$, or that the resulting achievable bound matches the corresponding single-user converse.

When the compatibility condition (Definition~\ref{def:compatible_functions}) holds, and $\mathcal{S}_i=X_i$, for all $i\in[K]$, Theorem~\ref{theo-general-achievable-rate-via-broadcast} can be tightened as follows.

\begin{prop}
\label{prop_achievable-rate-via-broadcast_compatible}
Given a $(K, N,\{\mathcal{S}_i\}_{i\in[K]},\{f_i\}_{i\in[K]})$ NFCB model, where $K=N$ and $\mathcal{S}_i=X_i$, for all $i\in[K]$, define the side-information-assisted chromatic-entropy bound as
\begin{align}
\label{eq:side-information-chromatic-bound}
H_{\chi,\mathrm B}^{\mathrm{SI}}(G_{\mathrm B},\mathbf X_{[N]}\mid \mathcal S_{[K]}=\mathbf X_{[N]})\triangleq\min_{\substack{P_{W\mid{\bf X}_{[N]}}:\\{\bf X}_{[N]}\in W,\;W\overset{\circ}{\in}\Gamma(G_{\mathrm B})}}\max_{i\in[K]} I_q({\bf X}_{[N]};W\mid X_i) .
\end{align}
Then, the broadcast rate satisfies
\begin{align}
\label{eq:prop_achievable-rate-via-broadcast_compatible}
 H_\kappa(G_{\mathrm B},\mathbf X_{[N]})=R^*(K,N)=H_{\chi,\mathrm B}^{\mathrm{SI}}(G_{\mathrm B},\mathbf X_{[N]}\mid \mathcal S_{[K]}=\mathbf X_{[N]})=\max\limits_{i\in[K]}\, H_q(f_i\mid X_i) .
\end{align}
provided that the compatibility condition in Definition~\ref{def:compatible_functions} holds.
\end{prop}

\begin{proof}[Proof of Proposition~\ref{prop_achievable-rate-via-broadcast_compatible}]
The relation~\eqref{eq:side-information-chromatic-bound} implies  $H_q(W\mid{\bf X}_{[N]})=0$ because $W\overset{\circ}{\in}\Gamma(G_{\mathrm B})$, thus $W$ is a deterministic function of ${\bf X}_{[N]}$. It follows that
\begin{align}
\label{eq:side-information-chromatic-entropy}
H_{\chi,\mathrm B}^{\mathrm{SI}}(G_{\mathrm B},\mathbf X_{[N]}\mid \mathcal S_{[K]}=\mathbf X_{[N]})=\min_{\substack{P_{W\mid{\bf X}_{[N]}}:\\{\bf X}_{[N]}\in W,\;W\overset{\circ}{\in}\Gamma(G_{\mathrm B})}}\max_{i\in[K]}H_q(W\mid X_i) .
\end{align}

Suppose that $\{f_i\}_{i\in[K]}$ are compatible, i.e., $\exists$ a common function $f:\mathcal X_{[N]}\longrightarrow\mathcal F$ such that, for every $i\in[K]$ and every pair ${\bf x}_{[N]},\widetilde{\bf x}_{[N]}$ satisfying $x_i=\widetilde x_i$, we have $f({\bf x}_{[N]})=f(\widetilde{\bf x}_{[N]})\,\,\Longleftrightarrow\,\, f_i({\bf x}_{[N]})=f_i(\widetilde{\bf x}_{[N]})$.
Equivalently,
\begin{align}
\label{eq:compatible-mutual-recovery}
H_q(f_i\mid f,X_i)=H_q(f\mid f_i,X_i)=0 ,
\qquad i\in[K] .
\end{align}


Indeed, for every feasible coloring variable $W$, we have $H_q(f_i\mid W,X_i)=0$, $i\in[K]$. Therefore,
\begin{align}
H_q(W\mid X_i)\geq I_q(f_i;W\mid X_i)=H_q(f_i\mid X_i) ,
\end{align}
and hence
\begin{align}
\label{eq:compatible-lower-direction}
H_{\chi,\mathrm B}^{\mathrm{SI}}(G_{\mathrm B},\mathbf X_{[N]}\mid \mathcal S_{[K]}=\mathbf X_{[N]})
\geq
\max_{i\in[K]}H_q(f_i\mid X_i) .
\end{align}

Conversely, under~\eqref{eq:compatible-mutual-recovery}, the level sets
of $f$ are independent sets of $G_{\mathrm B}$. Thus, choosing the coloring variable as $W=f({\bf X}_{[N]})$ is feasible and gives
\begin{align}
\label{eq:compatible-upper-direction}
H_{\chi,\mathrm B}^{\mathrm{SI}}(G_{\mathrm B},\mathbf X_{[N]}\mid \mathcal S_{[K]}=\mathbf X_{[N]})\leq\max_{i\in[K]}H_q(f\mid X_i)=\max_{i\in[K]}H_q(f_i\mid X_i),
\end{align}
where the last equality follows from~\eqref{eq:compatible-mutual-recovery}. Combining~\eqref{eq:compatible-lower-direction} and~\eqref{eq:compatible-upper-direction} proves~\eqref{eq:prop_achievable-rate-via-broadcast_compatible}.
%
\end{proof}

We now remark on the security implications of our broadcast graph-based approach. 

\begin{remark}[Security implication]
In Theorem~\ref{theo-general-achievable-rate-via-broadcast}, the broadcast message is given by $M_n=en(W^n)$, where $W^n\in\Gamma(\GUn)$. If, for a given user $i$ with $\mathcal{S}_i$ and demanded function $f_i$, two distinct source realizations induce the same function value, then these realizations are statistically indistinguishable from user $i$' s perspective. If, however, the demanded function of user $i'$ (where $i'\neq i$) differs across these realizations, then $f_{i'}$ cannot be a deterministic function of $\{M,\mathcal{S}_i, f_i\}$. Hence, for any source distribution that assigns positive probability to both realizations, $H_q\left(f_{i'}\mid M,\mathcal{S}_i,f_i\right)>0$, with a strict inequality. This establishes that while the encoding guarantees 
recovery of each user’s own demand, it does not, in general, eliminate uncertainty about the demands of other users. In general, the broadcast condition $\lim\limits_{n\to\infty} P_{\mathrm{e}}^{(n)}=0$ has a security  implication because $H_q(f_{i'}^n\mid M_n,\mathcal{S}_i^n,f_i^n) =H_q(g_{i',n}(M_n,\mathcal S_{i'}^n)\mid M_n,\mathcal \mathcal{S}_i^n,f_i^n)>0$, for some function $g_{i',n}$, unless $\mathcal{S}_{i'}$ is derived from $\{M_n,\mathcal{S}_i^n,f_i^n\}$.

\end{remark}

\section{Converse Bounds}
\label{sec:converse}

In this section, we first present a multi-letter response-profile converse bound (Theorem~\ref{theo-lower-bound-mutual-info}) on the broadcast rate for the $(K, N,\{\mathcal{S}_i\}_{i\in[K]},\{f_i\}_{i\in[K]})$ NFCB scheme. By carefully capturing each user’s demand and side information, Theorem~\ref{theo-lower-bound-mutual-info} gives tighter converses than a simple side-information-based converse. 
Compared with the achievable rate in Theorem~\ref{theo-general-achievable-rate-via-broadcast}, which optimizes the coloring distribution over MISs of $\GU$, the converse bound in Theorem~\ref{theo-lower-bound-mutual-info} relies on a carefully designed variable to approximate the fundamental rate limits induced by $({\bf X}_{[N]}, \{\mathcal{S}_i\}_{i\in[K]},\{f_i\}_{i\in[K]})$. 
Using the notion of clique-entropy, which provides a necessary condition to describe the broadcast graph, we will subsequently detail two characteristic graph-based converses, namely the zero-error clique-entropy converse  (Theorem~\ref{thm:zero_error_clique_converse}), and the vanishing-error clique-entropy converse (Theorem~\ref{thm:vanishing_error_clique_converse}) using the length-$n$ operational block broadcast graph $G_{\mathrm B}^{(n)}$. 
We will then present an additional genie-aided lower bound (Proposition~\ref{prop-genie}
).



\begin{defi}[Admissible response profiles]
\label{def:admissible-response-profiles}
For $\epsilon\in[0,1)$, let $\mathcal U_n(\epsilon)$ denote the
collection of random tuples $\mathbf U_{[K]}^{(n)}=\bigl(U_1^{(n)},\ldots,U_K^{(n)}\bigr)$ satisfying the following conditions:
\begin{enumerate}
\item $\mathbf U_{[K]}^{(n)}$ is a deterministic function of
$\mathbf X_{[N]}^n$, i.e.,
\begin{align}
\label{eq:response-profile-source-function}
H_q\bigl(\mathbf U_{[K]}^{(n)}\mid\mathbf X_{[N]}^n\bigr)=0 .
\end{align}

\item For every $i\in[K]$, there exists a reconstruction function $\eta_{i,n}:\mathcal U_i^{(n)}\times\mathsf S_i^n\longrightarrow\mathcal F_i^n$ such that
\begin{align}
\label{eq:response-profile-reconstruction}
\mathbb{P}\left(\eta_{i,n}\bigl(U_i^{(n)},\mathcal S_i^n\bigr)\neq f_i^n\right)\leq\epsilon .
\end{align}

\end{enumerate}
\end{defi}

We begin with the multi-letter response-profile converse.

\begin{theo}[Multi-letter response-profile converse]
\label{theo-lower-bound-mutual-info}
For the $\bigl(K,N,\{\mathcal S_i\}_{i\in[K]},\{f_i\}_{i\in[K]}\bigr)$ NFCB problem, define
\begin{align}
\label{eq:multiletter-response-entropy-bound}
R_{\mathrm U}(K,N)&\triangleq\sup_{0<\epsilon<1}\limsup_{n\to\infty}\frac{1}{n}\inf_{\mathbf U_{[K]}^{(n)}\in\mathcal U_n(\epsilon)}H_q\bigl(\mathbf U_{[K]}^{(n)}\bigr) ,\\
\label{eq:multiletter-conditional-bound}
R_{\mathrm c}(K,N)&\triangleq\sup_{0<\epsilon<1}\limsup_{n\to\infty}\frac{1}{n}\inf_{\mathbf U_{[K]}^{(n)}\in\mathcal U_n(\epsilon)}\max_{i\in[K]}I_q\bigl(\mathbf X_{[N]}^n;\mathbf U_{[K]}^{(n)}\mid\mathcal S_i^n\bigr) .
\end{align}
Then,
\begin{align}
\label{eq:multiletter-converse-chain}
R^*(K,N)\geq R_{\mathrm U}(K,N)\geq R_{\mathrm c}(K,N)\geq\max_{i\in[K]}H_q(f_i\mid\mathcal S_i).
\end{align}
\end{theo}

\begin{proof}[Proof of Theorem~\ref{theo-lower-bound-mutual-info}]
Consider an arbitrary achievable sequence of deterministic
blocklength-$n$ codes. For each $n$, let $M_n=\psi_n(\mathbf X_{[N]}^n)$ be the common broadcast message, and let the reconstruction at user $i$ be
\begin{align}
\widehat f_i^n=g_{i,n}(M_n,\mathcal S_i^n),\qquad i\in[K] .
\end{align}
The block error probability $P_{\mathrm e}^{(n)}\longrightarrow0$ (cf.~\eqref{eq:block-error-probability}) for an achievable sequence of codes.

For every message realization $m\in\mathcal M_n$, define the decoder-response function of user $i$ as
\begin{align}
\label{eq:decoder-response-function}
u_{i,m}^{(n)}:\mathsf S_i^n \longrightarrow \mathcal F_i^n, \qquad u_{i,m}^{(n)}(s_i^n) \triangleq g_{i,n}(m,s_i^n) .
\end{align}
The random response profile induced by the broadcast message is
\begin{align}
\label{eq:random-response-profile}
U_i^{(n)}\triangleq u_{i,M_n}^{(n)} .
\end{align}
Thus, the joint response profile induced by the broadcast message is $\mathbf U_{[K]}^{(n)}\triangleq \bigl(U_1^{(n)},\ldots,U_K^{(n)}\bigr)$. 
The response profile $U_i^{(n)}$ does not contain the realized side-information sequence $S_i^n$. Rather, it specifies the reconstruction produced by user $i$ for every possible side-information sequence. Once $M_n=m$ is fixed, the function $u_{i,m}^{(n)}$ is fixed. Therefore, $H_q\bigl(U_i^{(n)}\mid M_n\bigr)=0$, $i\in[K]$, and, jointly,
\begin{align}
\label{eq:joint-response-profile-given-message}
H_q\bigl(\mathbf U_{[K]}^{(n)}
\mid M_n\bigr)=0.
\end{align}

Since $M_n=\psi_n(\mathbf X_{[N]}^n)$, the joint response profile is also a deterministic function of $\mathbf X_{[N]}^n$:
\begin{align}
\label{eq:joint-response-profile-given-source}
H_q\bigl(\mathbf U_{[K]}^{(n)}\mid\mathbf X_{[N]}^n\bigr)=0.
\end{align}

By construction, $U_i^{(n)}(S_i^n)=g_{i,n}(M_n,\mathcal S_i^n)=\widehat f_i^n$. Therefore,
\begin{align}
\label{eq:response-profile-error}
\mathbb{P}\left(U_i^{(n)}(\mathcal S_i^n)\neq f_i^n\right)\leq P_{\mathrm e}^{(n)}.
\end{align}
It follows from~\eqref{eq:joint-response-profile-given-source} and~\eqref{eq:response-profile-error} that 
$\mathbf U_{[K]}^{(n)}\in\mathcal U_n\bigl(P_{\mathrm e}^{(n)}\bigr)$. 
Moreover, Fano's inequality gives  
$H_q\bigl(f_i^n\mid U_i^{(n)},\mathcal S_i^n\bigr)\leq h_q\bigl(P_{\mathrm e}^{(n)}\bigr)+P_{\mathrm e}^{(n)}\log_q\bigl(|\mathcal F_i|^n-1\bigr)$ for every $i\in[K]$. 
Hence,
\begin{align}
\label{eq:asymptotic-response-profile-recovery}
\lim_{n\to\infty}\frac{1}{n}H_q\bigl(f_i^n\mid U_i^{(n)},S_i^n\bigr)=0,\qquad i\in[K].
\end{align}

We first lower-bound the message entropy using the joint response
profile. From~\eqref{eq:joint-response-profile-given-message},
\begin{align}
\label{eq:message-response-profile-entropy}
H_q(M_n)\geq I_q\bigl(\mathbf U_{[K]}^{(n)};M_n\bigr)
=H_q\bigl(\mathbf U_{[K]}^{(n)}\bigr)-H_q\bigl(\mathbf U_{[K]}^{(n)}\mid M_n\bigr)=H_q\bigl(\mathbf U_{[K]}^{(n)}
\bigr) .
\end{align}
Therefore,
\begin{align}
\label{eq:finite-block-response-entropy-bound}
R_n(K,N)\geq\frac{1}{n}H_q\bigl(\mathbf U_{[K]}^{(n)}\bigr).
\end{align}

Fix any $\epsilon\in(0,1)$. Since
$P_{\mathrm e}^{(n)}\to0$, for all sufficiently large $n$, $P_{\mathrm e}^{(n)}\leq\epsilon$. Therefore, the response profile induced by the code belongs to $\mathcal U_n(\epsilon)$, and
\begin{align}
\label{eq:response-entropy-infimum}
R_n(K,N)\geq \frac{1}{n}\inf_{\mathbf V_{[K]}^{(n)}\in\mathcal U_n(\epsilon)}H_q\bigl(\mathbf V_{[K]}^{(n)}\bigr).
\end{align}
Using $\limsup\limits_{n\to\infty}R_n(K,N)\leq R$ for every achievable rate $R$, we obtain
\begin{align}
R\geq\limsup_{n\to\infty}\frac{1}{n}\inf_{\mathbf V_{[K]}^{(n)}\in\mathcal U_n(\epsilon)}H_q\bigl(\mathbf V_{[K]}^{(n)}\bigr) .
\end{align}
Taking the supremum over $\epsilon\in(0,1)$ and then the infimum over all achievable rates proves
\begin{align}
\label{eq:optimal-rate-response-entropy}
R^*(K,N)\geq R_{\mathrm U}(K,N).
\end{align}

Since $\mathbf U_{[K]}^{(n)}$ is a deterministic function of $M_n$, for every $i\in[K]$,
\begin{align}
\label{eq:conditional-response-profile-bound}
H_q(M_n)\geq H_q(M_n\mid \mathcal S_i^n)\geq I_q\bigl(\mathbf U_{[K]}^{(n)};M_n\mid\mathcal S_i^n\bigr)=H_q\bigl(\mathbf U_{[K]}^{(n)}\mid\mathcal S_i^n\bigr) .
\end{align}
Furthermore, because $\mathbf U_{[K]}^{(n)}$ is a deterministic function of $\mathbf X_{[N]}^n$, it follows that 
$H_q\bigl(\mathbf U_{[K]}^{(n)}\mid\mathcal S_i^n\bigr)=I_q\bigl(\mathbf X_{[N]}^n;\mathbf U_{[K]}^{(n)}\mid\mathcal S_i^n\bigr)$. 
Thus, we derive the conditional mutual-information bound:
\begin{align}
\label{eq:finite-block-conditional-response-bound}
R_n(K,N)\geq\frac{1}{n}\max_{i\in[K]}I_q\bigl(\mathbf X_{[N]}^n;\mathbf U_{[K]}^{(n)}\mid\mathcal S_i^n\bigr).
\end{align}
Applying the same infimum and asymptotic argument used above
gives
\begin{align}
\label{eq:optimal-rate-conditional-response}
R^*(K,N)\geq R_{\mathrm c}(K,N) .
\end{align}

For every admissible $\mathbf U_{[K]}^{(n)}$ and every $i\in[K]$,
\begin{align}
H_q\bigl(\mathbf U_{[K]}^{(n)}
\bigr)\geq H_q\bigl(\mathbf U_{[K]}^{(n)}\mid\mathcal S_i^n\bigr)=I_q\bigl(\mathbf X_{[N]}^n;\mathbf U_{[K]}^{(n)}\mid\mathcal S_i^n\bigr) .
\end{align}
Therefore,
\begin{align}
\label{eq:response-bound-ordering}
R_{\mathrm U}(K,N)\geq R_{\mathrm c}(K,N) .
\end{align}

Finally, for any $i\in[K]$,
\begin{align}
\label{eq:single-user-functional-lower-bound}
I_q\bigl(\mathbf X_{[N]}^n;\mathbf U_{[K]}^{(n)}\mid\mathcal S_i^n\bigr)\geq I_q\bigl(\mathbf X_{[N]}^n; U_i^{(n)}\mid\mathcal S_i^n\bigr)&\geq I_q\bigl(f_i^n;U_i^{(n)}\mid \mathcal S_i^n\bigr)\nonumber\\
&=H_q(f_i^n\mid\mathcal S_i^n)-H_q\bigl(f_i^n\mid U_i^{(n)},\mathcal S_i^n\bigr) ,
\end{align}
where the second inequality follows from data processing because $f_i^n$ is a deterministic function of $\mathbf X_{[N]}^n$. Using the temporal i.i.d. assumption, $H_q(f_i^n\mid S_i^n)=nH_q(f_i\mid S_i)$, and using
\eqref{eq:asymptotic-response-profile-recovery}, we obtain
\begin{align}
\label{eq:hardest-user-lower-bound}
R_{\mathrm c}(K,N)\geq\max_{i\in[K]}H_q(f_i\mid S_i).
\end{align}
Combining~\eqref{eq:optimal-rate-response-entropy},~\eqref{eq:response-bound-ordering}, and~\eqref{eq:hardest-user-lower-bound} proves~\eqref{eq:multiletter-converse-chain}.
\end{proof}

Theorem~\ref{theo-lower-bound-mutual-info} showcases that the converse bounds in~\eqref{eq:multiletter-response-entropy-bound} and~\eqref{eq:multiletter-conditional-bound} are tighter than a simple side information-based converse bound, derived from $\{f_i\mid\mathcal{S}_i\}_{i\in[K]}$.

We next derive a graph-theoretic converse based on the clique entropy introduced by Alon and Orlitsky~\cite{alon1996source}. Unlike a product-graph characterization, the result below directly uses the blocklength-$n$ broadcast graph, which is neither an OR nor an AND power of its single-letter counterpart (Definition~\ref{def:blocklength_n_BC}).

We next define cliques, which we subsequently use to derive our converses (Theorem~\ref{thm:zero_error_clique_converse} and Theorem~\ref{thm:vanishing_error_clique_converse}) 

\begin{defi}[Clique~\cite{CarFloJor2008}]
\label{def:clique_graph_theory}
A clique in a graph $G=(\mathcal{V}(G),\mathcal{E}(G))$ is a collection of vertices, where every two are connected. We denote by $\Omega(G)$ the collection of cliques of $G$. 
\end{defi}

\begin{defi}[Clique entropy~\cite{alon1996source}]
\label{def:clique}
For a random variable $X$ distributed over $\mathcal V$ of $G$, the clique entropy is
\begin{align}
\label{eq:clique_entropy}
H_\omega(G,X)\triangleq\max_{\substack{P_{W'|X}:\\
X\in W'\in\Omega(G)}}H_q(X\mid W'),
\end{align}
which is given in base-$q$, where $W'$ is a random clique of $G$ containing the realized vertex $X$.
\end{defi}


We next derive clique entropy-based lower bounds to the broadcast rate, which rely on a necessary condition to describe the broadcast graph, and improves upon a simple converse leveraging the side information of each user. 
We first state the zero-error version of the converse. 
\begin{theo}[Zero-error clique-entropy converse]
\label{thm:zero_error_clique_converse}
For any zero-error computation-broadcast code,
\begin{align}
\label{eq:zero_error_clique_bound}
H_q(M_n)\geq H_\omega\left(G_{\mathrm B}^{(n)}, \mathbf X_{[N]}^n\right).
\end{align}
Thus, the optimal zero-error broadcast rate satisfies
\begin{align}
\label{eq:asymptotic_zero_error_clique_bound}
R_0^*(K,N)\geq\limsup_{n\to\infty}\frac{1}{n}H_\omega\left(G_{\mathrm B}^{(n)},\mathbf X_{[N]}^n\right) .
\end{align}
\end{theo}

\begin{proof}[Proof of Theorem~\ref{thm:zero_error_clique_converse}]
Consider an arbitrary blocklength-$n$ zero-error code with encoder
\begin{align}
M_n=\psi_n(\mathbf X_{[N]}^n).
\end{align}
We first show that $\psi_n$ is a proper coloring of $G_{\mathrm B}^{(n)}$.

Consider two adjacent vertices $\mathbf{x}_{[N]}^n$ and $\widetilde{\mathbf{x}}_{[N]}^n$ of $G_{\mathrm B}^{(n)}$. By Definition~\ref{eq:block_broadcast_graph}, there exists $i\in[K]$ such that
\begin{align}
s_i^n(\mathbf{x}_{[N]}^n)=s_i^n(\widetilde{\mathbf{x}}_{[N]}^n)\quad\text{and}\quad f_i^n(\mathbf{x}_{[N]}^n)\neq f_i^n(\widetilde{\mathbf{x}}_{[N]}^n) .
\end{align}
Suppose, by contradiction, that
\begin{align}
\psi_n(\mathbf{x}_{[N]}^n)=\psi_n(\widetilde{\mathbf{x}}_{[N]}^n) .
\end{align}
User $i$ would then receive the same message and possess the same side-information sequence under both source realizations. Its decoder would thus produce the same reconstruction in the two cases, contradicting the zero-error requirement. Hence, $\psi_n(\mathbf{x}_{[N]}^n)\neq\psi_n(\widetilde{\mathbf{x}}_{[N]}^n)$, and $\psi_n$ is a proper coloring of $G_{\mathrm B}^{(n)}$.

Now let $W'_n$ be an arbitrary random clique satisfying
\begin{align}
\mathbf X_{[N]}^n\in W'_n\in\Omega(G_{\mathrm B}^{(n)}) .
\end{align}
Because $\psi_n$ is a proper coloring and every pair of distinct vertices in $W'_n$ is adjacent, $\psi_n$ is one-to-one when restricted to $W'_n$. Hence, given $(M_n,W'_n)$, the source sequence $\mathbf X_{[N]}^n$ is uniquely determined:
\begin{align}
H_q(\mathbf X_{[N]}^n\mid M_n,W'_n)=0.
\end{align}
It follows that
\begin{align}
\label{eq:clique_entropy_intermediate}
H_q(\mathbf X_{[N]}^n\mid W'_n)=I_q(\mathbf X_{[N]}^n;M_n\mid W'_n)\leq H_q(M_n\mid W'_n)\leq H_q(M_n) .
\end{align}
Maximizing the left-hand side of~\eqref{eq:clique_entropy_intermediate} over all 
$P_{W'_n|\mathbf X_{[N]}^n}$ satisfying $\mathbf X_{[N]}^n\in W'_n\in\Omega(G_{\mathrm B}^{(n)})$ gives
\begin{align}
H_q(M_n)\geq H_\omega \left(G_{\mathrm B}^{(n)},\mathbf X_{[N]}^n \right) .
\end{align}
Dividing by $n$ and taking the asymptotic limit yields~\eqref{eq:asymptotic_zero_error_clique_bound}.
\end{proof}

\begin{remark}
For zero-error codes, the preimage of every message $M_n$ is an independent set of $G_{\mathrm B}^{(n)}$. Hence, $M_n$ is a proper coloring, and
\begin{align}
\label{eq:ordering_of_various_graph_quantities}
H_q(M_n)\geq H_\chi\left(G_{\mathrm B}^{(n)},{\bf X}_{[N]}^n\right)\geq H_\kappa\left(G_{\mathrm B}^{(n)},\mathbf X_{[N]}^n\right)\geq H_\omega\left(G_{\mathrm B}^{(n)},{\bf X}_{[N]}^n\right) 
\end{align}
which applies to a deterministic zero-error code (cf.~\eqref{eq:asymptotic-zero-error-broadcast-entropy}). 
K\"orner’s graph entropy in~\eqref{eq:ordering_of_various_graph_quantities} is an asymptotic quantity given by $H_\kappa\left(G_{\mathrm B}^{(n)},{\bf X}_{[N]}^n\right)= \lim\limits_{m\to\infty} \frac1m H_\chi\left((G_{\mathrm B}^{(n)})^{\dot\vee m},\mathbf (X_{[N]}^n)^m\right)$ (cf.~\eqref{eq:graph_entropy}), where $(G_{\mathrm B}^{(n)})^{\dot \vee m}\neq G_{\mathrm B}^{(nm)}$ in general. The OR power allows the user witnessing confusability to vary across the $m$ blocks, whereas $G_{\mathrm B}^{(nm)}$ requires one fixed user to be confused over the complete length-$nm$ block. 
However, in the asymptotically lossless setting, detailed below, the corresponding graph quantities must be evaluated on the high-probability correct-decoding subset.
\end{remark}

We next extend the preceding argument to asymptotically lossless
reconstruction. To this end, let $G_{\mathrm B}^{(n)}[\mathcal T_n]$ denote the subgraph of $G_{\mathrm B}^{(n)}$ induced by a set
$\mathcal T_n\subseteq\mathcal V^n\subseteq\mathcal X_{[N]}^n$. Define
\begin{align}
\label{eq:epsilon_clique_entropy}
L_n(\epsilon)
\triangleq
\inf_{\substack{\mathcal T_n\subseteq\mathcal V^n:\\
\mathbb{P}(\mathbf X_{[N]}^n\in\mathcal T_n)\geq1-\epsilon}}H_\omega\left(G_{\mathrm B}^{(n)}[\mathcal T_n],
P_{\mathbf X_{[N]}^n\mid\mathbf X_{[N]}^n\in\mathcal T_n}\right) .
\end{align}

\begin{theo}[Vanishing-error clique-entropy converse]
\label{thm:vanishing_error_clique_converse}
The optimal asymptotically lossless computation-broadcast rate satisfies
\begin{align}
\label{eq:vanishing_error_clique_bound}
R^*(K,N)\geq\sup_{0<\epsilon<1}(1-\epsilon)
\limsup_{n\to\infty} \frac{1}{n}L_n(\epsilon) .
\end{align}
In particular, whenever the corresponding limit as $\epsilon\downarrow0$ exists,
\begin{align}
\label{eq:vanishing_error_clique_bound_limit}
R^*(K,N)\geq\lim_{\epsilon\downarrow0}(1-\epsilon)\limsup_{n\to\infty}\frac{1}{n}L_n(\epsilon) .
\end{align}
%
\end{theo}

\begin{proof}[Proof of Theorem~\ref{thm:vanishing_error_clique_converse}]
Consider an arbitrary blocklength-$n$ code with encoder $\psi_n$ and decoders $\{g_{i,n}\}_{i\in[K]}$. Define its high-probability source set on which all users decode correctly as
\begin{align}
\label{eq:good_source_set}
\mathcal T_n \triangleq \bigl\{\mathbf{x}_{[N]}^n:\ g_{i,n}\bigl(\psi_n(\mathbf{x}_{[N]}^n),s_i^n(\mathbf{x}_{[N]}^n)\bigr)=f_i^n(\mathbf{x}_{[N]}^n),\quad\forall i\in[K]\bigr\}.
\end{align}
By the definition of the block error probability, $P_{\mathbf X_{[N]}^n}(\mathcal T_n)=1-P_{\mathrm e}^{(n)}$.

We claim that the restriction of $\psi_n$ to $\mathcal T_n$ is a proper coloring of the induced graph $G_{\mathrm B}^{(n)}[\mathcal T_n]$. To see this, consider adjacent vertices $\mathbf{x}_{[N]}^n,\widetilde{\mathbf{x}}_{[N]}^n\in\mathcal T_n$. There exists a user $i\in[K]$ such that
\begin{align}
s_i^n(\mathbf{x}_{[N]}^n)=s_i^n(\widetilde{\mathbf{x}}_{[N]}^n)\quad\text{and}\quad f_i^n(\mathbf{x}_{[N]}^n)\neq f_i^n(\widetilde{\mathbf{x}}_{[N]}^n).
\end{align}
If both source sequences produced the same message, user $i$ would receive identical decoder inputs in both cases and hence produce identical reconstructions. The decoder could not then be correct for both source sequences, contradicting their membership in $\mathcal T_n$. Therefore,
\begin{align}
\psi_n(\mathbf{x}_{[N]}^n)\neq\psi_n(\widetilde{\mathbf{x}}_{[N]}^n).
\end{align}

Let $E_n\triangleq\mathbbm{1}\{\mathbf X_{[N]}^n\in\mathcal T_n\}$ and let $p_n\triangleq\mathbb{P}(E_n=1)=1-P_{\mathrm e}^{(n)}$. Conditioned on $E_n=1$, the encoder is a proper coloring of $G_{\mathrm B}^{(n)}[\mathcal T_n]$. Applying Theorem~\ref{thm:zero_error_clique_converse} to the induced graph and the conditional source distribution gives
\begin{align}
\label{eq:conditional_good_set_clique_bound}
H_q(M_n\mid E_n=1)\geq H_\omega\left(G_{\mathrm B}^{(n)}[\mathcal T_n],\mathbf X_{[N]}^n
\,\middle|\,\mathbf X_{[N]}^n\in\mathcal T_n\right) .
\end{align}

Moreover,
\begin{align}
\label{eq:message_entropy_good_set}
H_q(M_n)\geq H_q(M_n\mid E_n)=p_n H_q(M_n\mid E_n=1)+(1-p_n)H_q(M_n\mid E_n=0)\geq p_n H_q(M_n\mid E_n=1).
\end{align}
Combining~\eqref{eq:conditional_good_set_clique_bound} and~\eqref{eq:message_entropy_good_set} yields
\begin{align}
\label{eq:finite_block_clique_converse}
H_q(M_n)\geq (1-P_{\mathrm e}^{(n)})H_\omega\left(G_{\mathrm B}^{(n)}[\mathcal T_n],\mathbf X_{[N]}^n\,\middle|\, \mathbf X_{[N]}^n\in\mathcal T_n\right)\geq (1-P_{\mathrm e}^{(n)})L_n(P_{\mathrm e}^{(n)}) .
\end{align}

Fix any $\epsilon\in(0,1)$. For any achievable sequence of codes, $P_{\mathrm e}^{(n)}\to0$, and hence $P_{\mathrm e}^{(n)}\leq\epsilon$ for all sufficiently large $n$. The corresponding set $\mathcal T_n$ then satisfies $P_{\mathbf X_{[N]}^n}(\mathcal T_n)\geq1-\epsilon$. 
Thus, for all sufficiently large $n$,
\begin{align}
\label{eq:fixed_epsilon_clique_bound}
R_n(K,N)=\frac1nH_q(M_n)\geq\frac{1-\epsilon}{n}L_n(\epsilon) .
\end{align}
Taking $\limsup\limits_{n\to\infty}$ and using $\limsup\limits_{n\to\infty}R_n(K,N)\leq R$ for every achievable rate $R$, we obtain for every $\epsilon\in(0,1)$
\begin{align}
R\geq(1-\epsilon)\limsup_{n\to\infty}\frac1nL_n(\epsilon) .
\end{align}
Taking the supremum over $\epsilon$ and then the infimum over all achievable rates proves~\eqref{eq:vanishing_error_clique_bound}.
\end{proof}

For $\epsilon\in(0,1/2)$, let $L_n(\epsilon)$ be as defined in~\eqref{eq:epsilon_clique_entropy}. Combining the clique-entropy formulation of~\cite{alon1996source} with a standard entropy decomposition over a high-probability event~\cite{CovTho1991} yields the following finite-blocklength smoothing bound. After normalization and letting $\epsilon\downarrow0$, the bound shows that restriction to high-probability induced subgraphs does not change the regularized clique-entropy converse.

\begin{lem}[Smoothing of clique entropy]
\label{lem:smoothing_lemma}
For each blocklength $n\in\mathbb N$, and $\epsilon\in(0,1/2)$,
\begin{align}
\label{eq:smoothing_bound}
H_\omega\left(G_{\mathrm B}^{(n)},\mathbf X_{[N]}^n\right)-h_q(\epsilon)-\epsilon n\log_q|\mathcal V(G_{\mathrm B})|\leq L_n(\epsilon)\leq H_\omega\left(G_{\mathrm B}^{(n)},\mathbf X_{[N]}^n\right) .
\end{align}
Thus,
\begin{align}
\label{eq:regularized-smoothing-equality}
\lim_{\epsilon\downarrow0}(1-\epsilon)\limsup_{n\to\infty}\frac{1}{n}L_n(\epsilon)=\limsup_{n\to\infty}\frac{1}{n}H_\omega\left(G_{\mathrm B}^{(n)},\mathbf X_{[N]}^n\right) .
\end{align}
\end{lem}
The equality concerns the regularized clique-entropy lower bounds only. In general $R^*(K,N)\leq R_0^*(K,N)$.

\begin{proof}[Proof of Lemma~\ref{lem:smoothing_lemma}]
The result follows from a standard entropy-continuity argument. For each blocklength $n\in\mathbb N$, define $b_n\triangleq H_\omega\left(G_{\mathrm B}^{(n)},\mathbf X_{[N]}^n\right)$. Let $Z_n$ be a clique-valued auxiliary random variable achieving $b_n$, and let $E_n=\mathbbm 1\{\mathbf X_{[N]}^n\in\mathcal T_n\}$. Conditioned on $E_n=1$, the random set $Z_n\cap\mathcal T_n$ remains a clique of the induced graph $G_{\mathrm B}^{(n)}[\mathcal T_n]$ and contains $\mathbf X_{[N]}^n$~\cite{alon1996source,CovTho1991}. Decomposing $H_q(\mathbf X_{[N]}^n\mid Z_n)$ according to $E_n$ and bounding the contribution of the event $E_n=0$ by $h_q(\epsilon)+\epsilon n\log_q|\mathcal V(G_{\mathrm B})|$ gives
\begin{align}
b_n-h_q(\epsilon)-\epsilon n\log_q|\mathcal V(G_{\mathrm B})|\leq L_n(\epsilon)\leq b_n.
\end{align}
The upper bound follows by choosing $\mathcal T_n=\mathcal V^n(G_{\mathrm B})$.
\end{proof}

For multi-user NFCB, due to~\eqref{nonequality_BroadcastBlockGraph_OR_AND}, the exact operational block graph $G_{\mathrm B}^{(n)}$ does not, in general, coincide with the length-$n$ OR power graph $G_{\mathrm B}^{\dot \vee n}$. Thus, the graph-entropy tensorization results do not apply here. Lemma~\ref{lem:smoothing_lemma} nevertheless eliminates the additional smoothing optimization. Since $L_n(\epsilon)$ is nonincreasing in $\epsilon$, it follows that
\begin{align}
\sup_{0<\epsilon<1/2}(1-\epsilon)\limsup_{n\to\infty}\frac{1}{n}L_n(\epsilon)=\lim_{\epsilon\downarrow0}(1-\epsilon)\limsup_{n\to\infty}\frac{1}{n}L_n(\epsilon)=\limsup_{n\to\infty}\frac{1}{n}H_\omega\left(G_{\mathrm B}^{(n)},\mathbf X_{[N]}^n\right) .
\end{align}
Thus, the asymptotically lossless and zero-error settings have the same regularized clique-entropy converse. However, the right-hand side remains a genuinely multiletter quantity. Therefore, obtaining a single-letter characterization of this regularized clique entropy remains an open problem.

We proceed to exploit the genie-aided approach introduced in~\cite{jafar2006capacity} and later used in~\cite{sun2018capacity} to provide a lower bound on the sum rate for a $(2,2,\{X_1, X_2\},\{f_1,f_2\})$ NFCB model. For this model, consider user~$i=1$ with side information $\mathcal{S}_1=X_1$ that needs $H_q(f_1\mid X_1)$ bits to recover $f_1$. Once user~$i=1$ recovers $f_1$, a genie provides all information known to user~$i=1$, i.e., $(f_1, \mathcal{S}_1)$, to user~$i=2$. User~$i=2$, with side information $\mathcal{S}_2$ and the additional information provided by the genie, then requires $H_q(f_2\mid \mathcal{S}_1,\mathcal{S}_2, f_1)$ additional bits to recover their requested function $f_2$. Hence, the sender node needs to broadcast extra $H_q(f_2\mid X_1, X_2, f_1)$ bits for user~$i=2$ to recover function $f_2$. Swapping the users and applying the same method, it is clear that without the help of a genie, we require a higher broadcast rate.   
As a result, the required sum rate is lower bounded as~\cite{jafar2006capacity}:
\begin{align}
 \label{eq-genie-2-normal}
 R^*(K,N)\geq \max [H_q(f_1\mid X_1) + H_q(f_2 \mid X_2, X_1, f_1),\,H_q(f_2\mid X_2) + H_q(f_1 \mid X_1, X_2, f_2)] \ .
\end{align}

We next generalize the genie-aided lower bound on the sum rate for the $N=K=2$ case in~\cite{jafar2006capacity}, \cite{sun2019capacity} to our $K$-user model and compare it to the lower bound in Theorem~\ref{theo-lower-bound-mutual-info}. This lower bound is given by the maximum over all such rates across possible permutations, as formalized by the following result.



\begin{prop}
\label{prop-genie}   
Let $(K,N,\{\mathcal{S}_i\}_{i\in[K]},\{f_i\}_{i\in[K]})$ be a NFCB model. For a given 
$\pi \in \mathrm{Perm}([K])$, we assume that a genie reveals to each user $\pi(i)$, for $i\geq 2$, the following:
\begin{itemize}
    \item the side information sets $\{\mathcal{S}_{\pi(i')}\}_{i' \in [i-1]}$, and  
    \item the demanded functions $\{f_{\pi(i')}\}_{ i' \in [i-1]}$ decoded by users $\pi(1), \dots, \pi(i-1)$.  
\end{itemize} 

Then, the broadcast rate is lower-bounded as
\begin{align}
\label{eq-genie-aided-rate}
R^*(K,N) \geq \max_{\pi \in \mathrm{Perm}([K])} \left [H_q(f_{\pi(1)}\mid \mathcal{S}_{\pi(1)})+\sum_{i=2}^{K} H_q\big( f_{\pi(i)} \mid \{f_{\pi(i')},\mathcal{S}_{\pi(i')}\}_{i' \in[i-1]},\ \mathcal{S}_{\pi(i)} \big)\right] \ .
\end{align}
\end{prop}

\begin{proof}[Proof of Proposition~\ref{prop-genie}]
Given a user permutation $\pi$, the broadcast rate lower bound in~\eqref{eq-genie-aided-rate} is achieved using a genie-aided scheme, where each user $\pi(i)\in[K]$ decodes $f_{\pi(i)}$ using their own side information $\mathcal{S}_{\pi(i)}$ and the data $\bigcup_{i' \leq i} \mathcal{S}_{\pi(i')}$ of all preceding users, according to the permutation $\pi$. 

With this genie-aided scheme, for a given $i\in[K]$, user $\pi(1)$ must receive at least $H_q( f_{\pi(1)} \mid \mathcal{S}_{\pi(1)})$ bits, and user $\pi(i)$, with $i\geq 2$, at least  $H_q\big( f_{\pi(i)} \mid \{f_{\pi(i')},\mathcal{S}_{\pi(i')}\}_{i' \in[i-1]},\ \mathcal{S}_{\pi(i)} \big)$ bits to recover their requested functions, respectively. The total broadcast rate of the sender given $\pi$ is lower bounded as
 \begin{align}
 \label{eq-genie-aided-rate-1}
     H_q(f_{\pi(1)}\mid \mathcal{S}_{\pi(1)})+\sum_{i=2}^{K} H_q\big( f_{\pi(i)} \mid \{f_{\pi(i')},\mathcal{S}_{\pi(i')}\}_{i' \in[i-1]},\ \mathcal{S}_{\pi(i)} \big) \ \mbox{ bits.}
 \end{align}
The final result follows from considering all possible permutations $\mathrm{Perm}([K])$.
\end{proof}

While Theorem~\ref{theo-lower-bound-mutual-info} gives a rate lower bound on the broadcast rate without assuming any interaction between the $K$ users, Proposition~\ref{prop-genie} is derived by incorporating such interactions.

\section{Additive Gap between the Achievable Rate and the Genie-Aided Bound}
\label{sec:optimality_gap}

Given a $(K,N,\{\mathcal{S}_i\}_{i\in[K]},\{f_i\}_{i\in[K]})$ NFCB model where $N\geq K\geq 2$, in this section, we establish an upper bound on the average additive gap between 
Corollary~\ref{cor:side-information-assisted-graph-bound} and Proposition~\ref{prop-genie}, denoted by $\Delta(K,N)$. 
For the most general setting, characterization of this gap is non-trivial. To this end, we characterize the gap, assuming $N=K$ and $\mathcal{S}_i=X_i$ for all $i\in[K]$. 
Before proceeding to bound $\Delta$, we focus on several special cases, for which there exist coding schemes (not restricted to Theorem~\ref{theo-general-achievable-rate-via-broadcast} or Corollary~\ref{cor:side-information-assisted-graph-bound}) that achieve the optimal broadcast rate $R^*(K,N)$ given in~\eqref{eq:comm_cost_optimal}. 
We determine the optimum broadcast computation rate for compatible functions or Boolean functions with $N=K=2$ (Proposition~\ref{prop-2-user-rate_optimality}). 
We subsequently compute the worst-case gap $\Delta_{\mathcal{N}}(K,N)$ and upper bound the average additive gap $\Delta(K,N)$ for $K=N=3$ and $K=N=4$ (Propositions~\ref{prop-3-user-gap-genie} and~\ref{prop-4-user-gap-genie}, respectively).

\paragraph{Linear computation broadcast with linear encoding}
\label{par:linear}
For a $(K,N,\{{\bf X}_{[N]}^{\intercal}{\bf V}'_i\}_{i\in[K]},\{{\bf X}_{[N]}^{\intercal}{\bf V}_i\}_{i\in[K]})$ linear computation broadcast scheme, the MinRank formulation in Section~\ref{sec:MinRank_ach} (cf.~\eqref{eq-min-rank-linear-demand}) is optimal under linear encoding/decoding, i.e., $\Delta(K,N)=0$ for all $K\geq 2$.

\paragraph{Recovering all datasets}
\label{par:data}
For a $(K,N,\{\mathcal{S}_i\}_{i\in[K]},\{f_i\}_{i\in[K]})$ computation broadcast scheme, where $\mathcal{S}_i\subset {\bf X}_{[N]}$, and each user $i\in[K]$ requests $f_i=(X_j \ : \ X_j\notin \mathcal{S}_i \ , \ \mathcal{S}_i\subset {\bf X}_{[N]})$, the minimum sum-rate~\eqref{eq:comm_cost_optimal} required for asymptotically lossless encoding of the dataset vector ${\bf X}_{[N]}$ for each user $i\in[K]$ to jointly decode ${\bf X}_{[N]}$ given its available side information $\mathcal{S}_i$ is achieved by the optimal coding scheme of Slepian and Wolf in~\cite{SlepWolf1973}, and the optimal broadcast rate accepts the following closed-form solution:
\begin{align}
\label{eq:RSW_definition}
R^*(K,N)=R_{\mathrm{SW}}(K,N)=\max\limits_{i\in[K]}\  H_q({\bf X}_{[N]}\mid \mathcal{S}_i) \ .
\end{align}

\paragraph{Complementary side information with $K=N=2$} 
\label{par:complementary}
For a $(2,2,\{\mathcal{S}_i\}_{i\in[2]},\{f_i\}_{i\in[2]})$ NFCB scheme with complementary side information, i.e., $\mathcal{S}_1=X_1$ and $\mathcal{S}_2=X_2$, we next demonstrate that the optimal broadcast rate (proven in~\cite[Theorem~3]{ravi2016broadcast}) can be achieved by our proposed scheme (see Theorem~\ref{theo-general-achievable-rate-via-broadcast}) when either (i) $f_1$, $f_2$ are compatible functions, or (ii) $X_1$, $X_2$ are binary random variables.

\begin{prop}
\label{prop-2-user-rate_optimality}
Given a $(2,2,\{X_1,X_2\},\{f_1,f_2\})$ NFCB model, the rate
\begin{align}
\label{eq-2-user-gap}
H_\kappa(G_{\mathrm B},\mathbf X_{[2]})=R^*(2,2)=\max\ (H_q(f_1\mid X_1),H_q(f_2\mid X_2))  \,\mbox{bits} 
\end{align} 
is optimum when $f_1$, $f_2$ are compatible functions or when $q=2$.
\end{prop}

\begin{proof}[Proof of Proposition~\ref{prop-2-user-rate_optimality}]
\label{proof_prop-2-user-rate_optimality}

For $N=K=2$, and  $\mathcal{S}_i=X_i$, $i\in\{1,2\}$, from Theorem~\ref{theo-lower-bound-mutual-info} and Proposition~\ref{prop-genie}, 
it holds 
\begin{align}
\label{eq:converse_K2}
R^*(2,2)\geq \max\ (H_q(f_1\mid X_1),H_q(f_2\mid X_2)) \ . 
\end{align}

We next consider the following two scenarios:

(i) When $f_1$, $f_2$ are compatible functions, following Definition~\ref{def:compatible_functions} (cf.~\cite{ravi2016broadcast}), there exists a deterministic function $f(X_1, X_2)$ such that $G^{ff}_{X_1X_2}=G^{f_1f_2}_{X_1X_2}$, implying that $H_q(f_1\mid X_1,f)=H_q(f_2\mid X_2,f)=0$. Hence, $H_q(f\mid X_i)=H_q(f_i \mid X_i)$ for $i\in[2]$. Using similar arguments, $H_q(f\mid X_1,f_1)=H_q(f\mid X_2,f_2)=0$. 
From Proposition~\ref{prop_achievable-rate-via-broadcast_compatible}, the following rate is achievable and optimal due to~\eqref{eq:converse_K2}:
\begin{align}
\label{eq:prop_achievable-rate-via-broadcast_compatible_K2}
H_\kappa(G_{\mathrm B},\mathbf X_{[2]})=\max \ (H_q(f_1\mid X_1),H_q(f_2\mid X_2))
\end{align}

(ii) When $q=2$, any $G_{X_1X_2}^{f_1f_2}$ is a subgraph of a graph with an edge set determined by $f_1(0,0)\neq f_1(0,1)$, $f_1(1,0)\neq f_1(1,1)$, $f_2(0,0)\neq f_2(1,0)$, and $f_2(0,1)\neq f_2(1,1)$, employing~\eqref{eq:vertex_set_rook} and~\eqref{eq:edge_set_rook} for $\mathcal{X}_1=\mathcal{X}_2=\mathbb{F}_2$, meaning that $G_{X_1X_2}^{f_1f_2}$ has at most $4$ edges. 
Similarly, any $G_{X_1X_2}^{ff}$ is a subgraph of a graph with an edge set determined by $f(0,0)\neq f(0,1)$, $f(1,0)\neq f(1,1)$, $f(0,0)\neq f(1,0)$, and $f(0,1)\neq f(1,1)$, meaning that $f(0,1)=f(1,0)$ and $f(0,0)=f(1,1)$. 
If there is no $f$ such that  $G_{X_1X_2}^{ff}= G_{X_1X_2}^{f_1f_2}$, then following the steps of proof of~\cite[Theorem~3]{ravi2016broadcast}, the only scenario where $f_1$, $f_2$ are not compatible is if graph $G^{f_1f_2}_{X_1X_2}$ has one edge. Any other possible 
$G^{f_1f_2}_{X_1X_2}$ is compatible. For those compatible graphs, the proof follows from part (i). For a graph with only one edge resulting from $f_1(0,0)\neq f_1(0,1)$, it is clear that $H_2(f_2 \mid X_2) = 0$, i.e., user $2$ can recover $f_2$ only from $X_2$. 
User $1$ requires an encoding rate $H_2(f_1\mid X_1)$. Thus, the rate $\max\ (H_2(f_1\mid X_1), H_2(f_2\mid X_2 )) = H_2(f_1\mid X_1)$ is achievable.
\end{proof}

After reviewing the special cases (cf.~Sections~\ref{par:linear}-
\ref{par:complementary}) where optimal coding schemes exist, we proceed to consider a general $(3,3,\{X_i\}_{i\in[3]},\{f_i\}_{i\in[3]})$ NFCB model with $q=2$, and upper bound $\Delta(3,3)$. 
\begin{prop}
\label{prop-3-user-gap-genie}
Given a $(3,3,\{X_i\}_{i\in[3]},\{f_i\}_{i\in[3]})$ NFCB model for general $\{f_1,f_2,f_3\}$ with $q=2$, the worst-case gap over all possible topologies is
\begin{align}
\label{eq-theo-3-user-gap}
\max\limits_{\mathcal{N}\in \mathfrak{N}}\,\Delta_{\mathcal{N}}(3,3)= 0.5  \,\mbox{bits} 
\end{align} 
where $\mathcal{N}$ and $\mathfrak{N}$ denote a specific topology and the class of all possible topologies, respectively.
The average gap over all possible topologies in Tables~\ref{tab:my_label-2}-
\ref{table-gap-3-diff}, with each topology being equally likely, is
\begin{align}
\Delta(3,3)=\mathbb{E}_{\mathfrak{N}}[\Delta_{\mathfrak{N}}(3,3)]=0.125 \,\mbox{bits}\ .
\end{align}
\end{prop}

\begin{proof}[Proof of Proposition~\ref{prop-3-user-gap-genie}]
\label{proof-theo-3-user-gap-genie}
Setting $\mathcal{S}_i=X_i$ for all $i\in[K]$ and using $q=2$, we note that
%
%
\begin{align}
\label{eq-proof-lem-3-gap-1}
I_2({\bf X}_{[N]};W\mid X_i)=H_2(W\mid X_i)=H_2(f_i\mid X_i)+H_2(W \mid f_i, X_i) 
\end{align}
where $W$ is the MIS induced by $P^*_{\mathbf X_{[N]},W}$, given in~\eqref{eq:Wopt}. 
Without loss of generality, let us assume that the maximum $H_{\kappa,\mathrm B}^{\mathrm{SI}}\left(G_{\mathrm B},\mathbf X_{[N]}\mid\mathcal S_{[K]}\right)$ (cf.~\eqref{eq:single_letter_side_info_assisted_def_G_B}) in~\eqref{eq-proof-lem-3-gap-1} is attained by user~$i=1$ with demand $f_1$ given $\mathcal{S}_1=X_1$. Thus,
\begin{align}
\label{eq-proof-lem-3-gap-2}
    H_{\kappa,\mathrm B}^{\mathrm{SI}}\left(G_{\mathrm B},\mathbf X_{[3]}\mid\mathcal S_{[3]}\right)=H_2(f_1\mid X_1)+H_2(W\mid f_1, X_1) \ . 
\end{align} 

From~\eqref{eq-genie-aided-rate}, we calculate the genie-aided lower bound as follows:
\begin{align}
\label{eq-proof-lem-3-user-gap-genie-2}  
R^*(3,3)
&\geq \max_{\substack{(i_1,i_2,i_3)\in \mathrm{Perm}([3])}}\ \sum\limits_{l=1}^3 \ H_2(f_{i_l}\mid X_{i_1},\cdots,X_{i_l},f_{i_1},\cdots f_{i_{l-1}})
\end{align}
where $H_2(f_{i_3}\mid X_{i_1},X_{i_2},X_{i_3},f_{i_1},f_{i_2})=0$ for any $(i_1,i_2,i_3)\in \mathrm{Perm}([3])$.

\input{plots/Gap_3_user}
Given the union graph $G_{\mathrm B} = G(\mathcal{V}, \mathcal{E})$ and assuming that the maximum in~\eqref{eq-proof-lem-3-gap-1} is attained by user~$i=1$, we next define three classes of edges, namely $\mathcal{E}_1$, $\mathcal{E}_2$, and $\mathcal{E}_3$, where 
\begin{align}
\mathcal{E}_1 &= \{({\bf x}_{[3]},{\bf x}_{[3]}') \in \mathcal{E} : ({\bf x}_{[3]},{\bf x}_{[3]}') \,\mbox{induced solely by } f_1\mid {X_1}\} , \\
\mathcal{E}_2 &= \{({\bf x}_{[3]},{\bf x}_{[3]}') \in \mathcal{E} : ({\bf x}_{[3]},{\bf x}_{[3]}') \: \mbox{induced by } f_2\mid {X_2} \mbox{ and } f_1\mid {X_1}\} ,\\
\mathcal{E}_3 &= \{({\bf x}_{[3]},{\bf x}_{[3]}') \in \mathcal{E} : ({\bf x}_{[3]},{\bf x}_{[3]}') \: \mbox{induced by } f_3\mid{X_3} \mbox{ and } f_1\mid{X_1} \}  
\end{align}
such that $\mathcal{E}=\bigcup_{i=1}^{3}\mathcal E_i$\footnote{Note that we do not need to consider the set of edges induced by all $f_1\mid X_1$, $f_2\mid X_2$, and $f_3\mid X_3$, as this set is empty.}, where ${\bf x}_{[3]}=(x_1,x_2,x_3)$ and ${\bf x}_{[3]}'=(x_1',x_2',x_3')$. 
More specifically,  $\mathcal{E}_1$ includes edges induced solely by $f_1\mid X_1$, and $\mathcal{E}_2$ by both $f_1\mid X_1$ and $f_2\mid X_2$, and $\mathcal{E}_3$ by both $f_1\mid X_1$ and $f_3\mid X_3$.  
The additive gap, denoted by $\Delta_{\mathcal{N}}(3,3)$, between~\eqref{eq-proof-lem-3-gap-2} and~\eqref{eq-proof-lem-3-user-gap-genie-2} for $3$-user NFCB is  
\begin{align}
\label{eq-proof-lem-3-user-gap-genie-3}
\Delta_{\mathcal{N}}(3,3)
&\overset{(a)}{\leq} H_2(f_1\mid X_1)+H_2(W\mid f_1, X_1)-\{H_2(f_1\mid X_1)+\max\,\{H_2(f_2\mid X_1,X_2,f_1), \: H_2(f_3\mid X_1,X_3,f_1)\}\}\nonumber\\
&\overset{}{=}H_2(W\mid f_1, X_1)-\max_{\substack{(i_2,i_3)\in \mathrm{Perm}([3]\setminus\{1\})}}\{H_2(f_{i_2}\mid X_1, X_{i_2},f_1) \} 
\end{align}
where $(a)$ follows from employing $i_1=1$ in~\eqref{eq-proof-lem-3-user-gap-genie-2}, and from~\eqref{eq-proof-lem-3-gap-2}. We illustrate this gap for different $\{f_i\}_{i\in [K]}$ in Tables~\ref{tab:my_label-2},~\ref{tab:my_label-3}, and~\ref{table-gap-3-diff}. For each distinct scenario, $\Delta_{\mathcal{N}}(3,3)$ is calculated by subtracting the first column, representing $G_{W \mid f_1, X_1}$, from the maximum of the second and third columns, which correspond to $G_{f_2 \mid X_1, X_2, f_1}$ and $G_{f_3 \mid X_1, X_3, f_1}$, respectively. 
We indicate distinct edge classes using distinct line styles for different source tuples ${\bf x}_{[3]}=(x_1,x_2,x_3)\in\mathbb{F}_2^3$. More specifically, $\mathcal{E}_1$ is absent (due to conditioning on $X_1$ and $f_1$), $\mathcal{E}_2$, and $\mathcal{E}_3$ are denoted by long-dashed lines and short-dashed lines, respectively.

In $G_{W \mid f_1, X_1}$, since we condition on $X_1$, there exist no cross edges between sets $\{(000),(001), (010),(011)\}$ and $\{(100),(101),(110),(111)\}$. 
For $X_1=0$, corresponding to the right half of $G_{W \mid f_1, X_1}$, the edges $\{(000),(011)\}$ and $\{(001),(010)\}$ do not appear in the graph, as they are induced solely by $f_1$, and the graph is constructed conditioned on knowing $f_1$. Similarly, for $X_1=1$, corresponding to the left half of $G_{W \mid f_1, X_1}$, edges $\{(100),(111)\}$ and $\{(101),(110)\}$ are also absent for the same reason. Thus, no complete subgraph is induced by the functions on either side of $G_{W\mid f_1, X_1}$.

To characterize $\Delta_{\mathcal{N}}(3,3)$ in~\eqref{eq-proof-lem-3-user-gap-genie-3}, consider a fixed value of $X_1$ and an arbitrary $f_1$. Without loss of generality, assume that the maximum of the two terms on the RHS of~\eqref{eq-proof-lem-3-user-gap-genie-3} is attained by $H_2(f_2\mid f_1, X_1, X_2)$. We then partition the dataset tuples $(X_1, X_2, X_3)$ into two subgraphs based on whether $X_1=0$ or $X_1=1$. 
In $G_{W \mid f_1, X_1}$, where each subgraph has $4$ vertices and at most $6$ edges ($K_4$ is the complete graph on four vertices). 
A subset of these edges, labeled $\mathcal{E}_1$, are removed due to knowledge of $(f_1, X_1)$, while the remaining edges, labeled $\mathcal{E}_2$ (long-dashed lines) and $\mathcal{E}_3$ (short-dashed lines), remain.

If $G_{f_2 \mid f_1, X_1, X_2}$ retains all edges in $\mathcal{E}_2$, shown in Tables~\ref{tab:my_label-2} (Top Row), then $H_2(W\mid f_1, X_1)=H_2(f_2 \mid f_1, X_1, X_2)\leq 1$, and thus $\Delta_{\mathcal{N}}(3,3) = 0$. Conversely, the maximum value of $\Delta_{\mathcal{N}}(3,3)$ is attained when $G_{f_2\mid X_1,X_2,f_1}$ contains only a subset of $\mathcal{E}_2$, e.g., either the subsets $\{(110),(111)\}$, $\{(101),(100)\}$, shown in Tables~\ref{tab:my_label-2} (Bottom Row), or from symmetry $\{(000),(001)\}$, $\{(011),(010)\}$, resulting in $\max\limits_{\mathcal{N}\in \mathfrak{N}}\,\Delta_{\mathcal{N}}(3,3)=0.5$ bits, hence $\Delta(3,3)\leq 0.5$ bits. 

The average gap over all possible topologies summarized in Tables~\ref{tab:my_label-2}, \ref{tab:my_label-3}, \ref{table-gap-3-diff} --- assuming that each topology is equally likely --- can be observed to be $\Delta(3,3)=0.125$ bits.
\end{proof}

We next consider a general $(4,4,\{X_i\}_{i\in[4]},\{f_i\}_{i\in[4]})$ NFCB model with $q=2$, and upper bound $\Delta(4,4)$. 
\begin{prop}
\label{prop-4-user-gap-genie}
Given a general $(4,4,\{X_i\}_{i\in[4]},\{f_i\}_{i\in[4]})$ NFCB model with $q=2$, the worst-case gap over all possible topologies (see Tables~\ref{tab:Subgraphs derived from Gf-4user}-\ref{tab:entropyvalues-4user} and Figures~\ref{fig:Gf1f2_Gf1f3_Gf1f4}-
\ref{fig:Gf1f3f4_Gf1f4f2_Gf1f4f3}) is
\begin{align}
\label{eq-theo-4-user-gap}
\max\limits_{\mathcal{N}\in \mathfrak{N}}\,\Delta_{\mathcal{N}}(4,4)\leq 1  \, \mbox{bit} \
\end{align} 
while the average gap over all possible topologies, with each topology being equally likely, is
\begin{align}
\mathbb{E}_{\mathfrak{N}}[\Delta_{\mathfrak{N}}(4,4)]\leq 0.25 \ \mbox{bits}\ .
\end{align}
\end{prop}

\begin{proof}[Proof of Proposition~\ref{prop-4-user-gap-genie}]
\label{proof-prop_nonlinear-k-user-communication-broadcast}  

Building on Proposition~\ref{prop-3-user-gap-genie}, we now derive the gap for the given $(4,4,\{\mathcal{S}_i\}_{i\in[4]},\{f_i\}_{i\in[4]})$ non-linear broadcast model. Setting $\mathcal{S}_i = X_i$ for all $i \in [K]$ and using $q=2$, we note that
\begin{align}
\label{eq-proof-lem-4-gap-1}
I_2({\bf X}_{[N]};W\mid \mathcal{S}_{i}) 
= H_2(W\mid X_i) 
= H_2(f_i\mid X_i) + H_2(W\mid f_i, X_i) 
\end{align}
where $W$ denotes the MIS induced by $P^*_{\mathbf X_{[N]},W}$ that achieves~\eqref{eq-theo-general-achievable-rate-via-broadcast}, and is defined in~\eqref{eq:Wopt}. 
Without loss of generality, assume that the maximum in~\eqref{eq-proof-lem-4-gap-1} is attained by user~$i=1$, with side information $\mathcal{S}_1 = X_1$ and demand $f_1$. Thus, adapting~\eqref{eq-proof-lem-3-gap-2} to the $N=K=4$ setting, we have
\begin{align}
    \label{eq-proof-lem-4-gap-2}
    H_{\kappa,\mathrm B}^{\mathrm{SI}}\left(G_{\mathrm B},\mathbf X_{[4]}\mid\mathcal S_{[4]}\right)= H_2(f_1\mid X_1) + H_2(W\mid f_1, X_1) \ .
\end{align}
From~\eqref{eq-genie-aided-rate}, we calculate the genie-aided lower bound as
\begin{align}
\label{eq-proof-lem-4-user-gap-genie-2}  
R^*(4,4) 
&\geq\max_{\substack{(i_1,i_2,i_3,i_4)\in \mathrm{Perm}([4])}}  \ \sum\limits_{l=1}^4 \ H_2(f_{i_l}\mid X_{i_1}, \cdots, X_{i_l}, f_{i_1}, \cdots, f_{i_{l-1}}) 
\end{align}
where $H_2(f_{i_4}\mid X_{i_1}, X_{i_2}, X_{i_3}, X_{i_4}, f_{i_1}, f_{i_2}, f_{i_3})=0$ for all distinct $(i_1,i_2,i_3,i_4)\in \mathrm{Perm}([4])$.  

The additive gap, denoted by $\Delta_{\mathcal{N}}(4,4)$, between~\eqref{eq-proof-lem-4-gap-2} and~\eqref{eq-proof-lem-4-user-gap-genie-2} is bounded as
\begin{align}
\label{eq-proof-lem-3-user-gap-genie-4}
\Delta_{\mathcal{N}}(4,4)\leq  H_2(W\mid f_1, X_1) - \!\!\!\!\max\limits_{\substack{(i_2,i_3,i_4)\in \mathrm{Perm}([4]\setminus\{1\})}}\{ H_2(f_{i_2}\mid X_1, X_{i_2}, f_1) + H_2(f_{i_3}\mid X_1, X_{i_2}, X_{i_3}, f_1, f_{i_2}) \}  \ .
\end{align}

As in Proposition~\ref{prop-3-user-gap-genie}, we introduce a set of graphs to compute $\Delta(4,4)$. Consider $G_{W\mid X_1, f_1}$, obtained from $G_{\mathrm B}$ by conditioning on $(X_1,f_1)$. Conditioning on $X_1$ removes all cross edges between $\{(0000),\dots,(0111)\}$ and $\{(1000),\dots,(1111)\}$, partitioning the graph into two halves corresponding to $X_1=0$ and $X_1=1$. 
For $X_1=0$, the edges $\{(0000),(0111)\},\{(0001),(0110)\},\{(0010),(0101)\},\{(0011,0100)\}$—induced solely by $f_1$—are absent; similarly, for $X_1=1$, 
$\{(1000),(1111)\},\{(1001),(1110)\},\{(1010),(1101)\},\{(1011),(1100)\}$ are removed. Thus, no complete subgraph is induced on either half of $G_{W\mid X_1, f_1}$. The conditional entropy of $G_{W\mid X_1, f_1}$ is $H_2(W\mid X_1, f_1)$. 


Let $\mathcal{E}_W$ denote the set of all edges in $G_{W\mid X_1, f_1}$. We classify these edges into seven groups:
\begin{align}
\mathcal{E}_1 &= \{({\bf x}_{[4]},{\bf x}_{[4]}') \in \mathcal{E}_W : \mbox{induced by } f_1\mid X_1\} ,\\
\mathcal{E}_l &= \{({\bf x}_{[4]},{\bf x}_{[4]}') \in \mathcal{E}_W : \mbox{induced by } f_1\mid X_1 \mbox{ and } f_l\mid X_l\} ,\quad l\in\{2,3,4\} ,\\
\mathcal{E}_{l+l'} &= \{({\bf x}_{[4]},{\bf x}_{[4]}') \in \mathcal{E}_W : \mbox{induced by } f_1\mid X_1, f_l\mid X_l, \mbox{ and } f_{l'}\mid X_{l'} \} ,\quad l,l'\in\{2,3,4\}\,\mbox{and}\, l\neq l'
\end{align}
where ${\bf x}_{[4]}=(x_1,x_2,x_3,x_4)$ and ${\bf x}_{[4]}'=(x_1',x_2',x_3',x_4')$, and $\mathcal{E}_{W} = \cup_{i=1}^7 \mathcal{E}_i$.

Now, we consider the remaining functions, namely $f_2$, $f_3$ and $f_4$, and derive a collection of graphs and their entropies by conditioning on side information, as summarized in Table~\ref{tab:Subgraphs derived from Gf-4user}.

\begin{table}[h]\small
    \centering
    \begin{tabular}{|c|c|c|c|c|}
    \hline
       & Function  & Side information & Conditional entropy & Edge class\\
         \hline
         $G_{f_2\mid X_1,X_2,f_1}$&$f_2$& $X_1,\ X_2,\ f_1$& $H_2(f_2\mid X_1,X_2,f_1)$& $\mathcal{E}_2$\\
         \hline
         $G_{f_2\mid X_1, X_2,X_3, f_1,f_3}$&$f_2$& $X_1,\ X_2,\ X_3,\ f_1,\ f_3$& $H_2(f_2\mid X_1,X_2,X_3,f_1,f_3)$& $\mathcal{E}_5$\\
         \hline
         $G_{f_2\mid X_1, X_2,X_4, f_1,f_4}$&$f_2$& $X_1, \ X_2,\ X_4,\ f_1,\ f_4$& $H_2(f_2\mid X_1,X_2,X_4,f_1,f_4)$& $\mathcal{E}_6$\\
         \hline
    \hline
         $G_{f_3\mid X_1,X_3,f_1}$&$f_3$& $X_1,\ X_3,\ f_1$& $H_2(f_3\mid X_1,X_3,f_1)$& $\mathcal{E}_3$\\
         \hline
         $G_{f_3\mid X_1, X_2,X_3, f_1,f_2}$&$f_3$& $X_1,\ X_2,\ X_3,\ f_1,\ f_2$& $H_2(f_3\mid X_1,X_2,X_3,f_1,f_2)$& $\mathcal{E}_5$\\
         \hline
         $G_{f_3\mid X_1, X_3,X_4, f_1,f_4}$&$f_3$& $X_1,\ X_3,\ X_4,\ f_1,\ f_4$& $H_2(f_3\mid X_1,X_3,X_4,f_1,f_4)$& $\mathcal{E}_7$\\
         \hline
         \hline
         $G_{f_4\mid X_1,X_4,f_1}$&$f_4$& $X_1,\ X_4,\ f_1$& $H_2(f_4\mid X_1,X_4,f_1)$& $\mathcal{E}_4$\\
         \hline
         $G_{f_4\mid X_1, X_3,X_4, f_1,f_3}$&$f_4$& $X_1,\ X_3,\ X_4,\ f_1,\ f_3$& $H_2(f_4\mid X_1,X_3,X_4,f_1,f_3)$& $\mathcal{E}_7$\\
         \hline
         $G_{f_4\mid X_1, X_2,X_4, f_1,f_2}$&$f_4$& $X_1,\ X_2,\ X_4,\ f_1,\ f_4$& $H_2(f_4\mid X_1,X_2,X_4,f_1,f_2)$& $\mathcal{E}_6$\\
         \hline
    \end{tabular}
    \caption{\small{The $4$-user NFCB problem, corresponding to the worst-case configuration  depicted in Figure~\ref{fig:GWX1f1}, which includes the graphs for $f_2$, $f_3$, and $f_4$, each conditioned on their side information, along with their conditional entropies.}}
    \label{tab:Subgraphs derived from Gf-4user}
\end{table}

\input{plots/Gap_4_users}


Since constructing a table analogous to Tables~\ref{tab:my_label-2}- 
\ref{table-gap-3-diff} is tedious, 
we restrict our attention to the worst-case scenario. A worst-case instance of the $4$-user model is illustrated in Figure~\ref{fig:GWX1f1}. In this graph, the edges in $\mathcal{E}_1$ are removed due to knowledge of $(f_1, X_1)$, $\mathcal{E}_2$ is denoted by black solid lines, $\mathcal{E}_3$ by gray dotted lines, $\mathcal{E}_4$ by gray solid lines, $\mathcal{E}_5$ by black dashed lines, $\mathcal{E}_6$ by gray long-dashed lines, and $\mathcal{E}_7$ is denoted by black long-dashed lines. To characterize $\Delta(4,4)$, we construct a collection of graphs, obtained by conditioning on their respective side information, as shown in Figures~\ref{fig:Gf1f2_Gf1f3_Gf1f4},~\ref{fig:Gf1f2f3_Gf1f2f4_Gf1f3f2}, and~\ref{fig:Gf1f3f4_Gf1f4f2_Gf1f4f3}. The conditional entropies of these graphs are summarized in Table~\ref{tab:entropyvalues-4user}. 
We observe in Figure~\ref{fig:Gf1f2_Gf1f3_Gf1f4} (Left) that are no edges between $(x_1, x_2, x_3, x_4)\in \{0,1\}^4$ and $(x_1, x_2, x_3', x_4)\in \{0,1\}^4$ for $x_3\neq x_3'$. Hence, $H_2(f_2 \mid X_1, X_2, X_4,f_1)= 0$, and $G_{f_2\mid X_1,X_2,X_4,f_1,f_4}$ does not contain any edges. Similarly, in Figure~\ref{fig:Gf1f2_Gf1f3_Gf1f4} (Middle), there are no edges between $(x_1, x_2, x_3, x_4)\in \{0,1\}^4$ and $(x_1, x_2', x_3, x_4)\in \{0,1\}^4$. Hence, $H_2(f_3 \mid X_1, X_3, X_4,f_1) = 0$, and $G_{f_3\mid X_1,X_3,X_4,f_1,f_4}$ does not contain any edges. These observations yield
\begin{align}
    \label{eq:same_graphs}
    G_{f_4\mid X_1,X_4,f_1}\cup G_{f_2\mid X_1,X_2,X_4,f_1,f_4}=G_{f_4\mid X_1,X_4,f_1}\cup G_{f_3\mid X_1,X_3,X_4,f_1,f_4}=G_{f_4\mid X_1,X_4,f_1}
\end{align}
which can also be verified from Figures~\ref{fig:Gf1f2_Gf1f3_Gf1f4} (Right), \ref{fig:Gf1f3f4_Gf1f4f2_Gf1f4f3} (Middle), and \ref{fig:Gf1f3f4_Gf1f4f2_Gf1f4f3} (Right). In this case, we obtain
\begin{align}
\label{eq:W_cond_X1_f1}
H_2(W\mid X_1,f_1) \leq 2 \ ,\\
\label{eq:fj_cond}
H_2(f_{i_2}\mid X_1, X_{i_2}, f_1) + H_2(f_{i_3}\mid X_1, X_{i_2}, X_{i_3}, f_1, f_{i_2}) = 1 
\end{align}
for $i_2,i_3 \in [4]\setminus\{1\}$ and $i_2\ne i_3$. Plugging~\eqref{eq:W_cond_X1_f1} and~\eqref{eq:fj_cond} into~\eqref{eq-proof-lem-3-user-gap-genie-4}, we obtain $\max\limits_{\mathcal{N}\in \mathfrak{N}}\,\Delta_{\mathcal{N}}(4,4)\leq 1$ bit. Furthermore, the average gap over all possible topologies summarized in Tables~\ref{tab:Subgraphs derived from Gf-4user}-\ref{tab:entropyvalues-4user} --- assuming that each topology is equally likely and the gap for each half graph is upper bounded by $\Delta(3,3)$ --- satisfies $\Delta(4,4)\leq 2\Delta(3,3)=0.25$ bits.
\end{proof}

\begin{table}[]\small
    \centering
    \begin{tabular}{|c|c|}
    \hline
    Graph& Entropy expressions\\
    \hline
       $G_{W\mid X_1,f_1}$  &$H_2(W\mid X_1,f_1)$= 2  \\
       \hline
       $G_{f_2\mid X_1,X_2,f_1}$& $H_2(f_2\mid X_1,X_2, f_1)=\frac{1}{2}$\\
       \hline
       $G_{f_3\mid X_1,X_3,f_1}$& $H_2(f_3\mid X_1,X_3, f_1)=\frac{1}{2}$\\
       \hline
       $G_{f_4\mid X_1,X_4,f_1}$& $H_2(f_4\mid X_1,X_4, f_1)=1$\\
       \hline
         $G_{f_2\mid X_1,X_2,f_1}\cup G_{f_3\mid X_1,X_2,X_3,f_1,f_2}$  &$H_2(f_2\mid X_1,X_2,f_1)+H_2(f_3\mid X_1,X_2,X_3,f_1,f_2)$= 1  \\
       \hline
       $G_{f_2\mid X_1,X_2,f_1}\cup G_{f_4\mid X_1,X_2,X_4,f_1,f_2}$  &$H_2(f_2\mid X_1,X_2,f_1)+H_2(f_4\mid X_1,X_2,X_4,f_1,f_2)$= 1  \\
       \hline
       $G_{f_3\mid X_1,X_3,f_1}\cup G_{f_2\mid X_1,X_2,X_3,f_1,f_3}$  &$H_2(f_3\mid X_1,X_3,f_1)+H_2(f_2\mid X_1,X_2,X_3,f_1,f_3)$= 1  \\
       \hline
       $G_{f_3\mid X_1,X_3,f_1}\cup G_{f_4\mid X_1,X_3,X_4,f_1,f_3}$  &$H_2(f_3\mid X_1,X_3,f_1)+H_2(f_4\mid X_1,X_3,X_4,f_1,f_3)$= 1  \\
       \hline
       $G_{f_4\mid X_1,X_4,f_1}\cup G_{f_2\mid X_1,X_2,X_4,f_1,f_4}$  &$H_2(f_4\mid X_1,X_4,f_1)+H_2(f_2\mid X_1,X_2,X_4,f_1,f_4)$= 1  \\
       \hline
       $G_{f_4\mid X_1,X_4,f_1}\cup G_{f_3\mid X_1,X_3,X_4,f_1,f_4}$  &$H_2(f_4\mid X_1,X_4,f_1)+H_2(f_3\mid X_1,X_3,X_4,f_1,f_4)$= 1  \\
       \hline
    \end{tabular}
    \caption{\small Subgraphs and their conditional entropies  for the $4$-user NFCB problem (Figures~\ref{fig:Gf1f2_Gf1f3_Gf1f4},~\ref{fig:Gf1f2f3_Gf1f2f4_Gf1f3f2}, and~\ref{fig:Gf1f3f4_Gf1f4f2_Gf1f4f3}), corresponding to the worst-case configuration  illustrated in Figure~\ref{fig:GWX1f1}.}    \label{tab:entropyvalues-4user}
\end{table}

While our converses are proven tight for $K=N=3$, with $\Delta(3,3)=0.125$ bits (Proposition~\ref{prop-3-user-gap-genie}), and $K=N=4$, with $\Delta(4,4)\leq 0.25$ bit (Proposition~\ref{prop-4-user-gap-genie}), a challenging open problem is determining the optimality gap $\Delta(K,N)$ for general $(K,N)$ with general $\mathcal{S}_i$, $i\in[K]$ and for $q\geq 2$.

\section{Examples on \texorpdfstring{$(K=3,N=3)$}{(K=3,N=3)}-Computation Broadcast}
\label{sec:3-user-comp-broad}

In this section, we focus on $(K=3, N=3)$-NFCB, where $\mathbf X_{[3]}\sim\mathrm{Unif}(\mathbb F_q^3)$, and each user $i\in[3]$ has side information $\mathcal{S}_i={\bf X}_{\mathcal A_i}$ and requests a distinct function $f_i=f_i(\mathbf X_{[N]})$. 
We contrast the coding schemes in Section~\ref{sec:achievability} and the converse bounds in Section~\ref{sec:converse} through various examples. 
Examples~\ref{ex-nonlinear-comp},~\ref{ex:nonlinear-outperforms-linear}, and~\ref{ex-nonlinear-comp_self-complementary_graph} examine Boolean (non-linear) function demands, while Example~\ref{ex-linear-comp-jafar} focuses on linear function demands. 
Example~\ref{ex:nonlinear-outperforms-linear} demonstrates the advantage of non-linear encoding over scalar and vector linear encoding (Proposition~\ref{prop:advantage_nonlinear}), and upper bounds the zero-error and asymptotically lossless operational broadcast entropies by the performance of linear encoding (Proposition~\ref{prop:nonlinear-graph-separation}).

The next example considers Boolean (non-linear) function demands over $\mathbb{F}_2$.

\begin{ex}[{\bf \texorpdfstring{A $(3,3,\{X_i\}_{i\in[3]},\{f_i\}_{i\in[3]})$}{} Boolean computation broadcast model}]
\label{ex-nonlinear-comp}
Consider a NFCB model $(3,3,\{X_i\}_{i\in[3]},\{f_i\}_{i\in[3]})$, where elements of ${\bf X}_{[3]}\sim\mathrm{Unif}(\mathbb{F}_2^3)$, and the desired functions are distinct and given as
\begin{align}
\label{eq:desired_functions_example1}
f_1=X_1\wedge X_2\wedge X_3 \ ,\quad
f_2=X_1\vee X_2\vee X_3 \ ,\quad  f_3=X_1\wedge(X_2\vee X_3) \ .
\end{align}

Below, we present three approaches to design $M$: 

{\bf i) Slepian-Wolf coding approach.} 
Employing the Slepian-Wolf theorem~\cite{SlepWolf1973}, the minimum broadcast rate that allows for asymptotically lossless reconstruction of $\{X_j\}_{j\in[3]}$ at all $K=3$ users is given as 
\begin{align}
\label{eq:rate_SW_ex1}
R_{\mathrm{SW}}(3,3)=\max\limits_{i\in[3]}\,  H_2({\bf X}_{[3]\backslash i})=2\ \mbox{bits} 
\end{align}
which can be achieved through transmissions $X_1\oplus_2 X_2$ and $X_2\oplus_2 X_3$. Each user can recover $\{X_j\}_{j\in[3]}$ using $X_1\oplus_2 X_2$ and $X_2\oplus_2 X_3$ as well as their side information.

{\bf ii) Broadcast graph approach.} 
We next employ the broadcast graph approach (cf. Definition~\ref{def:broadcast_graph}), where $G_{f_1}=(\mathbb{F}_2^3,\mathcal{E}_1)$ given $\mathcal{S}_1=X_1$ is shown in Figure~\ref{fig:Graphs_ex1}-$(\mathrm{a})$. In $G_{f_1}$, we do not need to distinguish the vertices with $X_1=0$ because the outcome of $f_1 =0$ for any $(x_2,x_3)\in\mathbb{F}_2^2$. When $X_1=1$, $f_1 =1$ if and only if $X_2=X_3=1$, and otherwise $f_1=0$. Thus, $\mathcal{E}_1=\{(111,100), (111,101),(111,110)\}$. Hence, 
\begin{align}
\label{eq:example1_LB1}
H_\chi(G_{f_1},{\bf X}_{[N]}\mid X_1)\triangleq\min\limits_{\substack{P_{U_1\mid{\bf X}_{[3]}}:\\{\bf X}_{[3]}\in U_1\overset{\circ}{\in}\Gamma(G_{f_1})}}\, I_2({\bf X}_{[3]};U_1\mid X_1)=H_2(U_1\mid X_1)
=H_2(f_1\mid X_1)=0.406 \ \mbox{bits}
\end{align}
noting that $U_1\in\{\{111\},\mathcal{X}_{[3]}\backslash\{111\}\}$, and $U_1\overset{\circ}{\in}\Gamma(G_{f_1})$~\cite{alon1996source}, i.e., $H_2(U_1\mid {\bf X}_{[3]})=0$. 

Similarly, we construct $G_{f_2}$ and $G_{f_3}$, as shown in Figure~\ref{fig:Graphs_ex1}-$(\mathrm{b})$ and $(\mathrm{c})$, respectively. Furthermore, 
\begin{align}
\label{eq:example1_LB2}
H_\chi(G_{f_2},{\bf X}_{[N]}\mid X_2)=H_2(U_2\mid X_2)
=H_2(f_2\mid X_2)=0.406 \ \mbox{bits}
\end{align}
noting that $U_2\in\{\{000\},\mathcal{X}_{[3]}\backslash\{000\}\}$. 
\begin{align}
\label{eq:example1_LB3}
H_\chi(G_{f_3},{\bf X}_{[N]}\mid X_3)=H_2(U_3\mid X_3)
=H_2(f_3\mid X_3)=0.906 \ \mbox{bits}
\end{align}
noting that $U_3\in\{\{000,001,010,011,100\},\{101,110,111\}\}$. 

A universal lower bound, using~\eqref{eq:example1_LB1},~\eqref{eq:example1_LB2}, and~\eqref{eq:example1_LB3}, is given by 
\begin{align}
\label{eq:example1_lowerbound}
R^*(3,3)\geq \max\limits_{i\in[3]}\, H_\chi(G_{f_i},{\bf X}_{[N]}\mid X_i) =H_2(f_3\mid X_3)=0.906 \ \mbox{bits}
\end{align}
where the inequality can be strict as functions are not compatible (see Theorem~\ref{theo-lower-bound-mutual-info}). 

The broadcast graph, $G_{\mathrm B}$, which captures $(\{\mathcal{S}_i\}_{i\in[3]},\ \{f_i\}_{i\in[3]})$, is shown in Figure~\ref{fig:Graphs_ex1}~$(\mathrm{d})$. Its independence number is $\alpha(G_{\mathrm B})=4$, and it has the following set of MISs:
 \begin{align}
 \label{Gamma-3-user}
\Gamma(G_{\mathrm B})=\left\{\begin{array}{c}
       W_1=\{101,110\} \ ,
       W_2=\{000,010,011\} \ ,\ W_3=\{000,010,111\} \ , \, \\
       W_4=\{001,011,110\}\ ,\ W_5=\{ 010,100,101\} \ ,
       W_6=\{001,010,011,100\} 
  \end{array}\right\} \ .
  \end{align}

The achievable broadcast rate, employing~\eqref{eq-theo-general-achievable-rate-via-broadcast} from Theorem~\ref{theo-general-achievable-rate-via-broadcast}, is given by 
\begin{align}
\label{eq:mutual_info_ex1}
H_\kappa(G_{\mathrm B},\mathbf X_{[3]})=I_2({\bf X}_{[3]};W)=1.417 \ \mbox{bits}
\end{align}
where the conditional distribution $P^*_{W\mid \mathbf X_{[3]}}$ that achieves the optimum~\eqref{eq:mutual_info_ex1} is illustrated in Table~\ref{tab:optimal-joint-distribution}.

\begin{table}[t!]\footnotesize
 \setlength{\extrarowheight}{1pt}
\setlength{\tabcolsep}{2pt}
    \begin{center}
          \begin{tabular}
          {|c|c|c|c|c|c|c|c|c|}
            \hline&&&&&&&&\\[-1em]
            $W\,\backslash\, {\bf X}_{[3]}$&$000$&$001$&$010$&$011$&$100$&$101$&$110$&$111$\\
             \hline &&&&&&&&\\[-1em]
             $W_1=\{101,110\}$&$0$&$0$&$0$&$0$&$0$&$0.375$&$0.323$&$0$\\[0.1em]
            \hline&&&&&&&&\\[-1em]
             $W_2=\left\{000,010,011\right\}$&$0$&$0$&$0$&$0$&$0$&$0$&$0$&$0$  \\[0.1em]
             \hline&&&&&&&&\\[-1em]
             $W_3=\left\{
                 000,010,111 \right\}$&$1$&$0$&$0.413$&$0$&$0$&$0$&$0$&$1$\\[0.1em]
             \hline&&&&&&&&\\[-1em]
             $W_4=\left\{001,011,110\right\}$&$0$&$0.393$&$0$&$0.393$&$0$&$0$&$0.676$&$0$  \\[0.1em]
             \hline&&&&&&&&\\[-1em]
             $W_5=\left\{
                 010,100,101 \right\}$&$0$&$0$&$0.199$&$0$&$0.34$&$0.625$&$0$&$0$\\[0.1em]
             \hline&&&&&&&&\\[-1em]
             $W_6=\left\{001,010,011,100\right\}$&$0$&$0.607$&$0.387$&$0.607$&$0.66$&$0$&$0$&$0$  \\[0.1em]
             \hline
        \end{tabular}
        \end{center}
\caption{\small The optimal  $P^*_{W\mid \mathbf X_{[3]}}$ (cf. Example~\ref{ex-nonlinear-comp}), where $W_2
$ is not used by this numerical optimizer.}
\label{tab:optimal-joint-distribution}
\end{table}

\begin{figure*}[t!]
    \centering
\includegraphics[width=\textwidth]{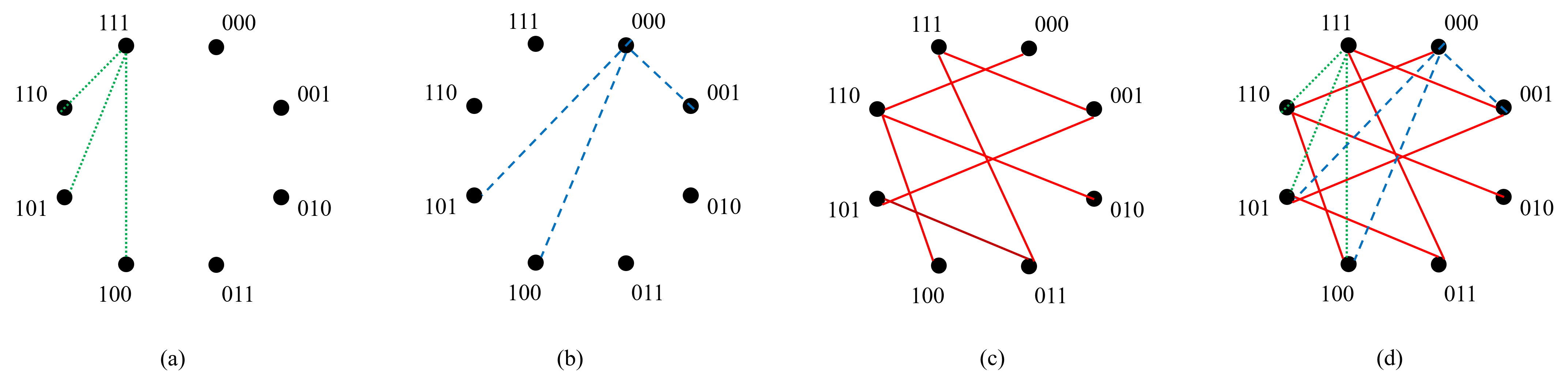}
\vspace{-0.4cm}
    \caption{\small Individual characteristic graphs and the broadcast graph (cf. Example~\ref{ex-nonlinear-comp}): (a) $G_{f_1}$, (b) $G_{f_2}$, and (c) $G_{f_3}$, based on users' demands $f_1$, $f_2$, and $f_3$, 
    respectively, and (d) the broadcast graph $G_{\mathrm B}$. To distinguish lines for different $i\in[K]$, we use different dash types in the individual characteristic graphs.}
    \vspace{-0.4cm}
    \label{fig:Graphs_ex1}
\end{figure*}

{\bf iii) Minimum entropy coloring-based approach.} 
From the sandwich property of $\minrank$~\cite{bar2011index}, $\chi(G_{\mathrm B})=\alpha(\overline{G}_{\mathrm B})=3$. 
Thus, the optimal one-shot zero-error rate is 
\begin{align}
R_{1,0}^*(3,3)=\log_2 \chi(G_{\mathrm B})=\log_2 3=1.585 \, \mbox{bits}.
\end{align}

Let $C=c_{G_{\mathrm B}}(\mathbf X_{[3]})\in\{R,B,G\}$ denote the color assigned to the source realization, where
\begin{align}
c_{G_{\mathrm B}}^{-1}(R)=\{000,111\},\quad
c_{G_{\mathrm B}}^{-1}(B)=\{001,010,011,100\},\quad
c_{G_{\mathrm B}}^{-1}(G)=\{101,110\} .
\end{align}
Under $\mathbf X_{[3]}\sim\mathrm{Unif}(\mathbb F_2^3)$, its unconditional distribution is $\bigl(\mathbb{P}(C=R),\mathbb{P}(C=B),\mathbb{P}(C=G)\bigr)=\left(0.25,0.5,0.25\right)$. 
Conditioned on the $i$-th source coordinate, the color distributions are
\begin{align}
\bigl(\mathbb{P}(C=R\mid X_1=0),\mathbb{P}(C=B\mid X_1=0),
\mathbb{P}(C=G\mid X_1=0)\bigr)
&=\left(0.25,0.75,0\right),\\
\bigl(\mathbb{P}(C=R\mid X_1=1),\mathbb{P}(C=B\mid X_1=1),
\mathbb{P}(C=G\mid X_1=1)\bigr)
&=\left(0.25,0.25,0.5\right),\\
\bigl(\mathbb{P}(C=R\mid X_i=x_i),\mathbb{P}(C=B\mid X_i=x_i),\mathbb{P}(C=G\mid X_i=x_i)\bigr)&=\left(0.25,0.5,0.25\right)
\end{align}
for every $i\in\{2,3\}$ and $x_i\in\{0,1\}$. Thus,
\begin{align}
H_2(C\mid X_1)&=\frac12 \mathsf H_2((0.25,0.75,0))+\frac12 \mathsf H_2((0.25,0.25,0.5))\approx 1.156 \ \text{bits},\\
H_2(C\mid X_2)&=H_2(C\mid X_3)=1.5 \ \text{bits}.
\end{align}
For this deterministic coloring, one can equivalently define $W=c_{G_{\mathrm B}}^{-1}(C)$, in which case $H_2(W\mid X_i)=H_2(C\mid X_i)$. Therefore, the conditional rate associated with this coloring is
\begin{align}
\label{eq:conditional_color_entropy}
\max_{i\in[3]}H_2(C\mid X_i)=H_2(W)=1.5 \ \text{bits}.
\end{align}

A proper minimum entropy coloring of $G_{\mathrm B}$ is achieved using a coloring   $c_{G_{\mathrm B}}\in\{R,B,G\}$ that satisfies $\mathbb{P}(R)=\mathbb{P}(W=\{000,111 \})=0.25$, $\mathbb{P}(B)=\mathbb{P}(W=\{001,010,011,100\})=0.5$, and $\mathbb{P}(G)=\mathbb{P}(W=\{101,110\})=0.25$ (see Figure~\ref{fig:ex1_coloring_and_its_complement} (Left) for $G_{\mathrm B}$). Hence, from~\eqref{eq:chromatic_entropy_vs_graph_entropy}
\begin{align}
\label{eq:rate_minentropy_ex1}
H_\chi(G_{\mathrm B},{\bf X}_{[3]})&\overset{(a)}{=} \min\limits_{\substack{P_{W\mid{\bf X}_{[3]}}:\\{\bf X}_{[3]}\in W,\;W\overset{\circ}{\in}\Gamma(G_{\mathrm B})}}\, H_3(W) \nonumber\\
&\overset{(b)}{=} H_{\chi,\mathrm B}^{\mathrm{SI}}(G_{\mathrm B},\mathbf X_{[3]}\mid \mathcal S_{[3]}=\mathbf X_{[3]})
=H_2(W)=\mathsf H_2((0.25,0.5,0.25))=1.5 \, \mbox{bits}\nonumber\\
&\overset{(c)}{>} \max_{i\in [3]}\, H_\kappa(G_{\mathrm B},{\bf X}_{[3]}\mid X_i)=1.406 \, \mbox{bits}
\end{align}
where $(a)$ holds since $W\overset{\circ}{\in}\Gamma$~\cite{alon1996source}, and equality in $(b)$ is specific to this example (~\eqref{eq:conditional_color_entropy}) and may not hold in general because $H_2(W)\geq H_2(W\mid X_i)$, for all $i\in[K]$, and step $(c)$ follows from~\eqref{eq:Rate_ach}, which is evaluated numerically. The solution in~\eqref{eq:rate_minentropy_ex1}, $H_\chi(G_{\mathrm B},{\bf X}_{[3]})$ is not optimal, and  can be tightened using $\GUn$
. Sending colorings of  $G^{\dot{\vee}n}_{\mathrm B}$ for sufficiently large $n$, followed by Slepian-Wolf coding, yields an achievable coding scheme (cf.~\cite{korner1971coding} and~\cite{feizi2014network}). However, the lower bound is not necessarily achievable as $\{f_i\}$ are not compatible (Propositions~\ref{prop_achievable-rate-via-broadcast_compatible} and~\ref{prop-2-user-rate_optimality}).

\begin{figure}[ht]
\centering
\includegraphics[width=0.25\textwidth]{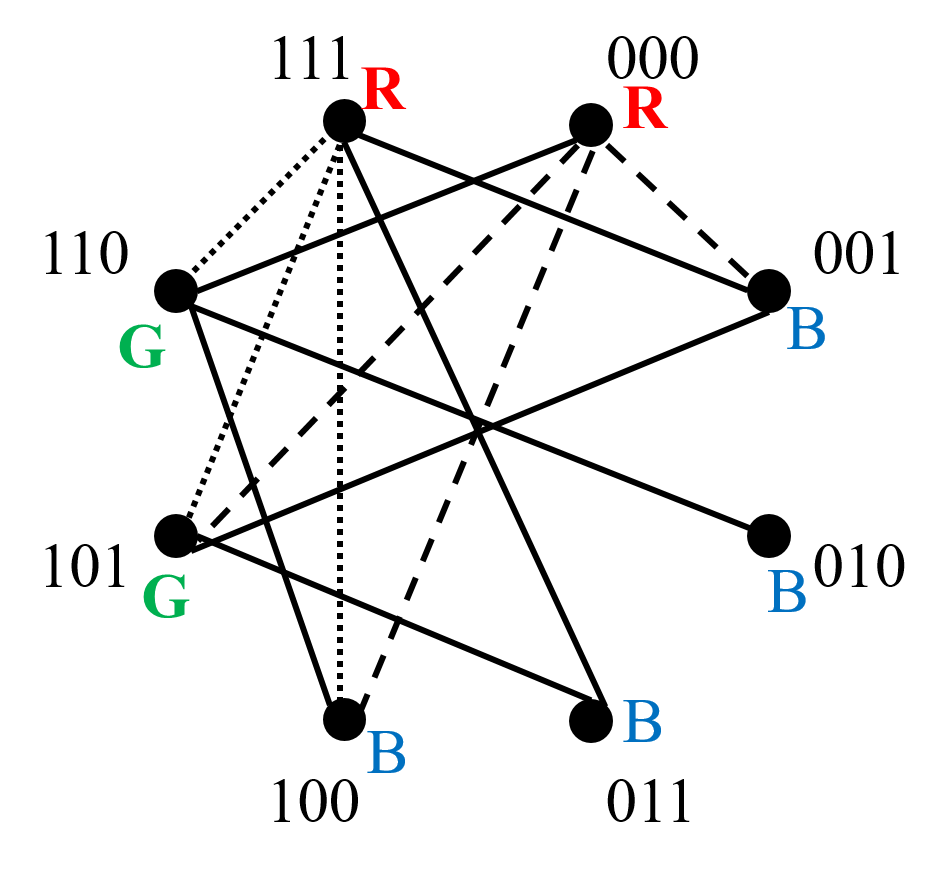}
\includegraphics[width=0.25\textwidth]{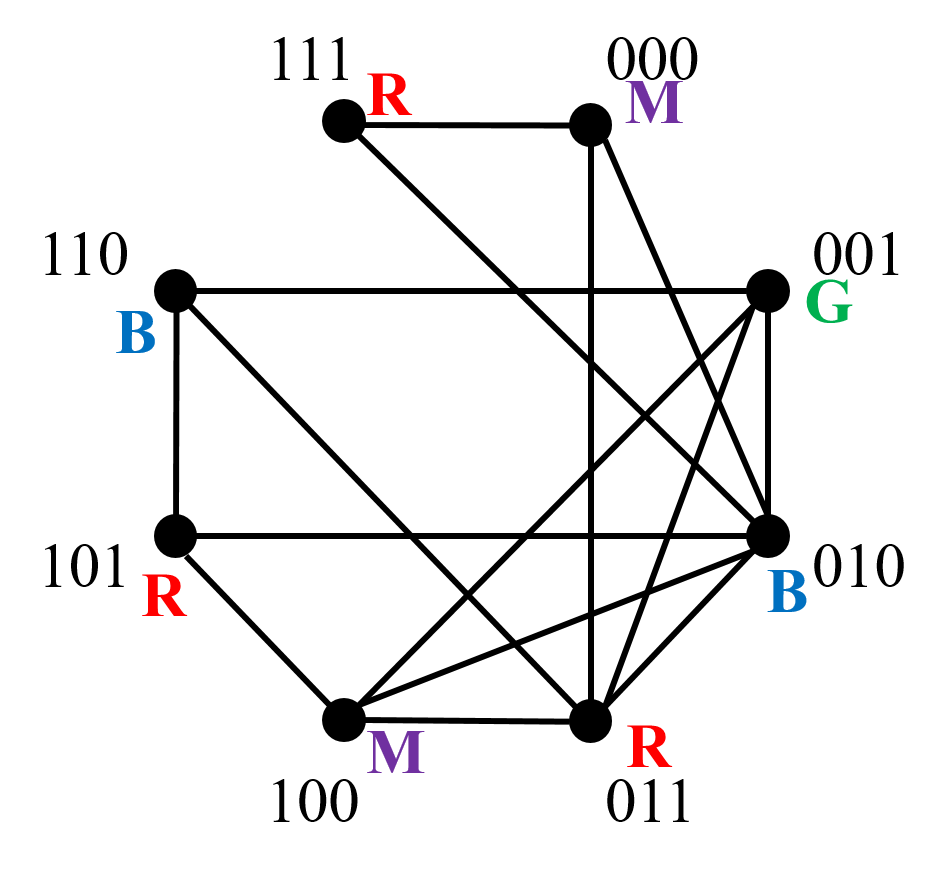}
\caption{\small (Left) The color encoded broadcast graph $G_{\mathrm B}$ (cf. Example~\ref{ex-nonlinear-comp}), where green $\textcolor{pptgreen}{\bf G}$, red $\textcolor{red}{\bf R}$, and blue $\textcolor{blue}{\bf B}$ represent different encoding classes, respectively. (Right) The color encoded $\overline{G}_{\mathrm B}$ (cf. Example~\ref{ex-nonlinear-comp}).}
\label{fig:ex1_coloring_and_its_complement}
\end{figure}

For the given $G_{\mathrm B}$, the number of maximal cliques is $6$, of which five are triangles, and one is a maximal edge, and the set of cliques is given as
\begin{align}
\label{eq:Omega-3-user}
\Omega(G_{\mathrm B})=\left\{\begin{array}{c}
W'_1=\{000,001,101\} \ , \ 
W'_2=\{000,100,110\} \ , \ 
W'_3=\{001,101,111\} \ , \\ 
W'_4=\{011,101,111\} \ , \ 
W'_5=\{100,110,111\} \ , \  
W'_6=\{010,110\} 
\end{array}\right\} .
\end{align}
The clique number $\omega(G_{\mathrm B})$ is the size of the largest clique in $G_{\mathrm B}$. Noting that there is no clique with more than three vertices, is $\omega(G_{\mathrm B})=\alpha(\overline{G}_{\mathrm B})=3$ (see  Figure~\ref{fig:ex1_coloring_and_its_complement} (Right) for $\overline{G}_{\mathrm B}$). 
We further note from~\eqref{eq:clique_entropy} that
\begin{align}
\label{eq:clique_entropy_ex1}
H_\omega(G_{\mathrm B},\mathbf X_{[3]})=H_2({\bf X}_{[3]}\mid W^')=1.417 \, \mbox{bits} 
\end{align}
where $W^'$ denotes the clique-set-valued random variable induced by $P^*_{\mathbf X_{[N]},W^'}=P_{\mathbf X_{[N]}}P^*_{W^'\mid{\mathbf X}_{[N]}}$. An optimal conditional distribution $P^*_{W^'\mid \mathbf X_{[3]}}$ that achieves~\eqref{eq:clique_entropy_ex1} is given in Table~\ref{tab:optimal-joint-distribution_clique}.

\begin{table}[t!]\footnotesize
 \setlength{\extrarowheight}{1pt}
\setlength{\tabcolsep}{2pt}
    \begin{center}
          \begin{tabular}
          {|c|c|c|c|c|c|c|c|c|}
            \hline&&&&&&&&\\[-1em]
            $W'\,\backslash\, {\bf X}_{[3]}$&$000$&$001$&$010$&$011$&$100$&$101$&$110$&$111$\\
             \hline &&&&&&&&\\[-1em]
             $W'_1=\{000, 001, 101\}$&$0.43$&$0.664$&$0$&$0$&$0$&$0.332$&$0$&$0$\\[0.1em]
             \hline&&&&&&&&\\[-1em]
             $W'_2=\{000, 100, 110\}$&$0.57$&$0$&$0$&$0$&$0.809$&$0$&$0.51$&$0$  \\
             [0.2em]
             \hline&&&&&&&&\\[-1em]
             $W'_3=\{001, 101, 111\}$&$0$&$0.336$&$0$&$0$&$0$&$0.168$&$0$&$0.217$\\
             \hline&&&&&&&&\\[-1em]
             $W'_4=\{011, 101, 111\}$&$0$&$0$&$0$&$1$&$0$&$0.5$&$0$&$0.648$  \\
             [0.2em]
             \hline&&&&&&&&\\[-1em]
             $W'_5=\{100, 110, 111\}$ &$0$&$0$&$0$&$0$&$0.191$&$0$&$0.12$&$0.135$\\[0.1em]
             \hline&&&&&&&&\\[-1em]
             $W'_6=\{010, 110\}$&$0$&$0$&$1$&$0$&$0$&$0$&$0.37$&$0$\\[0.1em]
             \hline
        \end{tabular}
        \end{center}
\caption{\small The conditional distribution $P^*_{W'\mid \mathbf X_{[3]}}$ 
that maximizes the RHS of~\eqref{eq:clique_entropy_ex1} (cf. Example~\ref{ex-nonlinear-comp}).}
\label{tab:optimal-joint-distribution_clique}
\end{table}

Moreover, neither $G_{\mathrm B}$ nor its complement $\overline{G}_{\mathrm B}$ contains an induced odd cycle of length at least five\footnote{
From the Strong Perfect Graph Theorem, $G_{\mathrm B}$ is perfect if no induced subgraph of $G_{\mathrm B}$ is an odd cycle of length at least five (that is $C_{2k+1}$, $k\geq 2$) or the complement of one~\cite{berge1961farbung}.}. Since $G_{\mathrm B}$ has eight vertices, it suffices to verify the absence of induced cycles of length five and seven in both $G_{\mathrm B}$ and $\overline{G}_{\mathrm B}$. Hence, by the Strong Perfect Graph Theorem, $G_{\mathrm B}$ is perfect\footnote{A perfect graph $G$ is a graph in which for every induced subgraph $G'\subseteq G$, 
$\omega(G')=\chi(G')$~\cite{hoang1987conjecture}.}~\cite{hoang1987conjecture}. Therefore, from~\cite{csiszar1990entropy} and noting that $\mathbf X_{[3]}\sim\mathrm{Unif}(\mathbb F_2^3)$
\begin{align}
\label{eq:perfect-graph-entropy-identity}
H_\kappa(G_{\mathrm B},\mathbf X_{[3]})+H_\kappa(\overline{G}_{\mathrm B},\mathbf X_{[3]})=H_2(\mathbf X_{[3]})=3 \, \mbox{bits} .
\end{align} 
From~\cite[Lemma~12]{alon1996source}, $H_\omega(G,X)=H_q(X)-H_\kappa(\overline{G},X)$ for every probabilistic graph $G$. Thus, $H_\omega(G_{\mathrm B},\mathbf X_{[3]})=3-H_\kappa(\overline{G}_{\mathrm B},\mathbf X_{[3]})=H_\kappa(G_{\mathrm B},{\mathbf X}_{[3]})=1.417$ bits, verifying~\eqref{eq:clique_entropy_ex1}. 
The optimal zero-error broadcast rate is (see~\eqref{eq:asymptotic_zero_error_clique_bound})
\begin{align}
\label{eq:asymptotic_zero_error_clique_bound_Ex1}
R_0^*(3,3)\geq\limsup_{n\to\infty}\frac{1}{n}H_\omega\left(G_{\mathrm B}^{(n)},\mathbf X_{[3]}^n\right)\neq H_\omega(G_{\mathrm B},\mathbf X_{[3]}) , 
\end{align} 
in general, because perfectness of $G_{\mathrm B}$ does not ensure that the exact operational block graph $G_{\mathrm B}^{(n)}$ is perfect, and $G_{\mathrm B}^{(n)}\neq \GUn$ in general, perfectness of the single-letter graph $G_{\mathrm B}$ alone does not establish a single-letter operational converse. 
Furthermore, the clique quantity for $G_{\mathrm B}^{(2)}$ satisfies $\frac12 H_\omega(G_{\mathrm B}^{(2)},\mathbf X_{[3]}^2)=1.364$ bits, for $G_{\mathrm B}^{(3)}$, it holds that $\frac13 H_\omega(G_{\mathrm B}^{(3)},\mathbf X_{[3]}^3)\gtrsim 1.301$, and for $G_{\mathrm B}^{(4)}$, we have $\frac14 H_\omega(G_{\mathrm B}^{(4)},\mathbf X_{[3]}^4)\gtrsim 1.234$,  showing that the single-letter clique entropy $H_\omega(G_{\mathrm B},\mathbf X_{[3]})=1.417$ bits does not tensorize, and does not provide a matching operational characterization. 
Similarly, employing Lemma~\ref{lem:smoothing_lemma} yields the same lower bound for $R^*(3,3)$ as in~\eqref{eq:asymptotic_zero_error_clique_bound_Ex1}.

Proposition~\ref{prop-genie}, from $H_2(f_3\mid X_3)=0.905$, $H_2(f_1\mid X_1, X_3, f_3)=0.25$, and $H_2(f_2 \mid X_2,X_1,X_3, f_1,f_3)=0$, yields 
 \begin{align}
 \label{eq-low-bound-genie-ex-3-user-nonlin}
 R^*(3,3) \geq H_2(f_3\mid X_3)+H_2(f_1\mid X_1, X_3, f_3) + H_2(f_2 \mid X_2,X_1,X_3, f_1,f_3) =1.155\, \mbox{ bits}\ .
 \end{align}
 Contrasting~\eqref{eq:rate_minentropy_ex1} 
with~\eqref{eq-low-bound-genie-ex-3-user-nonlin} showcases the tightness of the former bound  
versus Proposition~\ref{prop-genie}.

Contrasting the three schemes outlined above, we obtain the following. 
\begin{align}
\label{eq:rate_comparison_ex1}
1.155\leq R^*(3,3)\leq R_0^*(3,3)\leq H_\kappa(G_{\mathrm B},\mathbf X_{[3]})=1.417\leq  H_\chi(G_{\mathrm B},{\bf X}_{[3]})=1.5 \leq R_{\mathrm{SW}}(3,3)=2 .
\end{align}

\end{ex}

We next present an example (cf.~\cite[Example~2]{yao2024capacity}) for the problem of linear computation broadcast.

\begin{ex}[A {\bf \texorpdfstring{$(3,3,\{X_i\}_{i\in[3]},$ $\{\mathbf{X}_{[3]}^{\top}\mathbf{V}_i\}_{i\in[3]})$}{} linear computation broadcast model}]
\label{ex-linear-comp-jafar} 
Consider a $(3,3,\{X_i\}_{i\in[3]},$ $\{\mathbf{X}_{[3]}^{\top}\mathbf{V}_i\}_{i\in[3]})$ linear computation broadcast model, where ${\bf X}_{[3]}\sim\mathrm{Unif}(\mathbb{F}_3^3)$, and the desired functions are 
\begin{align}
\label{eq:desired_functions_example2}
f_1=\mathbf{X}_{[3]}^{\top}\mathbf{V}_1 ,\quad
f_2=\mathbf{X}_{[3]}^{\top}\mathbf{V}_2 ,\quad
f_3=\mathbf{X}_{[3]}^{\top}\mathbf{V}_3 
\end{align}
where $\mathbf{V}_1^{\intercal}=\small{\begin{bmatrix}0 & 1& 1\end{bmatrix}}$, $\mathbf{V}_2^{\intercal}=\small{\begin{bmatrix}1 & 0& 1\end{bmatrix}}$, and $\mathbf{V}_3^{\intercal}=\small{\begin{bmatrix}1 & 2& 0\end{bmatrix}}$. 

Below, we present three approaches to design $M$: 

{\bf i) MinRank approach.} 
Given the set of functions in~\eqref{eq:desired_functions_example2}, the vector $\mathbf{V}_i$ represents the coefficients of $f_i$, $i\in[3]$,  and each element $\mathbf{V}_i(j)$ for $j\in[3]$ corresponds to the coefficients assigned to $X_j\in\mathbb{F}_q$. 

Employing~\eqref{eq-min-rank-linear-demand} 
yields the following optimal linear vector index code length: 
\begin{align}
\label{eq-example-matrix-3-user-linear}
\minrank_q(\{{\bf V}_i, {\bf V}_i'\}_{i\in[3]})=\min_{\lambda_1, \lambda_2, \lambda_3}\, \left\{\rk\left(\begin{bmatrix}
    \lambda_{1} & 1 & 1 \\
    1 & \lambda_{2} & 1\\
    1 & 2 & \lambda_{3} 
\end{bmatrix} \right) \ : \ \lambda_i\in\mathbb{F}_3,\ \forall\, i\in[3]\right\}=2   
\end{align}
which follows from noting that ${\bf V}_1^{\prime\intercal}=\small{\begin{bmatrix}1 & 0 & 0 \end{bmatrix}}$, ${\bf V}_2^{\prime\intercal}=\small{\begin{bmatrix}0 & 1 & 0 \end{bmatrix}}$, and ${\bf V}_3^{\prime\intercal}=\small{\begin{bmatrix}0 & 0 & 1 \end{bmatrix}}$, respectively, because the side information of each user is $X_{i}={\bf X}_{[3]}^\intercal {\bf V}_{i}^{\prime}$ for $i\in [3]$. 
The optimal solution of~\eqref{eq-example-matrix-3-user-linear} is achieved by letting 
$\lambda_1^*=0$, $\lambda_2^*=0$, and $\lambda_3^*=2$, 
noting that the last row of the matrix in~\eqref{eq-example-matrix-3-user-linear} is a linear combination of the first two rows. The optimal index code with length $2$ is achievable when the sender node broadcasts
\begin{align}
\label{Aopt_Example3_standard}
{\bf W}^*=\begin{bmatrix}
    0 & 1 & 1 \\
    1 & 0 & 1
\end{bmatrix}{\bf X}_{[3]}^{\intercal} \ .
\end{align}


{\bf ii) Compatibility function-based approach.}  
Building on the compatibility function–based approach of~\cite{yao2024capacity, ravi2016broadcast}, 
we next employ a data-splitting approach to construct pairwise compatible transmissions and derive an upper bound on the total rate required to satisfy all users' demands.

Given user pair $(1,2)$, with demands given above, a compatible function is $Z_{12}=X_1\oplus_{3} X_2\oplus_{3} X_3$. Similarly, for user pairs $(1,3)$, and $(2,3)$,  functions $Z_{13}=X_1\oplus_{3} 2 X_2\oplus_{3} 2X_3$ and $Z_{23}=X_1\oplus_{3} 2X_2\oplus_{3} X_3$, respectively, are compatible. For any pair of distinct users $\{i_1,i_2\}$, the rate achieved by $Z_{{i_1}{i_2}}$ satisfies
\begin{align}
    \label{eq:Rij-constraint}
    R_{i_1 i_2}\leq H_3(Z_{{i_1}{i_2}})=H_3(f_{i_1}\mid X_{i_1})=1 \, \mbox{bit}, \quad 
    \{i_1,i_2\}\in[3]\ : \ i_1\neq i_2 \ .
\end{align}

Employing the data-splitting idea of~\cite[Example~2]{yao2024capacity}, we split each $X_j$, $j\in[3]$ into two disjoint sub-variables $X_j^l$, $l\in[2]$. 
Similarly, let $Z_{{i_1}{i_2}}^l$,  $l\in[2]$, denote the sub-variables for the compatible functions defined above, with $H_3(Z_{{i_1}{i_2}}^l)=0.5$ bits each. Consider the broadcast messages: (i) $Z_{12}^1$, (ii) $Z_{23}^2$, and (iii) $Z_{13}^1 \oplus_3 Z_{13}^2$. From (i) and (iii), $f_1$ is decoded; using (i) and (ii), $f_2$ is decoded; and finally from (ii) and (iii), $f_3$ is decoded. Thus, splitting each dataset into two subsymbols followed by jointly encoding the two instances, 
achieves the broadcast rate  
\begin{align}
 \label{ex-linear-rate-split}
 R_{\mathrm Compat.}(3,3)&= H_3(Z_{12}^1,Z_{23}^2,Z_{13}^1\oplus_{3}Z_{13}^{2})\nonumber\\
 &\overset{(a)}{=}H_3(Z_{12}^1)+H_{3}(Z_{23}^{2}\mid Z_{12}^1)+H_{3}(Z_{13}^1\oplus_{3}Z_{13}^{2}\mid Z_{12}^1, Z_{23}^{2})\nonumber\\
 &\overset{(b)}{=} H_3(Z_{12}^1)+H_{3}(Z_{23}^{2})+H_{3}(Z_{13}^1\oplus_{3}Z_{13}^{2})=1.5\,\mbox{bits} 
 \end{align}
where $(a)$ follows from the chain rule, and $(b)$ from 
the independence of $Z_{12}^1$, $Z_{23}^2$, and noting that $H_{3}(Z_{13}^1\oplus_{3}Z_{13}^{2}\mid Z_{12}^1, Z_{23}^{2})=H_{3}(X_1^1\oplus_{3}2X_2^1\oplus_{3}2X_3^1\oplus_{3}X_1^2\oplus_{3}2X_2^2\oplus_{3}2X_3^2\mid X_1^1\oplus_3 X_2^1\oplus_3 X_3^1, X_1^2\oplus_3 2X_2^2\oplus_3 X_3^2)=H_3(X_2^1\oplus_{3}X_3^1\oplus_{3}X_3^2\mid X_1^1\oplus_3 X_2^1\oplus_3 X_3^1, X_1^2\oplus_3 2X_2^2\oplus_3 X_3^2)=0.5$.

{\bf iii) Broadcast graph approach.} 
We next exploit the achievability result in 
Proposition~\ref{prop_achievable-rate-via-broadcast_compatible} for the linear function computation broadcast setting in this example. The sets of independent sets of each $G_{f_i}$, $i\in[3]$, are given as
\begin{align}
\label{Gamma-3-user-linear-separate}
\Gamma(G_{f_1})&=
\{\{\star 00,\star12,\star21\},\{\star01,\star10,\star22\},\{\star02,\star20,\star11\}\} \ ,\\
\Gamma(G_{f_2})&=\{ \{0\star0,1\star2,2\star1\},\{0\star1,1\star0,2\star2\},\{0\star2,2\star0,1\star1\}\}\ , \\ 
\Gamma(G_{f_3})&=\{\{00\star,11\star,22\star\}, \{02\star,10\star,21\star\},\{01\star,12\star,20\star\}\}
\end{align}
where $\star$'s denote don't care terms, i.e., values in  $\star$ positions may be chosen arbitrarily without altering the outcome. 

The resulting broadcast graph, $G_{\mathrm B}=(\mathbb F_3^3,\mathcal E_{\mathrm B})$, with $|\mathcal{V}(G_{\mathrm B})|=27$, can be derived from $(\{\mathcal{S}_i\}_{i\in[3]},\ \{f_i\}_{i\in[3]})$. Employing the method of~\cite{tsukiyama1977new}, we enumerated all MISs of a graph, and numerically\footnote{The implementation is available at \href{https://github.com/kumarkpvijith/Non-linear-function-computation-broadcast.git}{https://github.com/kumarkpvijith/Non-linear-function-computation-broadcast.git}.}, and verified that $G_{\mathrm B}$ has $225$ MISs. Table~\ref{tab:w-lin} presents only those MISs with nonzero probabilities that optimize $H_\chi(G_{\mathrm B},{\bf X}_{[3]})$ (see~\eqref{eq:chromatic_entropy_vs_graph_entropy}). 
The sender can then exploit Slepian-Wolf coding to encode $W^n$ against the side information $X_i^n$ simultaneously for all three users. Thus, random binning at any rate $R>\max\limits_{i\in[3]}H_3(W\mid X_i)=H_{\chi,\mathrm B}^{\mathrm{SI}}(G_{\mathrm B},\mathbf X_{[3]}\mid \mathcal S_{[3]}=\mathbf X_{[3]})
$ allows all users to recover $W^n$ with probability of error tending to zero. Hence, the side-information-assisted rate 
is
\begin{align}
\label{eq:ach_MinRank_Example2}
R^*(3,3)\leq H_{\chi,\mathrm B}^{\mathrm{SI}}(G_{\mathrm B},\mathbf X_{[3]}\mid \mathcal S_{[3]}=\mathbf X_{[3]})
=1.72 \   \mbox{trits}  
\end{align}
where $P^*_W$ given in Table~\ref{tab:w-lin} that achieves~\eqref{eq:chromatic_entropy_vs_graph_entropy} assigns positive probabilities to disjoint MISs $\Gamma(\GU)$.

\begin{table}[t!]\footnotesize
\centering
\setlength{\tabcolsep}{4pt}
\renewcommand{\arraystretch}{1.2}
\begin{tabular}{|c|c||c|c|}
\hline
$W$ & $\mathbb{P}(W)$ & $W$ & $\mathbb{P}(W)$ \\
\hline
$W_1=\{000,012,121,201\}$ & $\frac{4}{27}$ & $W_5=\{021,101,200,212\}$ & $\frac{4}{27}$ \\
$W_2=\{001,010,120,222\}$ & $\frac{4}{27}$ & $W_6=\{022,102,111,210\}$ & $\frac{4}{27}$ \\
$W_3=\{002,011,110,220\}$ & $\frac{4}{27}$ & $W_7=\{122,202,211\}$ & $\frac{3}{27}$ \\
$W_4=\{020,100,112,221\}$ & $\frac{4}{27}$ & & \\
\hline
\end{tabular}
\caption{\small Distribution $P^*_W$ that optimizes $H_\chi(G_{\mathrm B},{\bf X}_{[3]})$ (see~\eqref{eq:chromatic_entropy_vs_graph_entropy}) (cf. Example~\ref{ex-linear-comp-jafar}).}
\label{tab:w-lin}
\end{table}

$G_{\mathrm B}$ has $|\mathcal{V}|=27$, $\mathrm{deg}(v)=12$ for every $v\in\mathcal{V}$ and is vertex-transitive\footnote{A vertex-transitive graph is a graph where every vertex looks identical from a structural perspective, meaning that for any two vertices, there is an automorphism (a symmetry-preserving map) that moves one vertex to the other~\cite{godsil2001algebraic}.}. For a vertex-transitive graph $G$ under uniform distribution, $H_\kappa(G,X)\triangleq\log\frac{|\mathcal{V}(G)|}{\alpha(G)}$. Thus, from Theorem~\ref{theo-general-achievable-rate-via-broadcast}, the unconditional K\"orner graph-entropy is
\begin{align}
\label{eq:VT_graph_entropy}
H_\kappa(G_{\mathrm B},\mathbf X_{[3]})=\log_3 \frac{27}{4}=1.738 \   \mbox{trits}  
\end{align}
which does not exploit each user's side information, unlike the $1.72$ ternary symbols scheme. 

Because $G_{\mathrm B}$ is vertex-transitive and $\alpha(\overline{G}_{\mathrm B})=\omega(\overline{G}_{\mathrm B})=4$, and $\overline{G}_{\mathrm B}$ is also vertex-transitive, i.e., $H_\kappa(\overline{G}_{\mathrm B},\mathbf X_{[3]})=H_\kappa(G_{\mathrm B},\mathbf X_{[3]})=1.738 \ \mbox{trits}$. 
Furthermore, $\alpha(G_{\mathrm B})=\omega(G_{\mathrm B})=4$ and $\chi(G_{\mathrm B})=7$. Thus, $G_{\mathrm B}$ is not perfect. From~\cite[Lemma~12]{alon1996source}, $H_\omega(G_{\mathrm B},\mathbf X_{[3]})=3-H_\kappa(\overline{G}_{\mathrm B},\mathbf X_{[3]})=1.262<H_\kappa(G_{\mathrm B},\mathbf X_{[3]})=1.738$ trits.

From~\eqref{eq:ach_MinRank_Example2} (see Table~\ref{tab:w-lin}) and~\eqref{eq:VT_graph_entropy}, the asymptotically lossless functional-computation rate satisfies
\begin{align}
R^*(3,3)\leq \min\Bigl\{H_{\chi,\mathrm B}^{\mathrm{SI}}(G_{\mathrm B},\mathbf X_{[3]}\mid \mathcal S_{[3]}=\mathbf X_{[3]})
,\ H_\kappa(G_{\mathrm B},\mathbf X_{[3]})\Bigr\}= 1.72 \ \mbox{trits} .
\end{align}

The single-letter broadcast graph $G_{\mathrm B}$ has $225$ MISs, of which $162$ have cardinality four and $63$ have cardinality three. Solving the resulting convex minimax problem yields $H_\kappa(G_{\mathrm B},\mathbf X_{[3]}\mid X_i)=1.667$ trits. 
Adopting the definition~\eqref{eq:single_letter_side_info_assisted_def_G_B}, where, numerically, the three conditional mutual informations are equal at the optimum:
\begin{align}
H_{\kappa,\mathrm B}^{\mathrm{SI}}\left(G_{\mathrm B},\mathbf X_{[3]}\mid\mathcal S_{[3]}\right)=I_3({\bf X}_{[3]};W\mid X_i)=1.667 \ \mbox{trits},\quad i\in[3]
\end{align}
where $W\in\Gamma(G_{\mathrm B})$. Applying this construction repeatedly treats two sequences as requiring separation whenever $\exists t\in[n],\exists i\in[3]:\ X_{i,t}=\widetilde X_{i,t},\ f_i(\mathbf x_{[N],t})\neq\widetilde f_i(\widetilde{\bf x}_{[N],t})$ where the witnessing user $i\in[3]$ may change with $t$.

Applying the genie-aided lower bound of Proposition~\ref{prop-genie} yields~\eqref{eq-proof-lem-3-user-gap-genie-2}, which gives $R^*(3,3)\geq 1$, and is not tight. 

The operational block graph $G_{\mathrm B}^{(n)}=\left(\mathbb F_3^{3n},\mathcal E_{\mathrm B}^{(n)}\right)$ for this example is a Cayley graph because adjacency depends only on the difference between two source blocks~\cite{magnus2004combinatorial}. Furthermore, $G_{\mathrm B}^{(n)}$ is vertex-transitive. This gives
\begin{align}
\label{eq:Cayley_group_cosets}
\frac1n\log_3\frac{3^{3n}}{\alpha(G_{\mathrm B}^{(n)})}\overset{(a)}{\leq} \frac1n H_{\chi,\epsilon}\Bigl(G_{\mathrm B}^{(n)},\mathbf X_{[3]}^n\Bigr)\overset{(b)}{\leq} \frac1nH_3(c_n(\mathbf X_{[3]}^n)) 
\end{align}
for a proper coloring $c_n\in\mathscr C(G_{\mathrm B}^{(n)})$, where $(a)$ holds because $\mathbf X_{[3]}^n\sim\mathrm{Unif}(\mathbb F_3^{3n})$, the graph has $|\mathcal{V}(G_{\mathrm B}^{(n)})|=3^{3n}$ equally likely vertices, and because every color class of a proper coloring is an independent set, its cardinality cannot exceed $\alpha(G_{\mathrm B}^{(n)})$. Thus, the largest color probability is at most $\frac{\alpha(G_{\mathrm B}^{(n)})}{3^{3n}}$. Because this holds for every proper coloring, for the minimum-entropy coloring, $H_{\chi,\epsilon}\Bigl(G_{\mathrm B}^{(n)},\mathbf X_{[3]}^n\Bigr)\geq \log_3\frac{3^{3n}}{\alpha(G_{\mathrm B}^{(n)})}$. 
Inequality $(b)$ holds by~\eqref{eq:operational-chromatic-entropy}.

\paragraph{K\"orner graph entropy of block Cayley graphs}

Since ${\bf X}_{[3]}^n\sim\mathrm{Unif}(\mathbb F_3^{3n})$, the normalized graph entropy of a vertex-transitive graph satisfies
\begin{align}
\label{eq:normalized-korner-cayley-example-2}
\frac1n H_\kappa\left(G_{\mathrm B}^{(n)},{\bf X}_{[3]}^n\right)=\frac1n\log_3\frac{|\mathcal V(G_{\mathrm B}^{(n)})|}{\alpha(G_{\mathrm B}^{(n)})}=3-\frac1n\log_3\alpha(G_{\mathrm B}^{(n)})
\end{align}
which is determined by the exponential growth rate of the independence number.

For $n=2$, $H_\kappa(G_{\mathrm B}^{(2)},\mathbf X_{[3]}^2)=\log_3 \frac{|\mathcal{V}(G_{\mathrm B}^{(2)})|}{\alpha(G_{\mathrm B}^{(2)})}\leq 3$ by employing $|\mathcal{V}(G_{\mathrm B}^{(2)})|=3^6=729$ and $\alpha(G_{\mathrm B}^{(2)})\geq 27$ because each color class contains $3^{6-3}=27$ vertices and every coset is an independent set of $G_{\mathrm B}^{(2)}$ due to the zero error property. 
The $27$-color splitting construction yields the exact two-letter result $\frac12 H_{\chi,\epsilon}\Bigl(G_{\mathrm B}^{(2)},\mathbf X_{[N]}^2\Bigr)=1.5$ trits. Together with the $27$-color splitting construction, we obtain the exact-two letter result $\frac12 H_\kappa(G_{\mathrm B}^{(2)},\mathbf X_{[3]}^2)\leq 1.5$ trits. This scheme, similar to the splitting scheme detailed in part ii), only separates pairs that are genuinely confusable in $G_{\mathrm B}^{(2)}$ where $G_{\mathrm B}^{(2)}\subsetneq G_{\mathrm B}^{\dot\vee 2}$. 
For $n=3$, the known splitting construction uses five ternary transmissions, giving rate $1.667$ trits. Thus, $\frac13 H_{\chi,\epsilon}\Bigl(G_{\mathrm B}^{(3)},\mathbf X_{[N]}^3\Bigr)\leq 1.667$ trits. Furthermore, $\alpha(G_{\mathrm B}^{(3)})\geq 108$, where the construction on the RHS is obtained by concatenating an independent set of size $27$ for $n=2$ and an independent set of size $4$ for $n=1$. Thus, $H_\kappa(G_{\mathrm B}^{(3)},\mathbf X_{[3]}^3)\leq \log_3 \frac{3^9}{108}$, yielding $\frac13H_\kappa(G_{\mathrm B}^{(3)},\mathbf X_{[3]}^3)\leq 1.58$ trits. For $n>3$, we can upper bound the graph entropy, as detailed next.

\paragraph{Upper bound induced by the linear construction and the unrestricted K\"orner graph entropy construction}
Denoting by $\alpha(G_{\mathrm B}^{(n)})$, and $\alpha_{\mathrm{lin}}(G_{\mathrm B}^{(n)})$ unrestricted and linear independence numbers, respectively, it follows that 
\begin{multline}
\label{eq:unrestricted_vs_linear}
\alpha(G_{\mathrm B}^{(n)})=\max\left\{|\mathcal I|:\mathcal I\subseteq \mathcal V(G_{\mathrm B}^{(n)})\subseteq\mathbb F_3^{3n},\ \mathcal I\text{ is independent}\right\}\\
\geq \alpha_{\mathrm{lin}}(G_{\mathrm B}^{(n)})\triangleq\max\left\{|\mathcal K|:\mathcal K\leq\mathbb F_3^{3n},\ \mathcal K\text{ is an independent linear subspace}\right\}=3^{\lfloor1.5n\rfloor} 
\end{multline}
where $K\leq \mathbb F_3^{3n}$ means that $K$ is a linear subspace of $\mathbb F_3^{3n}$. 
The linear splitting construction (part ii)) produces an independent subspace of dimension $\left\lfloor1.5n\right\rfloor$. Thus, 
the achievable rate for linear encoding is
\begin{align}
\label{eq:normalized-korner-linear-upper}
R_n^{\mathrm{lin}}= 3-\frac1n\log_3(\alpha_{\mathrm{lin}}(G_{\mathrm B}^{(n)}))=\lceil1.5n\rceil n^{-1}\leq 1.5+0.5n^{-1} ,
\end{align}
where the even-blocklength subsequence is bounded above by $1.5$, while the odd-blocklength upper bound approaches $1.5$ from above. 
Equality in~\eqref{eq:normalized-korner-linear-upper} would require proving $\alpha(G_{\mathrm B}^{(n)})=3^{\lfloor1.5n\rfloor}$. The linear kernel argument establishes that no larger linearly independent subspace exists, but it does not exclude larger non-linear independent sets.

From~\eqref{eq:normalized-korner-cayley-example-2} and~\eqref{eq:unrestricted_vs_linear}, 
\begin{align}
\label{eq:normalized-korner-upper}
\frac1n H_\kappa\left(G_{\mathrm B}^{(n)},{\bf X}_{[3]}^n\right)=3-\frac1n\log_3(\alpha(G_{\mathrm B}^{(n)}))\leq 1.5+0.238n^{-1}
\end{align}
where for $n=2\ell$, $\ell\in\mathbb Z$, $3-\lfloor1.5n\rfloor n^{-1}=1.5$, and because all coordinates can be paired and each pair contributes an independent set of size $27$, we have $\alpha(G_{\mathrm B}^{(n)})\geq 27^\ell$. Hence, $3-\frac1n\log_3(\alpha(G_{\mathrm B}^{(n)}))\leq 3-1.5=1.5$. 
For $n=2\ell+1$, $\ell\in\mathbb Z$, only $2\ell=n-1$ coordinates can be paired, and one remains unpaired, where the paired coordinates contribute $27^\ell$, while the remaining one-letter coordinate contributes an independent set of size $\alpha(G_{\mathrm B}^{(1)})=4$. Thus, $\alpha(G_{\mathrm B}^{(n)})\geq 4\cdot 27^\ell$. Hence, $3-\frac1n\log_3(\alpha(G_{\mathrm B}^{(n)}))\leq 3-\frac1n(3\ell+\log_3 4)=1.5+(1.5-\log_3 4)n^{-1}$. 
Upper-bounding these two cases and noting that $1.5-\log_3 4=0.238$, the inequality in~\eqref{eq:normalized-korner-upper} follows.

\paragraph{A general lower bound}
Employing the clique--coclique bound, which gives $\alpha(G_{\mathrm B}^{(n)})\omega(G_{\mathrm B}^{(n)})\leq|\mathcal V(G_{\mathrm B}^{(n)})|$,  and exploiting~\eqref{eq:normalized-korner-cayley-example-2} yield
\begin{align}
\label{eq:clique_coclique_vertex_transitive_ex2}
H_\kappa\left(G_{\mathrm B}^{(n)},{\bf X}_{[3]}^n\right)\geq \log_3 \omega(G_{\mathrm B}^{(n)}) .
\end{align}

Furthermore, since $G_{\mathrm B}^{(n)}$ is Cayley, using~\cite{conlon2024clique}, the clique number exponent  equals $1$,
\begin{align} 
\label{eq:clique_no_limit_ex2}
\lim\limits_{n\to\infty}\frac1n\log_3\omega(G_{\mathrm B}^{(n)})=1 .
\end{align}

Combining the clique lower bound~\eqref{eq:clique_coclique_vertex_transitive_ex2} and the linear upper bound~\eqref{eq:normalized-korner-linear-upper} gives
\begin{align}
\label{eq:normalized-korner-sandwich}
\frac1n\log_3\omega(G_{\mathrm B}^{(n)})\leq\frac1n H_\kappa\left(G_{\mathrm B}^{(n)},{\bf X}_{[3]}^n\right)\leq R_n^{\mathrm{lin}} 
\end{align}
where the second inequality follows from~\eqref{eq:unrestricted_vs_linear}. 
Therefore, taking the limit of both sides in~\eqref{eq:normalized-korner-sandwich},
\begin{align}
\label{eq:regularized-korner-range}
1\leq\lim_{n\to\infty}\frac1n H_\kappa\left(G_{\mathrm B}^{(n)},{\bf X}_{[3]}^n\right)\leq 1.5.
\end{align}
However, the lower bound in~\eqref{eq:regularized-korner-range} does not imply that the regularized K\"orner entropy equals $1$. 
The Cayley symmetry provides the exact representation of $\lim\limits_{n\to\infty}\frac1n H_\kappa\left(G_{\mathrm B}^{(n)},{\bf X}_{[3]}^n\right)$ in~\eqref{eq:normalized-korner-cayley-example-2}. While the superadditivity of $\log_3 \alpha(G_{\mathrm B}^{(n)})$ guarantees the existence of the limit, computing the exact value  
of $\alpha(G_{\mathrm B}^{(n)})$ remains open.

The asymptotic zero-error operational multiletter broadcast graph entropy~\eqref{eq:asymptotic-zero-error-broadcast-entropy} satisfies  
\begin{align}
R_0^*(3,3)
\leq 1.5 \ \mbox{trits}
\end{align}
due to the splitting condition~\eqref{eq:normalized-korner-linear-upper}. The asymptotically lossless operational broadcast graph entropy~\eqref{eq:asymptotically-lossless-operational-entropy} satisfies
\begin{align}
1\leq R^*(3,3)
\end{align}
noting that $R^*(3,3)
\geq H_3(f_i\mid X_i)=1$ for each $i\in[3]$. 
Then, employing the hierarchy in~\eqref{eq:broadcast-entropy-hierarchy}
\begin{align}
1\leq R^*(3,3)
\leq R_0^*(3,3)
\leq 1.5
\end{align}

Because $G_{\mathrm B}^{(n)}$ is vertex-transitive under uniform distribution, its automorphisms can move the set of all maximum cliques of $G_{\mathrm B}^{(n)}$ symmetrically across all vertices. Randomly selecting an image $W$ from this set (i.e., with size exactly $\omega(G_{\mathrm B}^{(n)})$), and then selecting $\mathbf X_{[3]}$ uniformly from $W$, symmetry ensures  $\mathbf X_{[3]}\sim\mathrm{Unif}(\mathbb F_3^3)$. Furthermore, $\mathbf X_{[3]}\mid W=w$ is uniform over the maximum clique $w$, and thus $H_3(\mathbf X_{[3]}\mid W=w)=\log_3\omega(G_{\mathrm B}^{(n)})$.  
Thus, the clique entropy under the uniform source reduces to the logarithm of its clique number:
\begin{align}
R^*(3,3)\geq \limsup_{n\to\infty}\frac{1}{n}H_\omega\left(G_{\mathrm B}^{(n)},\mathbf X_{[N]}^n\right)= \limsup\limits_{n\to\infty} \frac1n \log_3\omega(G_{\mathrm B}^{(n)})=1
\end{align}
where the last step follows from noting from~\eqref{eq:clique_no_limit_ex2} that  $\frac1n \log_3\omega(G_{\mathrm B}^{(n)})=1$.

Contrasting the schemes outlined above, and employing~\eqref{eq:broadcast-entropy-hierarchy} we obtain
\begin{align}
\label{eq:inequality_split_vs_one-shot_graph}
1\leq R^*(3,3)\leq R_0^*(3,3)\leq R_{\mathrm Compat.}(3,3)=1.5&<H_{\kappa,\mathrm B}^{\mathrm{SI}}\left(G_{\mathrm B},\mathbf X_{[3]}\mid\mathcal S_{[3]}\right)=1.667\nonumber\\
&<H_{\chi,\mathrm B}^{\mathrm{SI}}(G_{\mathrm B},\mathbf X_{[3]}\mid \mathcal S_{[3]}=\mathbf X_{[3]})
=1.72\nonumber\\
&<H_\kappa(G_{\mathrm B},{\bf X}_{[3]})=1.738 \nonumber\\
&<\minrank_q(\{{\bf V}_i, {\bf V}_i'\}_{i\in[3]})=R_{\mathrm{SW}}(3,3)=2 .
\end{align}
The $H_{\kappa,\mathrm B}^{\mathrm{SI}}\left(G_{\mathrm B},\mathbf X_{[3]}\mid\mathcal S_{[3]}\right)=1.667$ trits optimization 
uses the union graph before taking blocklength products. That graph operation loses which user creates each edge. This demonstrates a strict multiletter gain and shows that the single-letter conditional independent-set expression need not tensorize, and should be interpreted as a single-letter achievable bound rather than the operational broadcast graph entropy.

Both the compatibility function-based approach~\cite{ravi2016broadcast} (part ii) with rate $R_{\mathrm Compat.}(3,3)=1.5$ trits given in~\eqref{ex-linear-rate-split}, and the subspace decomposition-based approach of~\cite[Example~2]{yao2024capacity},  optimal under linear encoding, with $R_n^{\mathrm{lin}}$  in~\eqref{eq:normalized-korner-linear-upper} exploit precisely this lost information, which is not surprising because a vector linear code can exploit cross-coordinate alignment that a single-letter randomized independent-set construction cannot represent. 
\end{ex}


\begin{remark}
\label{remark-Non-Linear-Vector-Encoding}
In the $(K, N, \{\mathcal{S}_i\}_{i\in[K]}, \{f_i\}_{i\in[K]})$ linear computation broadcast setting, Example~\ref{ex-linear-comp-jafar} shows that a data-splitting-based linear vector encoding scheme, building on the fixed-length minrank framework~\cite{LeeDimaindex} and exploiting source statistics yields savings in the average broadcast length compared to~\eqref{eq-theo-general-achievable-rate-via-broadcast}. 
\end{remark}


We next present an example, which is a known non-linear index-coding instance into our computation-broadcast model~\cite{sharififar2021broadcast}. We demonstrate that non-linear encoder provides rate savings over the best linear encoder. 
\begin{ex}[A $(26,26,\{\mathcal{S}_i\}_{i\in[26]},\{f_i\}_{i\in[26]})$ Boolean computation broadcast model with linear demands with a strictly better non-linear encoder~\cite{sharififar2021broadcast}]
\label{ex:nonlinear-outperforms-linear}

Consider the $26$-user unicast index-coding instance of~\cite{sharififar2021broadcast}, interpreted as a computation-broadcast problem. The sources $X_1,\ldots,X_{26}\overset{\mathrm{i.i.d.}}{\sim}{\mathrm{Bern}}\left(\frac12\right)$ are mutually independent, user $i$ requests the linear function $f_i(\mathbf X_{[26]})=X_i$, $i\in[26]$, and has access to the side-information set $\mathcal S_i$ specified in~\cite{sharififar2021broadcast}.

For each user, define the interfering set $\mathcal B_i\triangleq[26]\setminus\bigl(\mathcal A_i\cup\{i\}\bigr)$, where $\mathcal A_i$ indexes the sources available to user~$i$ as side information. Consider the interfering sets
\begin{align}
\label{eq:nonlinear-example-interference-sets}
\mathcal B_1&=\{2,3,5,11,12,13,15,22,23,24,25,26\},\quad
\mathcal B_2=\{1,3,6,11,12,13,16,21,22,24,25,26\},\nonumber\\
\mathcal B_3&=\{1,2,4,11,12,13,14,21,22,23,24,26\},\nonumber\\
\mathcal B_4&=\{5,6,14,15,16\},\quad
\mathcal B_5=\{4,6,14,15,16\},\quad
\mathcal B_6=\{4,5,14,15,16\},\quad
\mathcal B_7=\{17\},\nonumber\\
\mathcal B_8&=\{1,5,7,11,15,17,18\},\quad
\mathcal B_9=\{2,6,7,12,16,17,19\},\quad
\mathcal B_{10}=\{3,4,7,13,14,17,20\},\nonumber\\
\mathcal B_{11}&=\{1,2,3,5,12,13,15,21,23,24,25,26\},\quad
\mathcal B_{12}=\{1,2,3,6,11,13,16,21,22,23,25,26\},\nonumber\\
\mathcal B_{13}&=\{1,2,3,4,11,12,14,21,22,23,24,25\},\quad
\mathcal B_{14}=\{4\},\quad
\mathcal B_{15}=\{5\},\quad
\mathcal B_{16}=\{6\},\quad
\mathcal B_{17}=\{7\},\nonumber\\
\mathcal B_{18}&=\{1,5,7,8,11,15,17\},\quad
\mathcal B_{19}=\{2,6,7,9,12,16,17\},\quad
\mathcal B_{20}=\{3,4,7,10,13,14,17\},\nonumber\\
\mathcal B_{21}&=\{4,6,14,16,22\},\quad
\mathcal B_{22}=\{4,6,14,16,21\},\quad
\mathcal B_{23}=\{4,5,14,15,24\},\nonumber\\
\mathcal B_{24}&=\{4,5,14,15,23\},\quad
\mathcal B_{25}=\{5,6,15,16,26\},\quad
\mathcal B_{26}=\{5,6,15,16,25\}.
\end{align}
Thus, user~$i$ observes the side-information set
\begin{align}
\mathcal S_i=\{X_j:j\in\mathcal A_i\},
\qquad
\mathcal A_i=[26]\setminus\bigl(\mathcal B_i\cup\{i\}\bigr) , \quad i\in[26].
\end{align}
The broadcast graph for this instance is $G_{\mathrm B}\triangleq\left(\mathbb F_2^{26},\mathcal E_{\mathrm B}\right)$, where, for distinct $\mathbf x,\widetilde{\mathbf x}\in\mathbb F_2^{26}$,
\begin{align}
\label{eq:26-variable-broadcast-graph}
\{\mathbf x_{[26]},\widetilde{\mathbf x}_{[26]}\}\in\mathcal E_{\mathrm B}\iff\exists i\in[26]:\quad \mathbf x_{\mathcal A_i}=\widetilde{\mathbf x}_{\mathcal A_i}\quad\text{and}\quad f_i(\mathbf x_{[26]})=x_i\neq f_i(\widetilde{\mathbf x}_{[26]})=\widetilde x_i .
\end{align}
Because adjacency depends only on the difference $\mathbf x_{[26]}-\widetilde{\mathbf x}_{[26]}$, $G_{\mathrm B}$ is a Cayley graph and is thus vertex-transitive.

\begin{prop}[The advantage of non-linear encoding over scalar and vector linear encoding]
\label{prop:advantage_nonlinear}
The instance described in Example~\ref{ex:nonlinear-outperforms-linear} has the following properties (cf.~\cite{sharififar2021broadcast}):
\begin{enumerate}[(i)]
\item The set $\mathcal J=\{1,2,3,11,12,13\}$
satisfies $\mathcal S_i\subseteq\mathbf X_{\mathcal J^c}$, $i\in \mathcal J$, where $\mathcal J^c=[26]\backslash \mathcal J$.
\item There exists a scalar non-linear zero-error encoder with six binary transmissions, $M=(Z_1,\ldots,Z_6)\in\mathbb F_2^6$.
\item No scalar or vector linear encoder over $\mathbb F_2$ attains rate $6$.
\end{enumerate}
\end{prop}

\begin{proof}[Proof of Proposition~\ref{prop:advantage_nonlinear}]
The sender broadcasts the following six binary symbols (see~\cite[Proposition 6]{sharififar2021broadcast}):
\begin{align}
\label{eq:nonlinear-example-encoder}
Z_1={}&X_1\oplus X_4\oplus X_6\oplus X_7
       \oplus X_{10}\oplus X_{21},\nonumber\\
Z_2={}&X_{11}\oplus X_{14}\oplus X_{16}\oplus X_{17}
       \oplus X_{20}\oplus X_{22}\oplus (X_4X_6\oplus X_4X_7\oplus X_6X_7) ,\nonumber\\
Z_3={}&X_2\oplus X_4\oplus X_5\oplus X_7
       \oplus X_8\oplus X_{23},\nonumber\\
Z_4={}&X_{12}\oplus X_{14}\oplus X_{15}\oplus X_{17}
       \oplus X_{18}\oplus X_{24}\oplus (X_4X_5\oplus X_4X_7\oplus X_5X_7) ,\nonumber\\
Z_5={}&X_3\oplus X_5\oplus X_6\oplus X_7
       \oplus X_9\oplus X_{25},\nonumber\\
Z_6={}&X_{13}\oplus X_{15}\oplus X_{16}\oplus X_{17}
       \oplus X_{19}\oplus X_{26}\oplus (X_5X_6\oplus X_5X_7\oplus X_6X_7) .
\end{align}
Because $Z_2$, $Z_4$, and $Z_6$ contain quadratic Boolean terms, the encoder $\psi:\mathbb F_2^{26}\to\mathbb{F}_2^6$ in~\eqref{eq:nonlinear-example-encoder} is nonlinear. 

Using the broadcast message $M=(Z_1,\ldots,Z_6)$ and its side information, every user can recover its requested source symbol. Hence, the non-linear scheme achieves
\begin{align}
R^{\mathrm{non-lin}}=6 \,\mbox{bits} .
\end{align}
Moreover, the set $\mathcal J=\{1,2,3,11,12,13\}$ forms a maximum acyclic induced subgraph, i.e., none of these six users knows another source indexed by $\mathcal J$, yielding the matching converse (from the MAIS lower bound~\eqref{eq:MAIS_bound})
\begin{align}
R^*(26,26)\geq 6 \,\mbox{bits} .
\end{align}
However, no scalar or vector linear encoder over $\mathbb F_2$ can attain rate $6$ for the side-information pattern in~\eqref{eq:nonlinear-example-interference-sets}, i.e., $R^{*\mathrm{lin}}>6$ bits. 
For blocklength $n$, a linear encoder has the form $M_n=\sum\limits_{j} \mathbf H_j X_j^n$ where $\mathbf H_j\in\mathbb F_2^{R\times n}$ is the matrix multiplying message $X_j^n$. After user $i$ substracts its side information, it observes $\mathbf H_iX_i^n+\sum\limits_{j\in\mathcal B_i} \mathbf H_jX_j^n$. It can recover $X_i^n$ only if the subspace generated by $\mathbf H_i$ contributes $n$ new dimensions beyond the interfering subspace:
\begin{align}
\label{eq:decodability_condition}
\rk(\mathbf H_{\{i\}\cup \mathcal B_i})=\rk(\mathbf H_{\mathcal B_i})+n .
\end{align}
If instead, $\rk(\mathbf H_{\{i\}\cup \mathcal B_i})=\rk (\mathbf H_{\mathcal B_i})$, then $\mathbf H_i$ lies in the span of the interference, and user $i$ cannot decode $X_i^n$. 
If a vector-linear code of rate $6$ existed, its entire broadcast space would have dimension $6n$ for a blocklength $n$, and every encoding matrix $\mathbf H_{j'}$ would have to be expressed as a linear combination of $\mathbf H_j$, $j\in \mathcal J$. 
Every linear combination carried by $(X_4^n,X_{14}^n)$ can already be generated using the encoding subspaces of $(X_5^n,X_{15}^n)$ and $(X_6^n,X_{16}^n)$. Because $\mathrm{col}(\mathbf H_4)\subseteq \mathrm{span}\{\mathbf H_5,\mathbf H_6,\mathbf H_{15},\mathbf H_{16}\}$, the contribution $\mathbf H_4X_4^n$ does not occupy any new dimension beyond interference, i.e., $\rk(\mathbf H_{\{4\}\cup \mathcal B_4})=\rk (\mathbf H_{\mathcal B_4})$. However, decodability requires $\rk(\mathbf H_{\{4\}\cup \mathcal B_4})=\rk(\mathbf H_{\mathcal B_4})+n$. Thus, $X_4^n$ is not decodable. Similarly, users $5$ and $6$ cannot decode their desired functions.
Thus, even though every demanded function $f_i({\bf X}_{[26]})=X_i$ is linear,
\begin{align}
R^*(26,26)=R^{\mathrm{non-lin}}=6<R^{*\mathrm{lin}} .
\end{align}
The strict improvement is entirely due to non-linear broadcast encoding.
\end{proof}

Thus, in this example, non-linear encoding strictly outperforms scalar and vector linear encoding.

\begin{prop}[Zero-error and asymptotically lossless operational broadcast entropies]
\label{prop:nonlinear-graph-separation}
For the computation-broadcast instance described above,
\begin{align}
\omega(G_{\mathrm B})=\chi(G_{\mathrm B})=64,\qquad \alpha(G_{\mathrm B})=2^{20},
\end{align}
and, under the uniform source distribution,
\begin{align}
\label{eq:clique_korner_chrom}
H_\omega(G_{\mathrm B},\mathbf X_{[26]})=H_\kappa(G_{\mathrm B},\mathbf X_{[26]})=H_\chi(G_{\mathrm B},\mathbf X_{[26]})=6 .
\end{align}
Furthermore, the zero-error and asymptotically lossless operational broadcast entropies coincide.
\begin{align}
R_0^*(26,26)
=R^*(26,26)
=6<R^{*\mathrm{lin}}.
\end{align}
\end{prop}

\begin{proof}[Proof of Proposition~\ref{prop:nonlinear-graph-separation}]
Let $G_{\mathrm B}=(\mathcal V(G_{\mathrm B}), \mathcal E(G_{\mathrm B}))$
, with $\mathcal V(G_{\mathrm B})=\mathbb F_2^{26}$ and $|\mathcal V(G_{\mathrm B}|=2^{26}$. Two realizations are adjacent if some user has the same side information under both realizations but requests different values.

The users indexed by $\mathcal J=\{1,2,3,11,12,13\}$ form an independent set in the side-information graph, i.e., none of these users knows another source indexed by $\mathcal J$. 
Fixing the sources $\mathbf X_{\mathcal J^c}$ and allowing $\mathbf X_{\mathcal J}$ to vary freely produces $2^6$ source realizations. Any two distinct realizations differ in at least one coordinate $X_i$, $i\in\mathcal J$, while user $i$'s side information is unchanged. Hence, the corresponding vertices are adjacent in $G_{\mathrm B}$, and $\omega(G_{\mathrm B})\geq 2^6=64$. 
On the other hand, the non-linear construction broadcasts $M=(Z_1,\ldots,Z_6)$ (cf.~\eqref{eq:nonlinear-example-encoder}), and permits zero-error reconstruction at every user. Thus, the message induces a proper coloring $G_{\mathrm B}$ with at most $64$ colors. Hence,
\begin{align}
\label{eq:omega_lower_bound_ex3}
64\leq\omega(G_{\mathrm B})\leq\chi(G_{\mathrm B})\leq64 .
\end{align}

Because adjacency depends only on the difference between two source vectors, $G_{\mathrm B}$ is a Cayley graph and is therefore vertex-transitive. The clique--coclique bound gives
\begin{align}
\label{eq:clique_coclique_ex3}
\alpha(G_{\mathrm B})\omega(G_{\mathrm B})
\leq|\mathcal V(G_{\mathrm B})|=2^{26} .
\end{align}
Hence, from~\eqref{eq:omega_lower_bound_ex3} and~\eqref{eq:clique_coclique_ex3}, $\alpha(G_{\mathrm B})\leq2^{20}$. 
Conversely, the $2^6$-coloring partitions $2^{26}$ vertices into $2^6$ independent color classes. At least one class must contain at least $2^{20}$ vertices, so $\alpha(G_{\mathrm B})\geq 2^{20}$. Hence, $\alpha(G_{\mathrm B})=2^{20}$, and 
\begin{align}
\label{eq:kappa_bound_ex3}
H_\kappa(G_{\mathrm B},\mathbf X_{[26]})=\log_2\frac{|\mathcal V(G_{\mathrm B}|}{\alpha(G_{\mathrm B})}=\log_2\frac{2^{26}}{2^{20}}=6 .
\end{align}
The non-linear six-bit coloring construction in~\eqref{eq:nonlinear-example-encoder} yields 
$\psi_1^{-1}(M)\triangleq \{\mathbf x_{[26]}\in\mathbb F_2^{26} : \phi(\mathbf x_{[26]})=M\}\in\Gamma(G_{\mathrm B})$ such that $|\psi_1^{-1}(M)|=2^{20}=\alpha(G_{\mathrm B})$ for every $M\in \mathbb F_2^6$. Thus, the non-linear encoder partitions $G_{\mathrm B}$ into $64$ maximum independent sets. 
Moreover, under the uniform distribution, every independent color class has the same probability, $\frac{\alpha(G_{\mathrm B})}{2^{26}}=2^{-6}$. 
Thus, the encoder in~\eqref{eq:nonlinear-example-encoder} achieves the optimal  one-shot zero-error chromatic entropy:
\begin{align}
H_\chi(G_{\mathrm B},\mathbf X_{[26]})=6 .
\end{align}
Finally, vertex transitivity and $\omega(G_{\mathrm B})=2^6$ give
\begin{align}
H_\omega(G_{\mathrm B},\mathbf X_{[26]})=\log_2\omega(G_{\mathrm B})=6 .
\end{align}

Repeating the non-linear code componentwise gives a zero-error rate of $6$ at every blocklength. Conversely, after revealing 
$\mathbf X_{\mathcal J^c}$, the message must enable the recovery of 
sources indexed by 
$\mathcal J$. Thus, Fano’s inequality gives $H_2\left(\mathbf X_{\mathcal J}^n\middle|,M_n,\mathbf X_{\mathcal J^c}^n\right)=6n-o(n)$. 
After normalization and passage to the limit, this proves $R^*(26,26)
\geq6$.

From the impossibility result~\cite[Theorem 3]{sharififar2021broadcast} for scalar and vector linear coding in the baseline $26$-user instance  
\begin{equation}
R_0^*(26,26)
=R^*(26,26)
=6<R^{*\mathrm{lin}}.\qedhere
\end{equation}
\end{proof}

\end{ex}

Example~\ref{ex:nonlinear-outperforms-linear} shows that a non-linear vector encoding scheme exploits the fundamental limits of compressibility based on characteristic graphs, which assigns codewords to MISs, and reduces average code length over scalar and vector linear encoding, highlighting the benefit of our framework. 
The same graph-based approach can also benefit NFCB, although deriving order-optimal non-linear encoding mappings remains open. 

We next consider a Boolean computation broadcast example where $G_{\mathrm B}$ is self-complementary, i.e., is isomorphic to its graph complement, denoted as $G_{\mathrm B}\simeq\overline{G}_{\mathrm B}$,~\cite {farrugia1999self}, 
where our scheme in Theorem~\ref{theo-general-achievable-rate-via-broadcast} achieves the optimum.

\begin{ex}[A {\bf \texorpdfstring{$(3,3,\{X_i\}_{i\in[3]},\{f_i\}_{i\in[3]})$}{} Boolean computation broadcast model}]
\label{ex-nonlinear-comp_self-complementary_graph}

Consider a NFCB model $(3,3,\{X_i\}_{i\in[3]},\{f_i\}_{i\in[3]})$, where ${\bf X}_{[3]}\sim\mathrm{Unif}(\mathbb{F}_2^3)$, and the desired functions are 
\begin{align}
\label{eq:function_demands_self_complementary}
f_1=(X_2\oplus X_3,X_2(X_2\oplus X_3)) , \quad
f_2=X_1\oplus X_3, \quad
f_3=X_1\oplus X_2 . 
\end{align}
These functions yield the respective set of edges in the corresponding $G_{f_i}$ for each user $i\in[3]$:
\begin{align}
\mathcal E_1=\{&\{000,001\},\{000,010\},\{001,010\},\{001,011\},\{010,011\},\nonumber\\
&\{100,101\},\{100,110\},\{101,110\},
\{101,111\},\{110,111\}\} , \\
\mathcal E_2=\{&\{000,001\},\{000,100\},\{001,101\},\{100,101\},\{010,011\},\{010,110\},\{011,111\},\{110,111\}\} , \\
\mathcal E_3=\{&\{000,010\},\{000,100\},\{010,110\},\{100,110\},\{001,011\},\{001,101\},\{011,111\},\{101,111\}\} .
\end{align}
Thus, taking their union gives a fourteen-edge broadcast graph $G_{\mathrm B}=(\mathcal{V}(G_{\mathrm B}),\mathcal{E}(G_{\mathrm B}))$, with $\mathcal{V}(G_{\mathrm B})=\mathbb F_2^3$ and $\mathcal{E}(G_{\mathrm B})=\bigcup_{i=1}^3 \mathcal{E}_i$, 
and is shown in Figure~\ref{fig:Graphs_ex4}. 
The following vertex permutation $\pi$ maps $G_{\mathrm B}$ to $\overline{G}_{\mathrm B}$: 
\begin{align}
\begin{array}{c|cccccccc}
\mathbf X_{[3]}
&000&001&010&011&100&101&110&111\\ \hline
\pi(\mathbf X_{[3]})
&001&100&111&010&110&011&000&101 .
\end{array}
\end{align}
Hence, $G_{\mathrm B}\simeq\overline{G}_{\mathrm B}$, implying $\{u,v\}\in\mathcal E(G_{\mathrm B})\,\,\Longleftrightarrow\,\, \{\pi(u),\pi(v)\}\in\mathcal E(\overline{G}_{\mathrm B})$. 
Both $G_{\mathrm B}$ and $\overline{G}_{\mathrm B}$ have $|\mathcal{V}(G_{\mathrm B})|=8$ vertices and $|\mathcal{E}(G_{\mathrm B})|=14$ edges, and the same degree sequence. Furthermore, 
\begin{align}
\label{chromatic_numbers_self_complementary}
\alpha(G_{\mathrm B})=\omega(G_{\mathrm B})=\chi(G_{\mathrm B})=3=\omega(\overline{G}_{\mathrm B})=\chi(\overline{G}_{\mathrm B})=3 \ .
\end{align}

\begin{figure*}[t!]
    \centering
\includegraphics[width=0.5\textwidth]{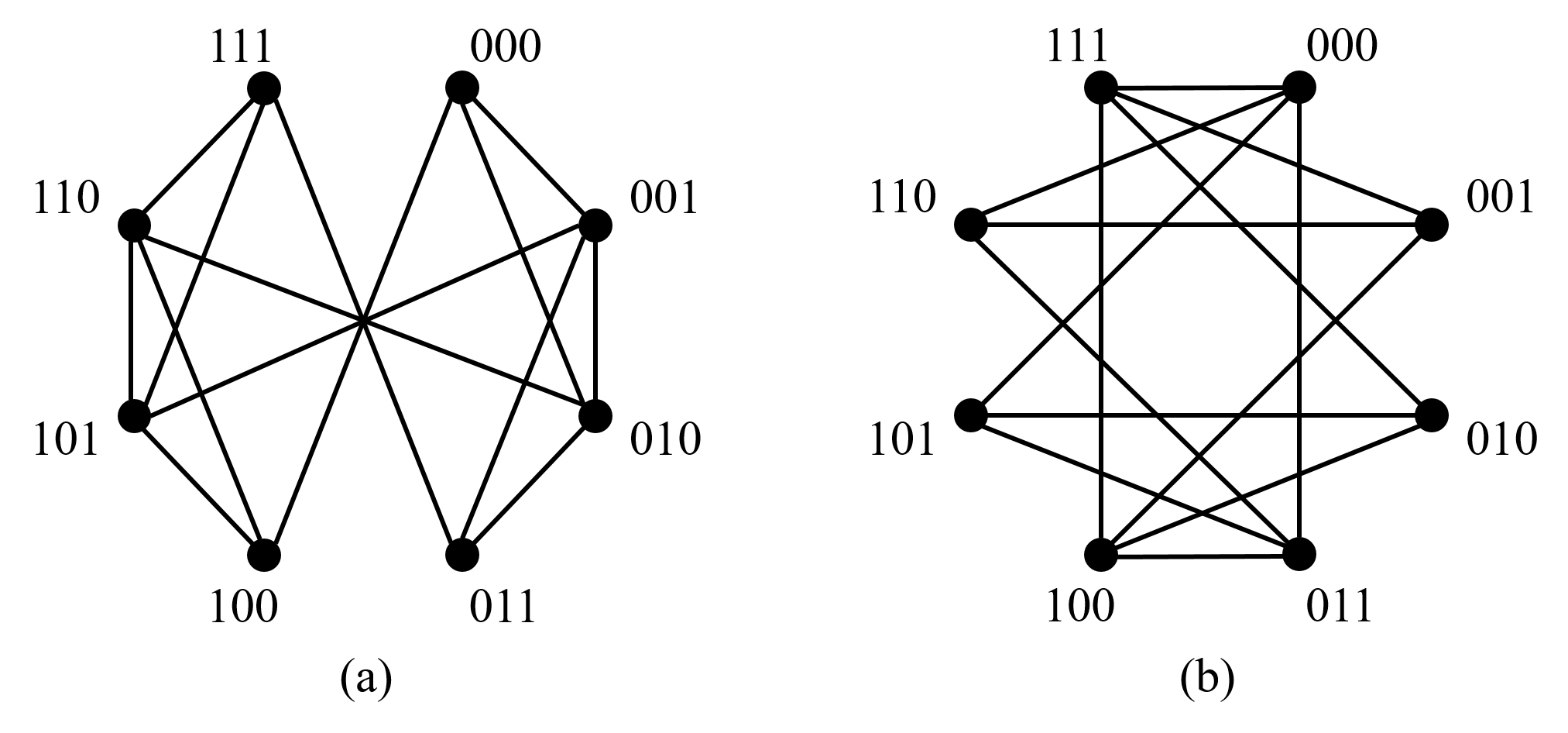}
    \caption{\small (a) A self-complementary broadcast graph $G_{\mathrm B}$ on $8$ vertices and (b) its complement $\overline{G}_{\mathrm B}$.}
    \label{fig:Graphs_ex4}
\end{figure*}

The graph consists of four triangles $\{000,001,010\}$, $\{001,010,011\}$, $\{100,101,110\}$, and $\{101,110,111\}$, which form two overlapping triangle pairs, where the left block is composed of $\{100,101,110\}$, $\{101,110,111\}$, and the right block consists of $\{000,001,010\}$, $\{001,010,011\}$, where each pair shares an edge. 
There are $7$ different $4$-vertex sets that contain a $4$-cycle, of which $5$ are chordless (induced $C_4$) and $2$ are chorded. Specifically, chordless $4$-cycles are $\{000,001,100,101\}$, $\{000,010,100,110\}$, $\{001,010,101,110\}$, $\{001,011,101,111\}$, and $\{010,011,110,111\}$, and chorded $4$-cycles are $\{000,001,010,011\}$, and $\{100,101,110,111\}$. 
Furthermore, $G_{\mathrm B}$ contains neither an induced $C_5$ nor an induced $C_7$. Since $\mathcal V(G_{\mathrm B})=8$, these are the only possible induced odd cycles of length at least five. By self-complementarity, $\overline{G}_{\mathrm B}$ also contains neither an induced $C_5$ nor an induced $C_7$. 
By the Strong Perfect Graph Theorem, $G_{\mathrm B}$ is perfect. 
%
Hence, from~\cite[Lemma~12]{alon1996source} and~\cite{csiszar1990entropy}, 
\begin{align}
\label{eq:clique_entropy_self-complementary_perfect_ex3}
H_\omega(G_{\mathrm B},{\bf X}_{[3]})=H_\kappa(G_{\mathrm B},{\bf X}_{[3]})=1.5 \,\mbox{bits} .
\end{align}

Employing the Slepian-Wolf theorem~\cite{SlepWolf1973}, $R_{\mathrm{SW}}(3,3)=\max\limits_{i\in[3]}\,  H_2({\bf X}_{[3]\backslash i})=2$ bits.

The optimal one-shot zero-error rate is $R_{1,0}^*(3,3)=\log_2 \chi(G_{\mathrm B})=\log_2 3=1.585$ bits.
To determine the entropy of one-shot coloring, let $C=c_{G_{\mathrm B}}(\mathbf X_{[3]})\in\{R,B,G\}$ denote the color assigned to the source realization, where
\begin{align}
c_{G_{\mathrm B}}^{-1}(R)=\{000,011,101\},\quad
c_{G_{\mathrm B}}^{-1}(B)=\{010,100,111\},\quad
c_{G_{\mathrm B}}^{-1}(G)=\{001,110\} .
\end{align}
Under $\mathbf X_{[3]}\sim\mathrm{Unif}(\mathbb F_2^3)$, its distribution is $\bigl(\mathbb{P}(C=R),\mathbb{P}(C=B),\mathbb{P}(C=G)\bigr)=\left(0.375,0.375,0.25\right)$. Thus,
\begin{align}
H_\chi(G_{\mathrm B},{\bf X}_{[3]})=\mathsf H_2((0.375,0.375,0.25))=1.561 \ \mbox{bits} .
\end{align}

For $X_1$, conditioning on $X_1=0$ leaves the four vertices $\{000,001,010,011\}$. Their induced graph is $K_4$ with one edge removed. The same structure appears when $X_1=1$. A minimum-entropy coloring of each induced graph has color-class sizes $2,1,1$. Because the four vertices are conditionally uniform, the color probabilities are $(0.5,0.25,0.25)$, leading to $H_\chi(G_{\mathrm B},\mathbf X_{[3]} \mid X_1)=1.5$ bits. 
For $X_j$, $j=2,3$, conditioning on $X_j=0$ or $X_j=1$ produces an induced $C_4$. Thus, $H_\chi(G_{\mathrm B},\mathbf X_{[3]} \mid X_j)=1$ bit. Hence,~$H_{\chi,\mathrm B}^{\mathrm{SI}}(G_{\mathrm B},\mathbf X_{[3]}\mid \mathcal S_{[3]}=\mathbf X_{[3]})
=1.5$ bits.

Given the functions in~\eqref{eq:function_demands_self_complementary}, for $f_1$, we have $\mathbb{P}(f_1=(0,0))=0.5$, and $\mathbb{P}(f_1=(1,0))=\mathbb{P}(f_1=(1,1))=0.25$. Thus, $H_2(f_1)=H_2(f_1\mid X_1)=\mathsf H_2((0.5,0.25,0.25))=1.5$ bits because $f_1$ is independent of $X_1$. 
For any zero-error blocklength-$n$ code, user one must recover $f_1^n$ from $(M_n,X_1^n)$. Thus, $H_2(f_1^n\mid M_n,X_1^n)=0$. Therefore,
\begin{align}
H_2(M_n)\geq H_2(M_n\mid X_1^n)\geq I_2(f_1^n;M_n\mid X_1^n)=H_2(f_1^n\mid X_1^n)=1.5n \ .
\end{align}
Thus, $R_0^*(3,3)
\geq 1.5$ bits. In other words, the K\"orner-entropy upper bound and the single-user converse coincide.

For asymptotically lossless reconstruction, employing Fano's inequality
\begin{align}
H_2(f_1^n\mid M_n,X_1^n)=o(n) ,
\end{align}
and therefore, $H_2(M_n)\geq 1.5n-o(n)$, yielding $R^*(3,3)
\geq 1.5$ bits.

Thus, for the given perfect graph setting,  the operational quantities are optimal:
\begin{multline}
\label{eq:rate_comparison_ex3}
R^*(3,3)=R_0^*(3,3)
= H_\kappa(G_{\mathrm B},\mathbf X_{[3]})=H_{\chi,\mathrm B}^{\mathrm{SI}}(G_{\mathrm B},\mathbf X_{[3]}\mid \mathcal S_{[3]}=\mathbf X_{[3]})
=1.5\\
\leq  H_\chi(G_{\mathrm B},{\bf X}_{[3]})=1.561 \leq R_{\mathrm{SW}}(3,3)=2 \ .
\end{multline}


In this example, where $G_{\mathrm B}$ is self-complementary and perfect, our proposed technique achieves the optimum.

\end{ex}

\section{Conclusion}
\label{sec:conc}

In this paper, we studied the general $K$-user NFCB problem for $K\geq 2$, with arbitrary data distributions and function demands over finite fields. On the achievability side, we proposed a new broadcast coding scheme based on K\"orner’s characteristic graph approach, where the broadcast message is designed as a function of each user’s desired function and side information (Theorem~\ref{theo-general-achievable-rate-via-broadcast}). Compared with the state of the art, the scheme reduces communication rate while accommodating both linear and non-linear demands. On the converse side, by carefully capturing each user’s demand and side information, Theorem~\ref{theo-lower-bound-mutual-info} gives tighter converses than a simple side-information-based converse induced by $({\bf X}_{[N]}, \{\mathcal{S}_i\}_{i\in[K]},\{f_i\}_{i\in[K]})$. Compared with the achievable rate in Theorem~\ref{theo-general-achievable-rate-via-broadcast}, which optimizes the coloring distribution over MISs of $\GU$, the converse bound in Theorem~\ref{theo-lower-bound-mutual-info} relies on a carefully designed variable to approximate the fundamental rate limits induced by $({\bf X}_{[N]}, \{\mathcal{S}_i\}_{i\in[K]},\{f_i\}_{i\in[K]})$. 
We also detailed two clique entropy-based converses (Theorem~\ref{thm:zero_error_clique_converse} and~\ref{thm:vanishing_error_clique_converse}) using the length-$n$ operational block broadcast graph $G_{\mathrm B}^{(n)}$. 
We also upper-bounded the gap between our achievability results and the converse in Proposition~\ref{prop-genie}, showing tightness when the broadcast graph satisfies certain structural properties and optimality when the desired functions are compatible. Further work is needed to sharpen these bounds and quantify the fundamental rate limits. 

Future directions include optimizing the placement of side information at users to reduce communication complexity, especially when side information is a function of the sender’s datasets and users have heterogeneous storage capacities. Another direction is to identify classes of NFCB problems for which the gap is provably small. Security and privacy are also promising directions: one may design a broadcast message that allows each user to recover only its requested function, without learning other users’ requests, while keeping communication and computation overhead low. Our tools are also relevant to practical non-asymptotic compression settings, where the limits depend on the entropy of the characteristic graph for finite-length sources, and to lossy settings, where broadcast graph edges are defined relative to a distortion criterion. These quantities are generally not known in closed form.

\begin{spacing}{1}
\bibliographystyle{IEEEtran}
\bibliography{ref}

@STRING{tcom= {IEEE Trans. Commun.}}

@STRING{jsac = {IEEE J. Sel. Areas Commun.}}

@STRING{tit = {IEEE Trans. Inf. Theory}}

@STRING{allerton ={Proc., Allerton Conf. Commun., Control, and Comput.}}

@STRING{asilomar = {Proc., Asilomar Conf. Signals, Systems, and Computers}}

@STRING{gcw = {Proc., IEEE Global Commun. Conf. (Globecom) Wkshps}}

@STRING{isit = {Proc., IEEE Int. Symp. Inf. Theory (ISIT)}}

@STRING{itw = {Proc., Inf. Theory Wksh. (ITW)}}

@STRING{spawc = {Proc., IEEE Int. Wksh. Signal Proces. Advances in Wireless Commun.}}

@STRING{preprint = {Preprint}}

@article{sun2019capacity,
  title={On the capacity of computation broadcast},
  author={Sun, Hua and Jafar, Syed Ali},
  journal=tit,
  volume={66},
  number={6},
  pages={3417--3434},
  year={2019},
  month={Dec.}
}

@article{yao2024capacity,
  title={The capacity of 3 user linear computation broadcast},
  author={Yao, Yuhang and Jafar, Syed A},
  journal=tit,
  year={2024},
  month={Apr.},
}

@article{yao2024generic,
  title={On the generic capacity of {K}-user symmetric linear computation broadcast},
  author={Yao, Yuhang and Jafar, Syed A},
  journal=tit,
  year={2024},
  month={Mar.},
  publisher={IEEE}
}

@article{wan2021optimal,
  title={On the optimal load-memory tradeoff of cache-aided scalar linear function retrieval},
  author={Wan, Kai and others},
  journal=tit,
  volume={67},
  number={6},
  pages={4001--4018},
  year={2021},
  month={Mar.}
}

@inproceedings{ravi2016broadcast,
  title={Broadcast Function Computation with Complementary Side Information},
  author={Ravi, Jithin and Dey, Bikash Kumar},
  booktitle=gcw,
  pages={1--6},
  year={2016},
  month={Dec.},
  address={Washington, DC, USA}
}

@article{khalesi2025hyper,
  title={Hyper-Minrank: A Unified Hypergraph Characterization of Multi-Sender Index Coding},
  author={Khalesi, Ali and Elia, Petros},
  journal={arXiv preprint arXiv:2512.13615},
  year={2025},
  month={Dec.}
}

@inproceedings{korner1971coding,
  title={Coding of an information source having ambiguous alphabet and the entropy of graphs},
  author={K{\"o}rner, J{\'a}nos},
  booktitle={Proc., Prague Conf. Inf. Theory},
  pages={411--425},
  year={1971},
  month={Sep.},
  address={Prague, Czech}
}

@book{el2011network,
  title={Network Information Theory},
  author={El Gamal, Abbas and Kim, Young-Han},
  year={2011},
  month={Dec.},
  publisher={Cambridge University Press}
}

@Article{OrlRoc2001,
  author  = {A. Orlitsky and J. R. Roche},
  title   = {Coding for computing},
  journal = tit,
  year    = {2001},
  volume  = {47},
  number  = {3},
  month   = {Mar.},
  pages   = {903--17},
}

@Article{AlonOrlit1996,
  author  = {N. Alon and A. Orlitsky},
  title   = {Source coding and graph entropies},
  journal = tit,
  year    = {1996},
  volume  = {42},
  number  = {5},
  month   = {Sep.},
  pages   = {1329--39},
}

@article{feizi2014network,
  title={On network functional compression},
  author={Feizi, Soheil and M{\'e}dard, Muriel},
  journal=tit,
  volume={60},
  number={9},
  pages={5387--5401},
  year={2014},
  month={Jun.}
}

@article{alon1996source,
  title={Source coding and graph entropies},
  author={Alon, Noga and Orlitsky, Alon},
  journal={IEEE Trans. Inf. Theory},
  volume={42},
  number={5},
  pages={1329--1339},
  year={1996},
  month={Aug.}
}

@article{elkordy2021compressed,
  title={Compressed coded distributed computing},
  author={Elkordy, Ahmed and others},
  journal=tcom,
  volume={69},
  number={5},
  pages={2773--2783},
  year={2021},
  month={Jan.},
  publisher={IEEE}
}

@article{tian2022fedbroadcast,
  title={Fedbroadcast: Exploit broadcast channel for fast convergence in wireless federated learning},
  author={Tian, Hao and others},
  journal={IEEE Internet Things J.},
  volume={10},
  number={5},
  pages={4652--4663},
  year={2022},
  month={Nov.},
  publisher={IEEE}
}

@article{salehkalaibar2021distributed,
  title={Distributed hypothesis testing over noisy broadcast channels},
  author={Salehkalaibar, Sadaf and Wigger, Mich{\`e}le},
  journal={Information},
  volume={12},
  number={7},
  pages={268},
  year={2021},
  month={Jun.},
  publisher={MDPI}
}

@inproceedings{salehi2023achievable,
  title={An achievable low complexity encoding scheme for coloring cyclic graphs},
  author={Deylam Salehi, Mohammad Reza and Malak, Derya},
  booktitle=allerton,
  pages={1--8},
  year={2023},
  month={Sep.},
  address={Monticello, IL, USA}
}

@article{korner1986fredman,
  title={Fredman--{K}oml{\'o}s bounds and information theory},
  author={K{\"o}rner, J{\'e}nos},
  journal={J. SIAM on Alg. Disc. Methods},
  volume={7},
  number={4},
  pages={560--570},
  year={1986},
  month={Oct.}
}

@inproceedings{zaharia2010spark,
  title={Spark: Cluster computing with working sets},
  author={Zaharia, Matei and others},
  booktitle={Proc., USENIX HotTop Cloud Comp. Wksh.},
  year={2010},
  month={Jun.},
  address={Boston, MA, USA}
}

@article{dean2008mapreduce,
  title={{MapReduce:} {Simplified} data processing on large clusters},
  author={Dean, Jeffrey and Ghemawat, Sanjay},
  journal={J. Commun. ACM},
  volume={51},
  number={1},
  pages={107--113},
  year={2008},
  month={Jan.}
}

@article{csahin2023survey,
  title={{A Survey on Over-the-Air Computation}},
  author={{\c{S}}ahin, Alphan and Yang, Rui},
  journal={IEEE Comm. Surveys \& Tutorials},
  volume={25},
  number={3},
  pages={1877--1908},
  year={2023},
  month={Apr.}
}

@article{montalban2021broadcast,
  title={{Broadcast Core-Network: Converging Broadcasting With the Connected World}},
  author={Montalban, Jon and others},
  journal={IEEE Trans. on Broadcast.},
  volume={67},
  number={3},
  pages={558--569},
  year={2021},
  month={Aug.}
}

@article{malak2024multis,
  title={{Multi-Server Multi-Function Distributed Computation}},
  author={Malak, Derya and others},
  journal={Entropy},
  volume={26},
  number={6},
  pages={448},
  year={2024},
  month={Jun.},
  publisher={MDPI}
}

@article{dutta2019optimal,
  title={On the optimal recovery threshold of coded matrix multiplication},
  author={Dutta, Sanghamitra and others},
  journal=tit,
  volume={66},
  number={1},
  pages={278--301},
  year={2019},
  month={Jul.}
}

@article{jia2021capacity,
  title={On the capacity of secure distributed batch matrix multiplication},
  author={Jia, Zhuqing and Jafar, Syed Ali},
  journal=tit,
  volume={67},
  number={11},
  pages={20--37},
  year={2021},
  month={Sep.}
}

@inproceedings{tandon2017gradient,
  title={{Gradient Coding: Avoiding Stragglers in Distributed Learning}},
  author={Tandon, Rashish and others},
  booktitle ={Proc., Int. Conf. Machine Learn.},
  year={2017},
  month={Jul.},
  address={Sydney, Australia},
  publisher ={PMLR},
}

@article{sun2018capacity,
  title={The capacity of private computation},
  author={Sun, Hua and Jafar, Syed Ali},
  journal=tit,
  volume={65},
  number={6},
  pages={3880--3897},
  year={2019},
  month={Jun.}
}

@inproceedings{el2010coding,
  title={On coding for cooperative data exchange},
  author={El Rouayheb, Salim and Sprintson, Alex and Sadeghi, Parastoo},
  booktitle=ITW,
  pages={1--5},
  year={2010},
  month={Jan.},
  address={Cairo, Egypt}
}

@inproceedings{sharififar2021broadcast,
  title={Broadcast rate requires nonlinear coding in a unicast index coding instance of size 36},
  author={Sharififar, Arman and Sadeghi, Parastoo and Aboutorab, Neda},
  booktitle={Proc., IEEE Int. Symp. Inf. Theory},
  pages={208--213},
  year={2021},
  month={Jul.}
}

@inproceedings{MaddahAli2013,
  author    = {M. A. Maddah-Ali and U. Niesen},
  title     = {Fundamental limits of caching},
  booktitle = {Proc., IEEE Int. Symp. Inf. Theory},
  year      = {2013},
  month     = {July},
  pages     = {1077-1081},
  address   = {Istanbul, Türkiye},
}

@article{khalesi2022multi,
  title={{Multi-User Linearly-Separable Distributed Computing}},
  author={Khalesi, Ali and Elia, Petros},
  journal=tit,
  year={2023},
  month={Jun.},
  volume={69},
  number={10},
  pages={6314-39}
}

@Article{NadMadAve2016,
  author  = {N. Naderializadeh and M. A. Maddah-Ali and A. S. Avestimehr},
  title   = {Fundamental limits of cache-aided interference management},
  journal = {IEEE Trans. Inf. Theory},
  year    = {2017},
  volume  = {63},
  number  = {5},
  month   = {May},
  pages   = {3092-107},
}

@Article{BirKol2006,
  author  = {Y. Birk and T. Kol},
  title   = {Coding on demand by an informed source {(ISCOD)} for efficient broadcast of different supplemental data to caching clients},
  journal = {IEEE/ACM Trans. Netw.},
  year    = {2006},
  volume  = {14},
  number  = {SI},
  month   = {Jun.},
  pages   = {2825 - 2830},
}

@Article{BarBirJayKol2011,
  author  = {Z. Bar-Yossef and Y. Birk and T. S. Jayram and T. Kol},
  title   = {Index coding with side information},
  journal = {IEEE Trans. Inf. Theory},
  year    = {2011},
  volume  = {57},
  number  = {3},
  month   = {Mar.},
  pages   = {1479 - 1494},
}

@article{farrugia1999self,
  title={Self-complementary graphs and generalisations: a comprehensive reference manual},
  author={Farrugia, Alastair},
  journal={University of Malta},
  year={1999},
  month={Aug.}
}

@article{berge1961farbung,
  title={Farbung von Graphen, deren samtliche bzw. deren ungerade Kreise starr sind},
  author={Berge, Claude},
  journal={Wissenschaftliche Zeitschrift},
  year={1961},
  publisher={Martin Luther Universitat Halle-Wittenberg}
}

@Book{CovTho1991,
  author  = {T. M. Cover and J. A. Thomas},
  title   = {Elements of Information Theory},
  year    = {1991},
  month   ={Dec.},
  address = {New York, NY: USA},
  publisher={John Wiley \& Sons},
  edition={2}
}

@book{godsil2001algebraic,
  title={Algebraic graph theory},
  author={Godsil, Christopher David and Royle, Gordon and Godsil, CD},
  volume={207},
  year={2001},
  publisher={Springer New York}
}

@book{magnus2004combinatorial,
  title={Combinatorial group theory: Presentations of groups in terms of generators and relations},
  author={Magnus, Wilhelm and Karrass, Abraham and Solitar, Donald},
  year={2004},
  publisher={Courier Corporation}
}

@article{conlon2024clique,
  title={On the clique number of random Cayley graphs and related topics},
  author={Conlon, David and Fox, Jacob and Pham, Huy Tuan and Yepremyan, Liana},
  journal={arXiv preprint arXiv:2412.21194},
  year={2024},
  month={Dec.}
}

@article{hoang1987conjecture,
  title={On a conjecture of Meyniel},
  author={Ho{\`a}ng, Ch{\'\i}nh T},
  journal={J. Comb. Theory, Ser. B},
  volume={42},
  number={3},
  pages={302--312},
  year={1987},
  publisher={Elsevier}
}

@article{alon2020minrank,
  title={The minrank of random graphs over arbitrary fields},
  author={Alon, Noga and Balla, Igor and Gishboliner, Lior and Mond, Adva and Mousset, Frank},
  journal={Israel Journal of Mathematics},
  volume={235},
  number={1},
  pages={63--77},
  year={2020},
  month={Jan.},
  publisher={Springer}
}

@article{csiszar1990entropy,
  title={Entropy splitting for antiblocking corners and perfect graphs},
  author={Csisz{\'a}r, Imre and K{\"o}rner, J{\'a}nos and Lov{\'a}sz, L{\'a}szl{\'o} and Marton, Katalin and Simonyi, G{\'a}bor},
  journal={Combinatorica},
  volume={10},
  number={1},
  pages={27--40},
  year={1990},
  month={Mar.},
  publisher={Springer}
}

@Article{SlepWolf1973,
  author  = {D. Slepian and J. K. Wolf},
  title   = {Noiseless coding of correlated information sources},
  journal = tit,
  year    = {1973},
  volume  = {19},
  number  = {4},
  month   = {Jul.},
  pages   = {471--480},
}

@Article{CarFloJor2008,
  author  = {J. Cardinal and S. Fiorini and G. Joret},
  title   = {Tight results on minimum entropy set cover},
  journal = {Algorithmica},
  year    = {2008},
  volume  = {51},
  number  = {1},
  month   = {May},
  pages   = {49--60},
}

@article{li2017fundamental,
  title={A fundamental tradeoff between computation and communication in distributed computing},
  author={Li, Songze and Maddah-Ali, Mohammad Ali and Yu, Qian and Avestimehr, A Salman},
  journal={IEEE Trans. Inf. Theory},
  volume={64},
  number={1},
  pages={109--128},
  year={2017},
  month={Sep.}
}

@article{halpern1991clock,
  title={Clock synchronization and the power of broadcasting},
  author={Halpern, Joseph Y and Suzuki, Ichiro},
  journal={Distributed Computing},
  volume={5},
  number={2},
  pages={73--82},
  year={1991},
  month={Sep.},
  publisher={Springer}
}

@article{reimers1998digital,
  title={Digital video broadcasting},
  author={Reimers, Ulrich},
  journal={IEEE Commun. Mag.},
  volume={36},
  number={6},
  pages={104--110},
  year={1998},
  month={Jun.},
  publisher={IEEE}
}

@inproceedings{li2016coded,
  title={Coded distributed computing: Fundamental limits and practical challenges},
  author={Li, Songze and others},
  booktitle={Proc., Asilomar Conf. Sig. Sys.  Comp. (ACSSC)},
  pages={509--513},
  year={2016},
  month={Nov.},
  address={Pacific Grove, CA, USA}
}

@ARTICLE{yang2023federated,
  author={Yang, Huiyuan and Ding, Tian and Yuan, Xiaojun},
  journal=tcom, 
  title={{Federated Learning With Lossy Distributed Source Coding: Analysis and Optimization}}, 
  year={2023},
  month={May},
  volume={71},
  number={8},
  pages={4561-4576}
}

@inproceedings{kate2010distributed,
  title={Distributed private-key generators for identity-based cryptography},
  author={Kate, Aniket and Goldberg, Ian},
  booktitle={Proc., Springer Int. Conf. Security and Crypto. for Networks},
  pages={436--453},
  year={2010},
  month={Sep.},
  address={Amalfi, Italy}
}

@article{gamal1982achievable,
  title={Achievable rates for multiple descriptions},
  author={Gamal, AE and Cover, Thomas},
  journal=tit,
  volume={28},
  number={6},
  pages={851--857},
  year={1982},
  month={Nov.}
}

@article{arbabjolfaei2018fundamentals,
  title={Fundamentals of index coding},
  author={Arbabjolfaei, Fatemeh and others},
  journal={Found. and Trends. Com. and Info. Theo.},
  volume={14},
  number={3-4},
  pages={163--346},
  year={2018},
  month={Oct.},
  publisher={Now Publishers, Inc.}
}

@inproceedings{MaTunin2025,
  author={Ma, Yinbin and Tuninetti, Daniela},
  booktitle=isit, 
  title={An Achievable Scheme for the K-User Linear Computation Broadcast Channel}, 
  year={2025},
  month={Jun.},
  pages={1-6},
  address={Ann Arbor, MI, USA}
}

@article{yu2020straggler,
  title={{Straggler mitigation in distributed matrix multiplication: Fundamental limits and optimal coding}},
  author={Yu, Qian and others},
  journal=tit,
  volume={66},
  number={3},
  pages={1920--1933},
  year={2020},
  month={Jan.}
}

@article{lee2017speeding,
  title={Speeding up distributed machine learning using codes},
  author={Lee, Kangwook and others},
  journal=tit,
  volume={64},
  number={3},
  pages={1514--1529},
  year={2017},
  month={Aug.}
}

@misc{cryptoeprint:2024/154,
      author = {Aurélien Dupin and Simon Abelard},
      title = {{Broadcast Encryption using Sum-Product decomposition of Boolean functions}},
      howpublished = {Cryptology {ePrint} Archive, Paper 2024/154},
      year = {2024}, 
      month={Jul.}
}

@inproceedings{malak2023distributed,
  title={Distributed computing of functions of structured sources with helper side information},
  author={Malak, Derya},
  booktitle=spawc,
  pages={276--280},
  year={2023},
  month={Sep.},
  address={Shanghai, China}
}

@article{jafar2006capacity,
  title={Capacity with causal and noncausal side information: A unified view},
  author={Jafar, Syed},
  journal=tit,
  volume={52},
  number={12},
  pages={5468--5474},
  year={2006},
  month={Nov.}
}

@article{LeeDimaindex,
  author={Lee, Namyoon and Dimakis, Alexandros G. and Heath, Robert W.},
  journal={IEEE Commun. Letter}, 
  title={{Index Coding With Coded Side-Information}}, 
  year={2015},
  month={Mar.},
  volume={19},
  number={3},
  pages={319-322},
}

@article{bar2011index,
  title={Index coding with side information},
  author={Bar-Yossef, Ziv and Birk, Yitzhak and Jayram, TS and Kol, Tomer},
  journal=tit,
  volume={57},
  number={3},
  pages={1479--1494},
  year={2011},
  month={Feb.},
  publisher={IEEE}
}

@article{tahmasbi2014critical,
  title={Critical graphs in index coding},
  author={Tahmasbi, Mehrdad and Shahrasbi, Amirbehshad and Gohari, Amin},
  journal=jsac,
  volume={33},
  number={2},
  pages={225--235},
  year={2014},
  month={Feb.},
  publisher={IEEE}
}

@article{li2018multi,
  title={Multi-sender index coding for collaborative broadcasting: A rank-minimization approach},
  author={Li, Min and Ong, Lawrence and Johnson, Sarah J},
  journal=tit,
  volume={67},
  number={2},
  pages={1452--1466},
  year={2018},
  month={Oct.},
  publisher={IEEE}
}

@article{el2010index,
  title={On the index coding problem and its relation to network coding and matroid theory},
  author={El Rouayheb, Salim and Sprintson, Alex and Georghiades, Costas},
  journal=tit,
  volume={56},
  number={7},
  pages={3187--3195},
  year={2010},
  month={Jun.},
  publisher={IEEE}
}

@article{deylam2025graph,
  title={{Graph-Theoretic Limits of Distributed Computation: Entropy, Eigenvalues, and Chromatic Numbers}},
  author={Deylam Salehi, Mohammad Reza and Malak, Derya},
  journal={Entropy},
  volume={27},
  number={7},
  pages={757},
  year={2025},
  month={Jul.},
  publisher={MDPI}
}

@inproceedings{DeyKizMal2025nonlinear,
  title={{Non-Linear Function Computation Broadcast}},
  author={Deylam Salehi, Mohammad Reza and Kizhakke Purakkal, Vijith Kumar  and Malak, Derya},
  booktitle=isit,
  year={2025},
  month={Jun.},
  address={Ann Arbor, MI, USA}
}

@article{gupta2016relation,
  title={A relation between network computation and functional index coding problems},
  author={Gupta, Anindya and Rajan, B Sundar},
  journal=tcom,
  volume={65},
  number={2},
  pages={705--714},
  year={2016},
  month={Dec},
  publisher={IEEE}
}

@article{tsukiyama1977new,
  title={A new algorithm for generating all the maximal independent sets},
  author={Tsukiyama, Shuji and Ide, Mikio and Ariyoshi, Hiromu and Shirakawa, Isao},
  journal={SIAM J. on Computing},
  volume={6},
  number={3},
  pages={505--517},
  year={1977},
  month={Sep.},
  publisher={SIAM}
}

@article{namboodiri2026fundamental,
  title={Fundamental Limits of Linearly Separable Distributed Computing without File Partitioning},
  author={Namboodiri, KK Krishnan and Peter, Elizabath and Malak, Derya and Elia, Petros},
  journal={IEEE Trans. Inf. Theory},
  volume={72},
  number={9},
  pages={6679-6701},
  year={2026},
  month={Sep.}
}

@article{bidokhti2018noisy,
  title={Noisy broadcast networks with receiver caching},
  author={Bidokhti, Shirin Saeedi and Wigger, Michele and Timo, Roy},
  journal=tit,
  volume={64},
  number={11},
  pages={6996--7016},
  year={2018},
  month={May}
}

@article{yao2022capacity,
  title={Capacity of 2-user linear computation broadcast with partially informed server},
  author={Yao, Yuhang and Jafar, Syed A},
  journal={IEEE Commun. Letters},
  volume={26},
  number={9},
  pages={2014--2017},
  year={2022},
  month={Jul.}
}

@book{karp2009reducibility,
  title={Reducibility among combinatorial problems},
  author={Karp, Richard M},
  booktitle={50 Years of Integer Programming 1958-2008: from the Early Years to the State-of-the-Art},
  pages={219--241},
  year={2009},
  month={Nov.},
  publisher={Springer}
}
\end{spacing}

\end{document}